\documentclass[a4paper,UKenglish,cleveref, autoref, thm-restate]{lipics-v2021}
\pdfoutput=1

\title{Current Algorithms for Detecting Subgraphs of Bounded Treewidth are Probably Optimal}

\titlerunning{Detecting Subgraphs of Bounded Treewidth}
\author{Karl Bringmann}{Saarland University, Saarland Informatics Campus, Germany \and Max Planck Institute for Informatics, Saarland Informatics Campus, Germany}{bringmann@cs.uni-saarland.de}{}{This work is part of the project TIPEA that has received funding from the European Research Council (ERC) under the European Unions Horizon 2020 research and innovation programme (grant agreement No.\ 850979).}
\author{Jasper Slusallek}{Saarland University, Saarland Informatics Campus, Germany}{s8jaslus@stud.uni-saarland.de}{}{}
\authorrunning{K.\,Bringmann and J.\,Slusallek}

\Copyright{Karl Bringmann and Jasper Slusallek}

\ccsdesc[100]{Theory of computation~Design and analysis of algorithms}
\ccsdesc[100]{Theory of computation~Computational complexity and cryptography}

\keywords{subgraph isomorphism, treewidth, fine-grained complexity, hyperclique}

\category{}

\relatedversion{}

\supplement{}

\funding{}

\nolinenumbers 

\hideLIPIcs  

\EventEditors{John Q. Open and Joan R. Access}
\EventNoEds{2}
\EventLongTitle{42nd Conference on Very Important Topics (CVIT 2016)}
\EventShortTitle{CVIT 2016}
\EventAcronym{CVIT}
\EventYear{2016}
\EventDate{December 24--27, 2016}
\EventLocation{Little Whinging, United Kingdom}
\EventLogo{}
\SeriesVolume{42}
\ArticleNo{23}

\DeclareMathOperator{\ParSol}{ParSol}
\DeclareMathOperator{\ParSolE}{ParSolE}
\DeclareMathOperator{\ValConf}{ValConf}
\DeclareMathOperator{\MM}{MM}
\DeclareMathOperator{\tw}{tw}
\DeclareMathOperator{\pw}{pw}
\DeclareMathOperator{\poly}{poly}
\DeclareMathOperator{\polylog}{polylog}
\DeclareMathOperator{\parent}{par}
\DeclareMathOperator{\children}{children}
\DeclareMathOperator{\width}{width}
\DeclareMathOperator{\help}{help}
\DeclareMathOperator{\MP}{MP}
\DeclareMathOperator{\TWL}{TWL}

\def\hidetodo{1}     

\def\islongversion{1}  

\newcommand{\shortonly}[1]{\ifx\islongversion\undefined {#1}\fi}
\newcommand{\longonly}[1]{\ifx\isshortversion\undefined {#1}\fi}
\newcommand{\todo}[1]{\ifx\hidetodo\undefined {\color{red}[TODO: #1]} \fi}
\newcommand{\karl}[1]{\ifx\hidetodo\undefined {\color{orange}[Karl: #1]} \fi}
\usepackage{dutchcal} 
\usepackage{rotating} 
\usepackage{tikz}

\usetikzlibrary{chains,decorations.pathreplacing}

\usetikzlibrary{backgrounds}
\usetikzlibrary{arrows}
\usetikzlibrary{shapes,shapes.geometric,shapes.misc}

\tikzstyle{tikzfig}=[baseline=-0.25em,scale=0.5]

\pgfkeys{/tikz/tikzit fill/.initial=0}
\pgfkeys{/tikz/tikzit draw/.initial=0}
\pgfkeys{/tikz/tikzit shape/.initial=0}
\pgfkeys{/tikz/tikzit category/.initial=0}

\pgfdeclarelayer{edgelayer}
\pgfdeclarelayer{nodelayer}
\pgfsetlayers{background,edgelayer,nodelayer,main}

\tikzstyle{none}=[inner sep=0mm]

\newcommand{\tikzfig}[1]{%
{\tikzstyle{every picture}=[tikzfig]
\IfFileExists{#1.tikz}
  {\input{#1.tikz}}
  {%
    \IfFileExists{./figures/#1.tikz}
      {\input{./figures/#1.tikz}}
      {\tikz[baseline=-0.5em]{\node[draw=red,font=\color{red},fill=red!10!white] {\textit{#1}};}}%
  }}%
}

\tikzstyle{every loop}=[]

\tikzstyle{small-node}=[fill=white, draw=black, shape=circle]
\tikzstyle{k-wise-node}=[fill=white, draw=black, shape=circle, scale=1.9]
\tikzstyle{k-wise-border}=[fill=white, draw=black, shape=circle, scale=1.55]
\tikzstyle{text node}=[text width=4cm, scale=0.7]

\tikzstyle{tree-edge}=[-]
\tikzstyle{area-border}=[-, draw=red, thick]
\tikzstyle{new edge style 0}=[->, draw={black!30!green}, thick]
\tikzstyle{area-pointer}=[->, draw=red, thick]
\tikzstyle{new edge style 1}=[-, draw=gray, dashed]

\tikzstyle{test}=[fill={rgb,255: red,128; green,0; blue,128}, draw=black, shape=circle, minimum size=100]
\tikzstyle{H-node}=[fill={rgb,255: red,191; green,0; blue,64}, draw={rgb,255: red,191; green,0; blue,64}, shape=circle, scale=0.8, very thick]
\tikzstyle{Hollow H}=[draw=blue, shape=circle, minimum size=50pt, very thick]
\tikzstyle{G-node}=[fill=blue, draw=blue, shape=circle, scale=0.7]
\tikzstyle{big invisible}=[scale=2pt]

\tikzstyle{H-edge}=[-, draw={rgb,255: red,191; green,0; blue,64}, line width=1.5pt]
\tikzstyle{dashedline}=[-, dashed]
\tikzstyle{G-edge}=[-, draw=blue, line width=.7pt]
\tikzstyle{pointed}=[->, thick]
\tikzstyle{thick edge}=[-, thick]
\tikzstyle{implies}=[{|->}, ultra thick]

\tikzstyle{center node}=[draw=black, shape=circle, minimum size=100pt, thick, dotted]
\tikzstyle{Rnodes}=[draw=red, shape=circle, fill=red, scale=0.6]
\tikzstyle{Rnodes 2}=[fill={rgb,255: red,255; green,179; blue,0}, draw={rgb,255: red,255; green,179; blue,0}, shape=circle, scale=0.6]
\tikzstyle{text node}=[text width=4.5cm, scale=0.7]
\tikzstyle{Hollow Hnode}=[draw=blue, shape=circle, minimum size=50pt]
\tikzstyle{Gnode}=[fill={rgb,255: red,85; green,0; blue,255}, draw={rgb,255: red,85; green,0; blue,255}, shape=circle, scale=0.5]
\tikzstyle{Gnode R}=[fill=cyan, draw=cyan, shape=circle, scale=0.5]
\tikzstyle{light hollow Hnode}=[draw=cyan, shape=circle, minimum size=50pt]

\tikzstyle{R-edge}=[-, draw={rgb,255: red,255; green,179; blue,0}, very thick]
\tikzstyle{Gedge}=[-, draw=cyan, thick]
\tikzstyle{description}=[->, densely dotted]
\tikzstyle{dashedline2}=[-, dashed]
\tikzstyle{Cedge}=[-, draw=red, very thick]
\tikzstyle{dashed Redge}=[-, draw={rgb,255: red,255; green,179; blue,0}, dashed]

\begin{document}

\maketitle


\begin{abstract}

  The Subgraph Isomorphism problem is of considerable importance in computer science. We examine the problem when the pattern graph $H$ is of bounded treewidth, as occurs in a variety of applications. This problem has a well-known algorithm via color-coding that runs in time $O(n^{\tw(H)+1})$ [Alon, Yuster, Zwick'95], where $n$ is the number of vertices of the host graph $G$. While there are pattern graphs known for which Subgraph Isomorphism can be solved in an improved running time of $O(n^{\tw(H)+1-\varepsilon})$ or even faster (e.g. for $k$-cliques), it is not known whether such improvements are possible for all patterns. The only known lower bound rules out time $n^{o(\tw(H) / \log(\tw(H)))}$ for any class of patterns of unbounded treewidth assuming the Exponential Time Hypothesis [Marx'07].

In this paper, we demonstrate the existence of maximally hard pattern graphs $H$ that require time $n^{\tw(H)+1-o(1)}$. Specifically, under the Strong Exponential Time Hypothesis (SETH), a standard assumption from fine-grained complexity theory, we prove the following asymptotic statement for large treewidth $t$:

\begin{center}
For any $\varepsilon > 0$ there exists $t \ge 3$ and a pattern graph $H$ of treewidth $t$ such that \\ Subgraph Isomorphism on pattern $H$ has no algorithm running in time $O(n^{t+1-\varepsilon})$.
\end{center}

\noindent
Under the more recent 3-uniform Hyperclique hypothesis, we even obtain tight lower bounds for each specific treewidth $t \ge 3$:

\begin{center}
For any $t \ge 3$ there exists a pattern graph $H$ of treewidth $t$ such that for any $\varepsilon>0$ \\ Subgraph Isomorphism on pattern $H$ has no algorithm running in time $O(n^{t+1-\varepsilon})$.
\end{center}

\noindent
In addition to these main results, we explore \textbf{(1)} colored and uncolored problem variants (and why they are equivalent for most cases), \textbf{(2)} Subgraph Isomorphism for $\tw < 3$, \textbf{(3)} Subgraph Isomorphism parameterized by pathwidth instead of treewidth, and \textbf{(4)} a weighted variant that we call Exact Weight Subgraph Isomorphism, for which we examine pseudo-polynomial time algorithms. For many of these settings we obtain similarly tight upper and lower bounds.
\end{abstract}

\clearpage
\section{Introduction}

The \textsc{Subgraph Isomorphism} problem is commonly defined as follows: Given a graph $H$ on $k$ vertices, and a graph $G$ on $n$ vertices, is there a (not necessarily induced) subgraph of G which is isomorphic to $H$?

\textsc{Subgraph Isomorphism} generalizes many problems of independent interest, such as the \textsc{$k$-path} and \textsc{$k$-clique} problems. The problem is also of considerable interest when $H$ is less structured, with applications to discovering patterns in graphs that, for example, arise from biological processes such as gene transcription or food networks, from social interaction, from electronic circuits, from neural networks~\cite{milo2002network}, from chemical compounds~\cite{sussenguth1965graph} or from control flow in programs~\cite{bruschi2006detecting}. In some fields, the problem is sometimes referred to as the search for ``network motifs'', i.e. subgraphs that appear more often than would normally be expected.

In its general form, the problem is NP-hard. We are interested in solving the problem when the pattern graph $H$ is ``tree-like'' or ``path-like'', i.e. when the treewidth $\tw(H)$ or the pathwidth $\pw(H)$ of $H$ is bounded. Such pattern graphs of low treewidth or pathwidth often arise in practice when considering the structure of chemical compounds, the control flow of programs, syntactic relations in natural language, or many other graphs from practical applications (see e.g. \cite{bodlaender1994tourist, bodlaender2005discovering}). On the theoretical side, many restricted classes of graphs have bounded treewidth, see also~\cite{bodlaender1998partial}. Restricting NP-hard problems to graphs of bounded tree- and pathwidth often yields polynomial-time algorithms, and \textsc{Subgraph Isomorphism} is no exception. Most notably, the classic Color-Coding algorithm by Alon, Yuster and Zwick~\cite{alon1995color} solves the problem by a Las Vegas algorithm in expected time $O(n^{\tw(H)+1}g(k))$, or by a deterministic algorithm in time $\widetilde{O}(n^{\tw(H)+1}g(k))$, where $g$ is a computable function (and \(\widetilde{O}(\cdot)\) is used to suppress factors that are polylogarithmic in the input size). In other words, if the pattern graph $H$ has treewidth bounded by some constant, the problem is fixed-parameter tractable when parameterized by $k$. The Color-Coding algorithm is also relevant for practical purposes: Recently, it has received an efficient implementation, which tested well against state-of-the-art programs for \textsc{Subgraph Isomorphism}~\cite{malik2019efficient}.

Many researchers wondered whether the Color-Coding algorithm can be improved. This question has been studied in many different directions, including the following: 
\begin{itemize}
\item Marx~\cite{marx2007can} showed that no algorithm solves the \textsc{Subgraph Isomorphism} problem in time $O(n^{o(\tw(H) / \log(\tw(H)))}g(k))$ unless the Exponential Time Hypothesis (ETH) fails, and this even holds when restricted to any class of pattern graphs of unbounded treewidth. 
\item A series of work has improved the computable function $g$, see e.g.~\cite{amini2009counting,fomin2012faster,pratt2019waring}. 
\item For many special pattern graphs faster algorithms have been found; the most famous example is the $k$-Clique problem, which can be solved in time $O(n^{k\omega/3}g(k))$~\cite{nevsetvril1985complexity}\footnote{This bound assumes that $k$ is divisible by 3; there are similar results for general \(k\)~\cite{eisenbrand2004complexity}.}.
\end{itemize}

In this paper, we use a different angle to approach the question whether Color-Coding can be improved. We ask whether there exist ``hard'' pattern graphs:
\begin{center}
  \emph{Do there exist pattern graphs $H$ for which \textsc{Subgraph Isomorphism} \\ cannot be
  solved in time $O(n^{\tw(H)+1-\varepsilon})$ for any constant $\varepsilon>0$?}
\end{center}
To the best of our knowledge, this question has not been previously studied.
As our main result, we (conditionally) give a positive answer to this question. 
More precisely, we show that for every $t \ge 3$ there exists a pattern graph $H$ with $\tw(H)=t$ for which \textsc{Subgraph Isomorphism} cannot be solved in time $O(n^{\tw(H)+1-\varepsilon})$ for any constant $\varepsilon > 0$, assuming the 3-uniform $k$-Hyperclique hypothesis; see Section~\ref{sec:results} for details on this hypothesis. 
We also show a slightly weaker statement under the Strong Exponential Time Hypothesis. 
This conditionally shows that the Color-Coding algorithm by Alon, Yuster and Zwick cannot be significantly improved while still working for all pattern graphs.

For the case of $\tw(H) = 2$, an algorithm of Curticapean, Dell and Marx~\cite{curticapean2017homomorphisms} can be adapted such that it solves \textsc{Subgraph Isomorphism} in time $\widetilde{O}(n^\omega g(k))$. We unify this with the algorithm of Alon, Yuster and Zwick by showing that both time bounds can be achieved within a simple framework. In particular, we use so-called $k$-wise matrix products, an operation which was introduced in its general form in~\cite{gnang2011spectral} and studied further in~\cite{lincoln2018tight}. 

We also study the \textsc{Subgraph Isomorphism} problem when the pathwidth of $H$ is bounded, and specialize our framework to show slight improvements in running time compared to the case of bounded treewidth. Here, we use rectangular matrix products, for which faster-than-naive algorithms are known~\cite{gall2018improved}.

In further results, our focus is on the weighted variant \textsc{Exact Weight Subgraph Isomorphism}, where the subgraph must also have total weight equal to zero. In this work, we consider both the node-weighted and the edge-weighted variant of this problem, for both bounded treewidth and bounded pathwidth, allowing the maximum absolute weight~$W$ to appear in the running time (i.e.\ the pseudopolynomial-time setting). We show that our algorithms for the unweighted case can be adapted to the weighted case. We also speed up the weighted algorithms by using the fact that fast convolution (or rather, sumset computation), a folklore technique that lies at the core of many fast algorithms for problems with weights (e.g.~\cite{chan2015clustered, bringmann2017near, kunnemann2017fine, koiliaris2017faster, bremner2006necklaces, bateni2018fast} and~\cite[exercise 30.1.7]{cormen2009introduction}), can easily be adapted to work with rectangular matrices and tensors. We furthermore show tight conditional lower bounds in many cases. Last but not least, we show that our algorithms can be slightly improved for the case of node-weighted instances for which either the pathwidth of $H$ is bounded, or $H$ is a tree. These algorithms also rely on fast rectangular matrix products.

\subsection{Related Work}
\label{sec:related-work}

Additional to the conditional lower bound of $O(n^{o(\tw(H) / \log(\tw(H)))}g(k))$ by Marx~\cite{marx2007can}, there is an unconditional lower bound of $O(n^{\kappa(H)})$ for the size of any $\mathrm{AC}^0$-circuit, for some graph parameter $\kappa(H) = \Omega(\tw(H)/\log(\tw(H)))$, which holds even when considering the average case~\cite{li2017ac}. Interestingly, the factor of $1/\log(\tw(H))$ does not seem to be an artefact of the proof: There is an $\mathrm{AC}^0$-circuit of size $O(n^{o(\tw(H))}g(k))$ that solves the problem on certain unbounded-treewidth classes in the average case~\cite{rosenthal2019beating}.

In a different direction, Dalirrooyfard et al.~\cite{dalirrooyfard2019graph} design various reductions from \(k\)-Clique to \textsc{Subgraph Isomorphism}, among other results. They also present results on the detection of induced subgraphs (we focus on non-induced subgraphs).

For the weighted variant of \textsc{Subgraph Isomorphism}, lower bounds under the \textsc{$k$-Sum} hypothesis for stars, paths, cycles and some other pattern graphs  
are presented in~\cite{abboud2013exact}. 
Edge-weighted triangle detection has a by-now classic $O(n^{3-\varepsilon})$ lower bound under both the \textsc{3Sum} hypothesis and the \textsc{APSP} hypothesis~\cite{abboud2018matching}. On the other hand, in~\cite{abboud2014losing}, it is proven that finding node-weighted $k$-cliques can be done almost as quickly as finding unweighted $k$-cliques. We are not aware of any results on the \textsc{Exact Weight Subgraph Isomorphism} problem when $W$ may appear in the running time (i.e.\ a pseudopolynomial-time algorithm), which is what we focus on here.

In our work, we pose no restrictions on the host graph $G$.
For an extensive classification of \textsc{Subgraph Isomorphism} with respect to various parameters of both $G$ and $H$, see~\cite{marx2013everything}. 






\subsection{Hardness Assumptions}

The most standard hypothesis from fine-grained complexity theory is the Strong Exponential Time Hypothesis (SETH)~\cite{impagliazzo2001complexity}, which postulates that for any $\varepsilon > 0$ there exists $k \geq 3$ such that $k$-\textsc{Sat} on $n$ variables cannot be solved in time $O^*(2^{(1-\varepsilon)n})$. 

More recent is the Hyperclique hypothesis. In the \(h\)-uniform \(k\)-\textsc{Hyperclique} problem, for a given \(h\)-uniform hypergraph we want to decide whether there exist a set of \(k\) vertices such that every size-\(h\) subset of these vertices forms a hyperedge. For any $k>3$, the 3-uniform $k$-Hyperclique hypothesis postulates that this problem cannot be solved in time $O(n^{k-\varepsilon})$ for any $\varepsilon>0$. This hypothesis has also been formulated when replacing 3 with any $h < k$, getting progressively more believable with larger $h$. For a more in-depth discussion of the believability of this hypothesis we refer to~\cite[Section 7]{lincoln2018tight}.

Note that we will also use the $h$-uniform Hyperclique hypothesis for various $h$, which is simply the conjecture that the $h$-uniform $k$-Hyperclique hypothesis is true for all $k>h$.

Related to this is the $k$-Clique conjecture, which postulates that the $k$-Clique problem (which is the 2-uniform $k$-Hyperclique problem) cannot be solved in time $O(n^{\omega k /3 - \varepsilon})$ for any constant $\varepsilon > 0$, where $\omega < 2.373$~\cite{le2014powers} is the exponent of matrix multiplication.

\subsection{Our Results}
\label{sec:results}

\subparagraph*{Unweighted Subgraph Isomorphism with Bounded Treewidth}
First, consider the case of the unweighted \textsc{Subgraph Isomorphism} problem for bounded-treewidth pattern graphs \(H\). As was said, and as we will re-prove with a unified algorithm later, this problem has an algorithm running in time \(\widetilde{O}(n^{\tw(H)+1})\) for \(\tw(H) \geq 3\). We show tight conditional lower bounds by proving the following obstacles to faster algorithms, which use the \(k\)-clique hypothesis and the \(h\)-uniform \(k\)-hyperclique hypothesis. Note that when we say, for some \(x\), that an algorithm has running time \(O(n^{x-\varepsilon})\), what we mean is that the algorithm runs in time \(O(n^{x-\varepsilon})\) for some constant \(\varepsilon > 0\).

\begin{theorem}\label{corollary-lower-bound-detection}
  The following statements are true. 
  \begin{enumerate}
  \item For each \(t\geq 3\) and each \(3\leq h \leq t\), there exists a connected, bipartite pattern graph \(\mathcal{H}_{t,h}\) of treewidth \(t\) such that there cannot be an algorithm solving the \textsc{Subgraph Isomorphism} problem on pattern graph \(\mathcal{H}_{t,h}\) in time \(O(n^{t+1-\varepsilon})\) unless the \(h\)-uniform \(h(t+1)\)-hyperclique hypothesis fails.
  \item For each \(t\geq 2\) and each \(h \geq 3\), there exists a connected, bipartite pattern graph \(\mathcal{H}_{t,h}\) of treewidth \(t\) such that there cannot be an algorithm solving the \textsc{Subgraph Isomorphism} problem on pattern graph \(\mathcal{H}_{t,h}\) in time \(O(n^{t-\varepsilon})\) unless the \(h\)-uniform \(ht\)-hyperclique hypothesis fails.
  \item For each \(t \geq 2\), there exists a connected, bipartite pattern graph \(\mathcal{H}_{t}\) of treewidth \(t\) such that there cannot be an algorithm solving the \textsc{Subgraph Isomorphism} problem on pattern graph \(\mathcal{H}_{t}\) in time \(O(n^{(t+1)\omega/3-\varepsilon})\) unless the \((t+1)\)-\textsc{Clique} hypothesis fails.
  \end{enumerate}
\end{theorem}

Indeed, with the very same reduction, we also get an obstacle from SETH. However, the lower bound it provides is not as tight as the above, and in the case of the second part does not work for each target treewidth \(t\).
\begin{theorem}\label{corollary-lower-bound-detection-seth}
  Assuming SETH, the following two statements are true.
  \begin{enumerate}
  \item For any \(t\geq 3\) and any \(\varepsilon > 0\) there exists a pattern graph \(\mathcal{H_{t,\varepsilon}}\) of treewidth \(t\) such that there cannot be an algorithm solving all instances of \textsc{Subgraph Isomorphism} with pattern graph \(\mathcal{H}_{t,\varepsilon}\) in time \(O(n^{t-\varepsilon})\).
    \item For any \(\varepsilon > 0\) there exists a \(t\geq 3\) and a pattern graph \(\mathcal{H}_{\varepsilon}\) of treewidth \(t\) such that there cannot be an algorithm solving all instances of \textsc{Subgraph Isomorphism} with pattern graph \(\mathcal{H}_{\varepsilon}\) in time \(O(n^{t+1-\varepsilon})\).
\end{enumerate}
\end{theorem}

On the algorithmic side, we present an algorithm that achieves matching running times (as listed in theorem~\ref{corollary-upper-bound} below). As was said, the results in the following theorem are not new. Part 1 was shown via Color-Coding in~\cite{alon1995color} and part 2 follows from techniques in~\cite{curticapean2017homomorphisms}. We unify these two results by providing a single, relatively simple algorithmic technique achieving both, based on $k$-wise matrix products. These techniques are later expanded to also work for the weighted version, where they then achieve new results. In the following, $\omega < 2.373$~\cite{le2014powers} is the exponent of matrix multiplication.

\begin{theorem}\label{corollary-upper-bound}
  There are algorithms which, given an arbitrary instance $\phi = (H,G)$ of \textsc{Subgraph Isomorphism} where $H$ has treewidth $\tw(H)$, solve $\phi$ in
  \begin{enumerate}
  \item time $\widetilde{O}(n^{\tw(H)+1}g(k))$ when $\tw(H)\geq 3$,
  \item time $\widetilde{O}(n^\omega g(k))$ when $\tw(H) = 2$, and
  \item time $\widetilde{O}(n^2g(k))$ when $\tw(H) = 1$,
  \end{enumerate}
  where $k := |V(H)|, n := |V(G)|$ and $g$ is a computable function.
\end{theorem}

\subparagraph*{Semi-Equivalence of Hyperclique and Subgraph Isomorphism}
We also discuss how our results not only show a reduction from \textsc{Hyperclique} to \textsc{Subgraph Isomorphism} with bounded treewidth, but also in the other direction. For this, we show that calculating the boolean \(k\)-wise matrix products, which is the bottleneck in our algorithm for bounded-treewidth \textsc{Subgraph Isomorphism}, is actually equivalent to the \(k\)-uniform \((k+1)\)-\textsc{Hypergraph} problem. Hence we have a reduction in the second direction. This gives an interesting intuition for why the Hyperclique hypothesis is the ``correct'' conjecture to prove conditional hardness of \textsc{Subgraph Isomorphism} for bounded treewidth.

We remark that this does not lead to a full equivalence of these problems because the uniformity (i.e.\ the size of hyperedges) of the \textsc{Hyperclique} problem we reduce from in the first reduction is much smaller than the size of the hypercliques we search for. Hence we only have a reduction from a \textsc{Hyperclique} instance with small edge uniformity to \textsc{Subgraph Isomorphism}, and a reduction from \textsc{Subgraph Isomorphism} to \textsc{Hyperclique} instances with large edge uniformity.

\subparagraph*{Weighted Subgraph Isomorphism with Bounded Treewidth}
Now consider the weighted version of \textsc{Subgraph Isomorphism} for bounded-treewidth graphs \(H\). Recall that the weighted version can be either node- or edge-weighted and is defined such that the weights in the solution subgraph must have total weight zero. A trivial dynamic programming algorithm on the tree decomposition achieves a running time of \(\widetilde{O}(n^{\tw(H)+1} \cdot W\log W)\) for \(\tw(H) \geq 3\).

Note that these results show conditional lower bounds even when the maximum weight is restricted to \(W=\Theta(n^{\gamma})\), for any constant \(\gamma > 0\).

\begin{theorem}\label{corollary-lower-bound-weighted}
  For both the node- and edge weighted variant of the problems, the following statements are true.
  \begin{enumerate}
  \item For each \(t \geq 3\), each \(\gamma \in \mathbb{R}^+\) and each \(3\leq h \leq t\), there exists a connected, bipartite graph \(\mathcal{H}_{t,h,\gamma}\) of treewidth \(t\) such that there cannot be an algorithm solving the \textsc{Exact Weight Subgraph Isomorphism} problem on pattern graph \(\mathcal{H}_{t,h,\gamma}\) for instances with maximum weight \(W = \Theta(n^{\gamma})\) in time \(O(n^{t+1-\varepsilon}W)\), unless the \(h\)-uniform Hyperclique hypothesis fails.
  \item For each \(t\geq 1\), each \(\gamma \in \mathbb{R}^+\) and each \(h \geq 3\), there exists a connected, bipartite graph \(\mathcal{H}_{t,h,\gamma}\) of treewidth \(t\) such that there cannot be an algorithm solving the \textsc{Exact Weight Subgraph Isomorphism} problem on pattern graph \(\mathcal{H}_{t,h,\gamma}\) for instances with maximum weight \(W = \Theta(n^{\gamma})\) in time \(O(n^{t-\varepsilon}W)\), unless the \(h\)-uniform Hyperclique hypothesis fails.
    \item For each \(t\geq 1\) and each \(\gamma \in \mathbb{R}^+\), there exists a connected, bipartite graph \(\mathcal{H}_{t,\gamma}\) of treewidth \(t\) such that there cannot be an algorithm solving the \textsc{Exact Weight Subgraph Isomorphism} problem on pattern graph \(\mathcal{H}_{t,\gamma}\) for instances with maximum weight \(W=\Theta(n^{\gamma})\) in time \(O(n^{(t+1)\omega/3-\varepsilon}W^{\omega/3})\), unless the \textsc{Clique} hypothesis fails.
    \end{enumerate}
  \end{theorem}

  Similar lower bounds also hold when trying to reduce the exponent of \(W\) instead of \(n\). Meaning there is also no algorithm of running time \(O(n^{t+1}W^{1-\varepsilon})\) in part 1, etc.

On the algorithmic side, we present an algorithm that achieves matching running times for \(\tw(H) \geq 3\), and almost matching running times for \(\tw(H) = 1,2\). Note that in terms of exponents, the first algorithm below is not better than the naive one with running time \(O(n^{\tw+1}W\log W)\). However, it avoids a factor of \(\log W\) in the largest term, and instead appends it to a smaller term, so in a way it presents an improvement of \(\log W\) in the running time. Specifically, we show

\begin{theorem}\label{corollary-upper-bound-weighted}
  There are algorithms which, given an arbitrary instance $\phi = (H,G,w)$ of the \textsc{Exact Weight Subgraph Isomorphism} problem where $H$ has treewidth $\tw(H)$, solve $\phi$ in
  \begin{enumerate}
  \item time $\widetilde{O}((n^{\tw(H)+1}W + n^{\tw(H)}W\log W)g(k))$ when $\tw(H) \geq 3$,
  \item time $\widetilde{O}((n^\omega W + n^2W\log W)g(k))$ when $\tw(H) = 2$, or
  \item time $\widetilde{O}((n^2W + nW \log W)g(k))$ when $\tw(H) = 1$,
  \end{enumerate}
  where $n := |V(G)|, k := |V(H)|$, $g$ is a computable function, and $W$ is the maximum absolute weight in the image of $w$.
\end{theorem}

Comparing these upper bounds with the lower bounds from Theorem~\ref{corollary-upper-bound-weighted}, we have a tight lower bound for the weighted case with $\tw(H) \geq 3$. For weighted $\tw(H) = 2$, we have a lower bound which is tight except for the exponent of $\omega/3$ to $W$; it is unclear whether this can be strengthened. The lower bound for weighted graphs with $\tw(H) = 1$ is obviously not tight: We have an upper bound of \(\widetilde{O}(n^2W + nW\log W)\), but our lower bounds only states that it requires time \(O(n^{1-o(1)}W^{1-o(1)})\) and \(O(n^{2\omega/3-o(1)}W^{\omega/3-o(1)})\). Tighter lower bounds for this case remain an important open problem.

\subparagraph*{Unweighted Subgraph Isomorphism with Bounded Pathwidth}
So far we have only looked at the case of bounded treewidth. However, similar results hold for the case of bounded pathwidth. Let us start with the unweighted \textsc{Subgraph Isomorphism} problem.

Note that we do not get any lower bounds for the current setting. This is because we prove all our lower bounds by showing an equivalence of the standard Subgraph Isomorphism problem to a colored variant (see also Section~\ref{sec:equivalence-statement}), and then proving a lower bound for the colored version. We do not know how to prove such an equivalence for the current setting, therefore we do not get lower bounds in this case; we leave this as an open problem.

Since a path decomposition is always also a tree decomposition, we trivially get upper bounds as in theorem~\ref{corollary-upper-bound} (when replacing treewidth by pathwidth). However, we can do better by using rectangular matrix multiplication to speed up the computation. For $z \in \mathbb{R}^+$, let $\omega(z)$ be the smallest real number such that multiplying a $n\times n$ matrix with a $n\times n^z$ matrix can be done in time $O(n^{\omega(z)})$, see Section~\ref{sec:algo-col-subiso-bounded-pathwidth} for discussion of this value. We prove the following upper bounds.
\begin{theorem}\label{corollary-upper-bound-detection-pathwidth}
  There are algorithms which, given an arbitrary instance $\phi = (H,G)$ of \textsc{Subgraph Isomorphism} where $H$ has pathwidth $p$, solve $\phi$ in
  \begin{enumerate}
  \item time $\widetilde{O}(n^{\omega(p-1)}g(k))$ when $p\geq 2$, and
  \item time $\widetilde{O}(n^2g(k))$ when $p=1$
  \end{enumerate}
  where $k := |V(H)|$ and $n := |V(G)|$.
\end{theorem}

We certainly have \(p \leq \omega(p-1) < p+1\), so these results represent only a minor improvement, which is nonetheless important because it ``beats'' the lower bound for treewidth. Hence the lower bound for pathwidth cannot be the same as for treewidth.

\subparagraph{Weighted Subgraph Isomorphism with Bounded Pathwidth}
We also analyze the bounded-pathwidth pattern graph version of \textsc{Weighted Subgraph Isomorphism}. Specifically, we get the following lower bound.
\begin{theorem}[Theorem~\ref{corollary-lower-bound-weighted} for pathwidth]\label{corollary-lower-bound-weighted-pathwidth}
  Parts 2 and 3 of Theorem~\ref{corollary-lower-bound-weighted} also hold when replacing the treewidth \(t\) by the pathwidth \(p\). Part 1 does not hold.
\end{theorem}
And on the algorithmic side, we can again use rectangular matrix multiplication to improve on the algorithms from the case of bounded treewidth. Specifically, we get:
\begin{theorem}\label{corollary-upper-bound-weighted-pathwidth}
  There are algorithms which, given an arbitrary instance $\phi = (H,G,w)$ of the \textsc{Exact Weight Subgraph Isomorphism} problem, solve $\phi$ in 
  \begin{enumerate}
  \item time $\widetilde{O}((n^{\omega(\pw(H)-1)}W + n^{\pw(H)}W\log W)g(k))$ when $\pw(H) \geq 2$,
  \item time $\widetilde{O}((n^2 W + nW\log W)g(k))$ when $\pw(H) = 1$
  \end{enumerate}
  where $n := |V(G)|, k := |V(H)|$, and $W$ is the maximum absolute weight in the image of $w$.
\end{theorem}
For $\pw(H) \geq 3$, the lower bounds are therefore obviously not tight (at least for current algorithms), unless significant advances in matrix multiplication techniques are made. For $\pw(H) = 1,2$, the situation is the same as with treewidth, see the discussion of Theorem~\ref{corollary-upper-bound-weighted}.

\subparagraph{Improvements to Special Cases of Weighted Subgraph Isomorphism}
It is natural to think that the exponents $\frac{\omega}{3}$ to $W$ in the lower bounds of Theorems~\ref{corollary-lower-bound-weighted} and~\ref{corollary-lower-bound-weighted-pathwidth} are only artefacts of the reduction, and that with more advanced methods, this exponent can be improved to 1. However, the following two theorems show that this notion is false for \(\tw(H) = 1\) and \(\pw(H) = 1,2\), at least when considering the node-weighted case. Indeed, for \(\tw(H) = 1\) (or \(\pw(H) = 1\)) and \(W = n\), these bounds are tight, so further general improvements on the exponent are impossible.

Specifically, Theorems~\ref{corollary-node-weighted-algo-trees} and~\ref{corollary-node-weighted-algo-pathwidth} show the following improvements of the algorithms from Theorems~\ref{corollary-upper-bound-weighted} and~\ref{corollary-upper-bound-weighted-pathwidth} for small tree- or pathwidth. Let $\MM(n,n,x)$ be the time in which a $n\times n$ matrix can be multiplied with with a $n\times x$ matrix.
\begin{theorem}\label{corollary-node-weighted-algo-trees}
  There is an algorithm which, given an arbitrary instance $\phi = (H,G,w)$ of the node-weighted \textsc{Exact Weight Subgraph Isomorphism} problem where $H$ is a tree, solves $\phi$ in time $\widetilde{O}((\MM(n,n,W) + nW\log W)g(k))$.
\end{theorem}
\begin{theorem}\label{corollary-node-weighted-algo-pathwidth}
  There is an algorithms which, given an arbitrary instance $\phi = (H,G,w)$ of the node-weighted \textsc{Exact Weight Subgraph Isomorphism} problem, solves $\phi$ in time $\widetilde{O}(\MM(n,n,n^{\pw(H)-1}W)g(k))$.
\end{theorem}

For $W = O(n^\gamma)$, the running time of Theorem~\ref{corollary-node-weighted-algo-trees} is $\widetilde{O}(n^{\omega(\gamma)}\poly(k))$. Using results from~\cite{gall2018improved}(discussed further in section~\ref{sec:algo-col-subiso-bounded-pathwidth}), this implies several interesting facts. First, for $\gamma < 0.31$, the node-weighted problem on trees can be solved in time $\widetilde{O}(n^2\poly(k))$, meaning it can be solved in the same running time as the unweighted case. In particular, this applies to $W = \polylog(n)$ or $W = O(\sqrt[4]{n})$.

Second, for arbitrary $\gamma$ we now have a running time of $\widetilde{O}(n^{\omega(\gamma) - \gamma}W\poly(k))$. Trivially, for any $\gamma > 0$, $\omega(\gamma) - \gamma < 2$. Indeed, for $\gamma \geq 5$ we have $\omega(\gamma) - \gamma < 1.16$ by~\cite{gall2018improved}. This shows that for node weighted trees, there cannot be a lower bound of $n^{1.16}W$, let alone $n^2W$. Indeed, it is known that $\lim_{\gamma \to \infty}\omega(\gamma) - \gamma = 1$~\cite{coppersmith1982rapid}, which implies that when restricting to instances where \(W=\Theta(n^\gamma)\), there cannot be a lower bound of $n^{1+\varepsilon}W$ for any $\varepsilon > 0$ that holds for any constant \(\gamma>0\). Similar results hold for bounded-pathwidth graphs.

  \paragraph*{Structure of the Paper}
We begin by giving a simplified view of our two main proofs in Section~\ref{section-sketches}. We then give formal definitions of key concepts in Section~\ref{section-preliminaries} before diving into the full proofs. In Section~\ref{section-hardness}, we prove all lower bounds. In Section~\ref{section-algos}, we prove all upper bounds, including improved algorithms for node-weighted graphs. In Section~\ref{sec:interconnections}, we prove the semi-equivalence between \textsc{Hyperclique} and \textsc{Subgraph Isomorphism} on bounded treewidth graphs. Finally, in Section~\ref{section-equivalence}, we prove the equivalences between the colored and uncolored problems which we use throughout the paper. We conclude by stating several important open problems in Section~\ref{sec:open-problems}.

\section{Technical Overview of Our Main Results}\label{section-sketches}

  We now give proof sketches of our two main lower bound results to present the main ideas of the proofs without giving too much detail. The full details are available in Section~\ref{section-hardness}.

  \subsection{Lower Bound for Subgraph Isomorphism}

  Our main result is the existence of the hard pattern graphs for bounded-treewidth \textsc{Subgraph Isomorphism}. We now prove their existence for treewidth at least 3 under the Hyperclique hypothesis, i.e.~part 1 of Theorem~\ref{corollary-lower-bound-detection}.

  The exact statement we prove is that for each $t\geq 3$ and each $3\leq h \leq t$, there exists a pattern graph of treewidth $t$ such that \textsc{Subgraph Isomorphism} cannot be solved in time $O(n^{t+1-\varepsilon})$ on that pattern graph unless the $h$-uniform $h(t+1)$-hyperclique hypothesis fails. Note that the proof actually shows this for the colored variant of \textsc{Subgraph Isomorphism}, after which we can use Lemma~\ref{lemma-equiv} to transfer the lower bound to the uncolored problem.

  \begin{figure}
    \centering
    \scalebox{1}{
      \tikzstyle{every picture}=[tikzfig]
      \begin{tikzpicture}
	\begin{pgfonlayer}{nodelayer}
		\node [style=Rnodes] (1) at (-19.5, -8) {};
		\node [style=Rnodes] (2) at (-23.25, -8.25) {};
		\node [style=Rnodes] (3) at (-22.75, -9.5) {};
		\node [style=Rnodes] (4) at (-20, -10.5) {};
		\node [style=Rnodes] (5) at (-21.5, -6.75) {};
		\node [style=Rnodes 2] (6) at (-16.5, -4.5) {};
		\node [style=Rnodes 2] (7) at (-25.75, -3.75) {};
		\node [style=Rnodes 2] (8) at (-22, -14.75) {};
		\node [style=Rnodes 2] (9) at (-25.5, -12.5) {};
		\node [style=Rnodes 2] (10) at (-26.75, -8.5) {};
		\node [style=none] (12) at (-22, -8.75) {\textcolor{red}{$A$}};
		\node [style=none] (13) at (-26.5, -11) {\textcolor{rgb,255: red,255; green,128; blue,0}{$B$}};
		\node [style=none] (14) at (-26.25, -6) {\textcolor{rgb,255: red,255; green,128; blue,0}{\reflectbox{$\vdots$}}};
		\node [style=none] (15) at (-24, -13.75) {\textcolor{rgb,255: red,255; green,128; blue,0}{$\ddots$}};
		\node [style=Hollow Hnode] (18) at (-9.25, -7.25) {};
		\node [style=Gnode] (19) at (-8.75, -6.5) {};
		\node [style=Gnode] (20) at (-10, -6.75) {};
		\node [style=Gnode] (21) at (-9.75, -8) {};
		\node [style=Gnode] (22) at (-8.5, -7.5) {};
		\node [style=light hollow Hnode] (23) at (-12.75, -14) {};
		\node [style=Gnode R] (24) at (-13.25, -14.5) {};
		\node [style=Gnode R] (25) at (-12.25, -13.25) {};
		\node [style=Gnode R] (26) at (-13.25, -13.25) {};
		\node [style=Gnode R] (27) at (-11.75, -14.25) {};
		\node [style=text node] (28) at (-2.5, -3.25) {\Large{Corresponds to a choice of one vertex from each of the $h$ partitions}};
		\node [style=text node] (29) at (-6, -15) {\Large{Corresponds to a hyperedge}};
		\node [style=text node] (30) at (-12.75, -1.5) {\Large{Corresponds to $h$ vertex partitions of the hyperclique instance}};
		\node [style=text node] (31) at (-6.25, -17.75) {\Large{Corresponds to the set of hyperedges between its $h$ defining vertex partitions}};
		\node [style=text node] (32) at (-2.25, -10) {\Large{The edge exists if and only if the endpoints are compatible, i.e.\ if for each hyperedge partition that involves both endpoints, the vertex choice on one end agrees with the vertex of the hyperedge on the other}};
		\node [style=none] (33) at (-14.25, -2.75) {};
		\node [style=none] (34) at (-4.25, -4.75) {};
		\node [style=none] (35) at (-9.5, -17) {};
		\node [style=none] (36) at (-6.75, -14.25) {};
		\node [style=none] (37) at (-10.1, -10) {};
		\node [style=none] (38) at (-5.5, -10) {};
		\node [style=none] (39) at (-20.25, -9.75) {};
		\node [style=none] (40) at (-19.25, -10.25) {};
		\node [style=none] (41) at (-21.75, -15.5) {};
		\node [style=none] (42) at (-22.75, -15) {};
		\node [style=none] (43) at (-6, -6) {};
		\node [style=none] (44) at (-10, -4) {};
		\node [style=none] (45) at (-12, -17) {};
		\node [style=none] (47) at (-16, -15) {};
		\node [style=none] (48) at (-16, -15) {};
		\node [style=center node] (49) at (-21, -8.5) {};
		\node [style=none] (50) at (-20.5, -3.25) {\textcolor{orange}{pattern graph $H$}};
		\node [style=none] (51) at (-8.25, -3.25) {\textcolor{blue}{host graph $G$}};
		\node [style=none] (52) at (-12.25, -8) {};
	\end{pgfonlayer}
	\begin{pgfonlayer}{edgelayer}
		\draw [style=R-edge, bend left=15] (1) to (6);
		\draw [style=R-edge, bend right=15] (1) to (6);
		\draw [style=R-edge, bend left=15] (5) to (6);
		\draw [style=R-edge] (5) to (7);
		\draw [style=R-edge, bend right] (5) to (7);
		\draw [style=R-edge, bend left] (5) to (7);
		\draw [style=R-edge] (4) to (8);
		\draw [style=R-edge, bend left=15] (5) to (8);
		\draw [style=R-edge, bend right=15] (3) to (8);
		\draw [style=R-edge, bend right] (9) to (4);
		\draw [style=R-edge, bend left=15] (9) to (3);
		\draw [style=R-edge, bend right=15] (9) to (1);
		\draw [style=R-edge, bend right=15] (10) to (4);
		\draw [style=R-edge, bend left] (10) to (1);
		\draw [style=R-edge] (10) to (2);
		\draw [style=Gedge] (22) to (25);
		\draw [style=description, bend right] (33.center) to (18);
		\draw [style=description, bend left=15] (34.center) to (22);
		\draw [style=description, bend left=15, looseness=0.75] (35.center) to (23);
		\draw [style=description, bend right=15] (36.center) to (25);
		\draw [style=description, bend left=15] (38.center) to (37.center);
		\draw (39.center) to (42.center);
		\draw (41.center) to (40.center);
		\draw (42.center) to (41.center);
		\draw (39.center) to (40.center);
		\draw (43.center) to (45.center);
		\draw (45.center) to (47.center);
		\draw (47.center) to (44.center);
		\draw (44.center) to (43.center);
		\draw [style=dashedline2] (41.center) to (45.center);
		\draw [style=dashedline2] (39.center) to (44.center);
		\draw [style=dashedline2] (42.center) to (48.center);
		\draw [style=dashedline2] (40.center) to (52.center);
	\end{pgfonlayer}
\end{tikzpicture}
    }
    \caption{A sketch of the reduction we use to prove part 1 of Theorem~\protect{\ref{corollary-lower-bound-detection}} where $h = 3$ and $t = 4$. Note that this is only a partial sketch of the pattern graph. We use multiedges to signify that for the endpoints $a\in A$ and $b = ((a_1, j_1), \ldots, (a_h,j_h)) \in B$ there exists more than one $\ell$ such that $a = a_\ell$.}
    \label{fig:main-proof}
  \end{figure}
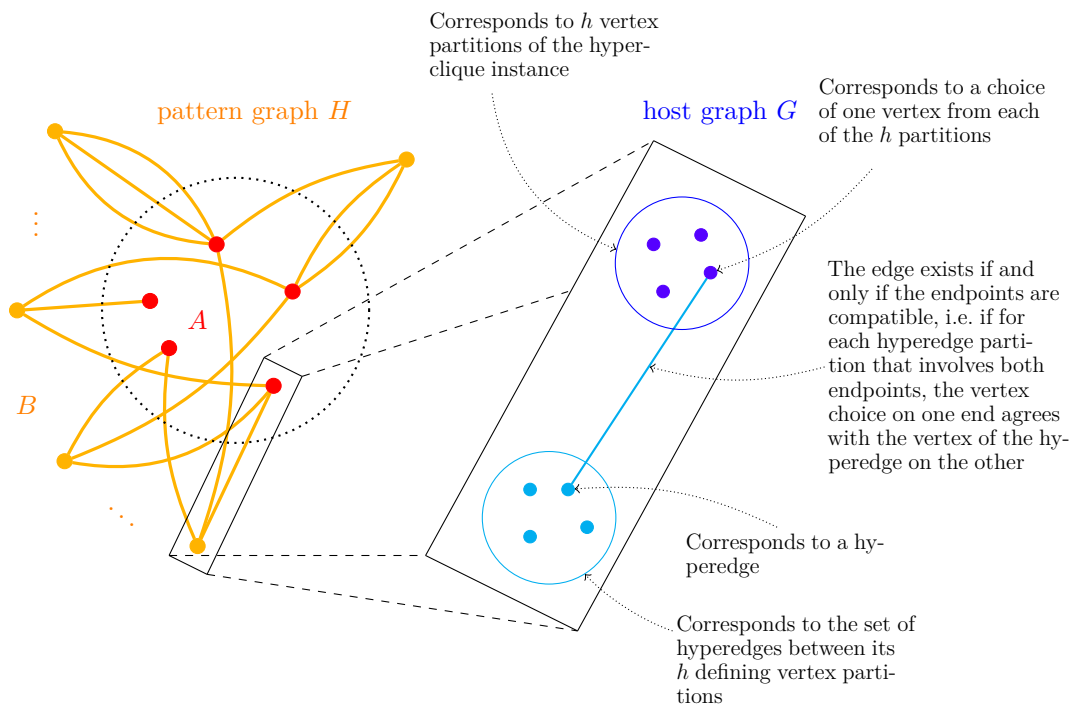

  \begin{proof}[Proof sketch] See Figure~\ref{fig:main-proof} for a sketch of the reduction. Let $t\geq 3$ and $3\leq h \leq t$ be given and assume that \textsc{Subgraph Isomorphism} can be solved in time $O(n^{t+1-\varepsilon})$ on pattern graphs of treewidth $t$. We show that the $h$-uniform $h(t+1)$-hyperclique hypothesis fails.
    \begin{description}
    \item [Construction of H] We construct a pattern graph $H$ as a bipartite graph with vertex set $A \cup B$ as follows. Writing $[c] := \{1,\ldots, c\}$, we set $A := [t+1]$ and $B := {A \times [h] \choose h}$. We connect a vertex $b = ((a_1,j_1),..,(a_h,j_h))$ in $B$ to a vertex $a$ in $A$ if $a = a_\ell$ for some $\ell$. Set $k := |A|+|B|$.

      We show that this pattern has a treewidth of $t$, via a well-known characterization of treewidth as a graph-theoretic game: A graph $F$ has treewidth $\leq t$ if and only if $t+1$ cops can catch\footnote{The game works as follows: The $k+1$ cops select their starting vertices in the graph. Then the robber may choose their starting vertex. The cops can always see the robber and adapt their strategy accordingly. Similarly, the robber can see the cops. The game now proceeds in steps, where in each step, one of the cops chooses an arbitrary destination vertex and takes off via helicopter in the direction of that vertex. While the cop is travelling, the robber sees where they will land and may now move arbitrarily along edges of the graph, as long as they do not pass through stationary cops. When the robber has finished moving, the cop lands. The cops win if and only if they are guaranteed to catch the robber after a finite number of moves, and lose otherwise.} a robber on $F$~\cite{seymour1993graph}. To show the bound on the treewidth of $H$, initially place a cop on each vertex of $A$. No matter on which vertex of $B$ the robber starts, they are surrounded by cops. Since every vertex in $B$ has $h \leq t < t+1$ neighbors in $A$, there must exist some cop which is not adjacent to the robber, so this cop can catch the robber in a single step. This concludes the proof that the pattern graph $H$ has treewidth $t$.

  \item[Construction of G] Now given a $h(t+1)$-partite hypergraph $H'$, i.e.~an instance of the $h$-uniform $k'$-\textsc{Hyperclique} problem for $k':=h(t+1)$, we write the vertex set of $H'$ as $U_{1,1} \cup \ldots \cup U_{1,h} \cup \ldots \cup U_{t+1,1} \cup \ldots \cup U_{t+1,h}$. Let $N_H$ be the number of vertices vertices in each partition and $n_H = O(N_H)$ the number of vertices overall.

  We construct a $k$-partite graph $G$ as follows. For $a$ in $A$ we set $V_a := U_{a,1} \times \ldots \times U_{a,h}$. For $b = ((a_1,j_1),..,(a_h,j_h))$ in $B$, we set $V_b := E(H') \cap (U_{a_1,j_1} \times \ldots \times U_{a_h,j_h})$. This describes the $k$ parts of the $k$-partite vertex set $V(G)$. Note that each part has size at most $N_G:= N_H^h$. Now we construct the edges. For any $a$ in $A$ and $b = ((a_1,j_1),..,(a_h,j_h))$ in $B$ with $(a,b)$ in $E(H)$, consider an arbitrary $u = (u_1,..,u_h)$ in $V_a$ and $u' = (u'_1,...,u'_h)$ in $V_b$. We say that $u$ and $u'$ are "compatible" if for every $\ell$ with $a_\ell = a$ we have $u'_\ell$ = $u_{j_\ell}$; in this case we connect $u$ and $u'$ by an edge. This finishes the construction of $G$.

\item[Correctness] Note that any colored subgraph isomorphism of $H$ in $G$ chooses vertices $v_a$ in $V_a$ for all $a$ in $A$. This corresponds to choosing vertices $u_{i,j}$ in $U_{i,j}$ for all $i$ in $[t+1]$, Moreover, the edges of a $h(t+1)$-hyperclique are in one-to-one correspondence with the set $B$. Since for each $b$ in $B$ the colored subgraph isomorphism of $H$ in $G$ needs to choose a vertex $v_b$ in $V_b$, which corresponds to an edge between certain vertices $u_{i,j}$, we indeed check that the chosen vertices $u_{i,j}$ form an $h$-uniform $h(t+1)$-hyperclique.

\item[Running Time] Trivially, the construction time and output size are $O(N_H^{2h})$ (actually, it is slightly better, but this is not important in this proof sketch), and $G$ has $n = O(N_G) = O(N_H^h)$ vertices. Now if we can solve \textsc{Subgraph Isomorphism} in time $O(n^{t+1-\varepsilon})$, we solve the $h$-uniform $h(t+1)$-\textsc{Hyperclique} instance in time $O(N_H^{2h} + (N_H^h)^{t+1-\varepsilon}) = O(N_H^{h(t+1)-h\varepsilon}) = O(N_H^{h(t+1) - \varepsilon'}) = O(n_H^{h(t+1) - \varepsilon'})$.
\end{description}
\end{proof}

This shows part 1 of Theorem~\ref{corollary-lower-bound-detection}. Part 2 can be shown by the almost the exact same proof, except that now we choose the size of $A$ to be $t$ instead of $t+1$, and we start with an $ht$-hyperclique instead of an $h(t+1)$-hyperclique. It can be seen that in this case, the pattern graph still has treewidth $t$. The third part of the theorem can be seen by simply taking an instance of $(t+1)$-\textsc{Clique} and subdividing the edges in the obvious way to make the graph bipartite.

Let us also quickly mention how the proof of the slightly weaker bounds under SETH, i.e.~Theorem~\ref{corollary-lower-bound-detection-seth}, works. The split-and-list technique from~\cite{williams2005new} allows one to reduce the Satisfiability problem to \textsc{Hyperclique}. Using this technique, the following result was shown in~\cite[Lemma 9.1]{lincoln2018tight}.
  \begin{lemma}[\cite{lincoln2018tight}]\label{seth-to-hyperclique}
    Assuming SETH, for any $\varepsilon>0$ there exists $h\geq3$ such that for all $k>h$, the $h$-uniform $k$-\textsc{Hyperclique} problem is not in time $O(n^{k-\varepsilon})$.
  \end{lemma}
The SETH result now follows by using essentially the same reduction as above, but we prefix it by the reduction from \textsc{SAT} to \textsc{Hyperclique}.

\subsection{Lower Bound for Exact Weight  Subgraph Isomorphism}

We also give lower bounds for the exact weight variant of the \textsc{Subgraph Isomorphism} problem. In particular, we prove the existence of hard pattern graphs for the bounded-treewidth \textsc{Exact Weight Subgraph Isomorphism} problem for any polynomial weight bound. We give this result for any treewidth which is at least 3, and under the Hyperclique hypothesis. This is part 1 of Theorem~\ref{corollary-lower-bound-weighted}.

The exact statement we prove is that for each $t \geq 3$, $\gamma \in \mathbb{R}^+$ and $3\leq h \leq t$, there exists a pattern graph of treewidth $t$ such that \textsc{Exact Weight Subgraph Isomorphism} with maximum weight $W = \Theta(n^\gamma)$ cannot be solved in time $O(n^{t+1-\varepsilon}W)$ unless the $h$-uniform Hyperclique hypothesis fails. Again, we show this statement for the colored problem and transfer the lower bound via Lemma~\ref{lemma-equiv}.

To do this, we will encode part of a hyperclique instance in the edges of the \textsc{Exact Weight Subgraph Isomorphism} problem, and the rest of the instance in the weights. To do the latter, we need to encode certain equality constraints only via weights. This can be done using so-called $k$-average free sets\footnote{These $k$-average-free sets are a tool which are very useful for weighted problems, especially when they have additive elements. Such problems include \textsc{$k$-sum}, \textsc{Subset Sum}, \textsc{Bin Packing}, various scheduling problems, \textsc{Tree Partitioning}, \textsc{Max-Cut}, \textsc{Maximum/Minimum Bisection}, a \textsc{Dominating Set} variant with capacities, and similar~\cite{abboud2019seth, abboud2020scheduling, abboud2014losing, an2017complexity, fomin2014almost, jansen2013bin, dudek2020all}. Other uses of $k$-average-free sets in computer science include constructions in extremal graph theory, see e.g.~\cite{abboud20174, abboud2018hierarchy, alon2002testing}.}, which we define below.

\begin{definition}[$k$-average free sets]
A set $S\subseteq \mathbb{Z}$ is called $k$-average-free if, for any $s_1, \ldots, s_{k'+1} \in S$ with $k' \leq k$, we have $s_1+ \ldots + s_{k'} = k'\cdot s_{k'+1}$ if and only if $s_1 = \ldots = s_{k'+1}$. In other words, the average of $s_1, \ldots, s_{k'}\in S$ is in $S$ if and only if all $s_i$ are equal.  \end{definition}

We use the following construction for $k$-average free sets, originally proven in~\cite{behrend1946sets}, modified into a more useful version in~\cite{abboud2014losing} and formulated in this form in~\cite{abboud2019seth}.

\begin{lemma}\label{k-average-free-lemma}
    There exists a universal constant $c > 0$ such that, for all constants $\varepsilon \in (0,1)$ and $k \geq 2$, a $k$-average-free set $S$ of size $n$ with $S \subseteq [0, k^{c/\varepsilon}n^{1+\varepsilon}]$ can be constructed in time $\poly(n)$.
  \end{lemma}

  Let us now prove the statement about \textsc{Exact Weight Subgraph Isomorphism}. We will construct an instance that is node-weighted, however this can easily be converted into an edge-weighted version by moving the weight of each vertex to all of its incident edges.

\begin{proof}[Proof sketch]
  Let $t \geq 3$, $3 \leq h \leq t$ and $\gamma \in \mathbb{R}$ be given and assume that \textsc{Exact Weight Subgraph Isomorphism} can be solved in time $O(n^{t+1-\varepsilon}W)$ on instances where the pattern graph has treewidth $t$ and all weights are bounded by $W = \Theta(n^\gamma)$. We show that the $h$-uniform $k$-hyperclique hypothesis fails for some large enough $k$.
  \begin{description}
\item[Construction of H] We construct a pattern graph $H$ as a  graph with vertex set $(A_1 \cup A_2) \cup B$ as follows. We set $A_1 := [t+1]$, $A_2 := [r]$ (for some $r$ large enough) and $B := {{(A_1 \cup A_2) \times [h]} \choose h}$. We connect a vertex $b = ((a_1, j_1), \ldots, (a_h, j_h))$ in $B$ to a vertex $a$ in $A_1$ (not in $A_2$) if $a = a_\ell$ for some $\ell$. Set $k := |A_1| + |A_2| + |B|$. By almost the same proof as in the unweighted version, it can be shown that this pattern $H$ has treewidth $t$.

\item[Grouping partitions] Now let an instance of the $h$-uniform $k'$-\textsc{Hyperclique} problem be given, and write the vertex set of $H'$ as $U_1 \cup \ldots \cup U_k$.  Let $N_H$ be the number of vertices vertices in each partition and $n_H = O(N_H)$ the number of vertices overall.

  We construct the $k$-partite graph $G$ with at most some number $N_G$ of vertices in each partition as follows. We will encode a $\beta$-fraction of the \textsc{Hyperclique} instance in the weights of the final \textsc{Exact Weight Subgraph Isomorphism} instance, and a $(1-\beta)$-fraction in the edges, for some $\beta$ chosen appropriately. To do this, we will choose $\beta$ such that $\frac{\beta k'}{hr}, \frac{(1-\beta) k'}{h(t+1)} \in \mathbb{N}$ and then group the sets $U_1, \ldots, U_{\beta k'}$ into $hr$ groups and the sets $U_{\beta k' +1}, \ldots, U_{k'}$ into $h(t+1)$ groups. Specifically, for each $(x,y) \in [r]\times [h]$, we create the set $U^1_{x,y} = U_{(xr + y - 1)\frac{\beta k'}{hr}+1} \times \ldots \times U_{(xr + y)\frac{\beta k'}{hr}}$, and for each $(x,y) \in [t+1]\times [h]$, we create the set $U^2_{x,y} = U_{\beta k' + (x(t+1) + y - 1)\frac{(1-\beta)k'}{h(t+1)}+1} \times \ldots \times U_{\beta k' + (x(t+1) + y)\frac{(1-\beta)k'}{h(t+1)}}$.

\item[Vertices of G] Now for each $a$ in $A_1$ we set $V_{a} := U^1_{a,1} \times \ldots \times U^1_{a, h}$ and for each $a$ in $A_2$ we set $V_{a} := U^2_{a,1} \times \ldots \times U^2_{a,h}$. Finally, for each $b = ((a_1, j_1), \ldots, (a_h, j_h))$  in $B$, where for each $\ell$ we have $a_\ell \in A_{i_\ell}$, we set $V_b := E(H') \cap (U^{i_1}_{a_1, j_1} \times \ldots \times U^{i_h}_{a_h,j_h})$. This describes the $k$ parts of the $k$-partite vertex set $V(G)$. We choose $r$ large enough so that the maximum size of each part is $N_G := N_H^{(1-\beta)k'/(t+1)}$.

\item[Edges of G] Now we construct the edges and weights of the graph. Let us start with the edges. The construction here is basically the same as the construction of the edges in the unweighted proof in the last section. For each $a$ in $A_2$ and $b = ((a_1,j_1), \ldots, (a_h, j_h))$ in $B$ with $(a,b)$ in $E(H)$, consider an arbitrary $u = (u_1, \ldots, u_h)$ in $V_a$ and $u' = (u'_1, \ldots, u'_h)$ in $V_b$. We say that $u$ and $u'$ are ``compatible'' if for every $\ell$ with $a_\ell = a$ we have $u'_\ell = u_{j_\ell}$; in this case we connect $u$ and $u'$ by an edge. This finishes the construction of the edges of $G$.

\item[Weights of G] Now we construct the weights. We want to encode the same edge constraints as we just encoded for $A_2$, but now for $A_1$, and we have to use weights instead of edges. To do this, we use $|B|$-average free sets via the construction of Lemma~\ref{k-average-free-lemma}. We simplify the usage in this shortened proof to avoid dealing with too many variables. We use the lemma to obtain in polynomial time (which we will treat as negligible here) a $|B|$-average free set $S$ of size $N_H^{\beta k' / (hr)}$ such that $S \subseteq [0, C]$, where $C \approx O(N_H^{\beta k' /(hr)})$ (up to a factor of $(1 + \varepsilon)$ in the exponent, but we will ignore this here for simplicity). From this, we can construct an arbitrary bijection $\varrho_S: [N_H^{\beta k'/(hr)}] \to S$.

  To simplify our construction, we specify a target weight $T$ (instead of the default target zero). We can easily get rid of this again later by subtracting $T$ from the weights of all vertices of some set of the partition. The binary representation of $T$ consists of $hr$ blocks of $\lceil 2|B|C \rceil$ bits, indexed by pairs $(i,j) \in [r]\times [h]$, each containing the binary representation of $|B|C$. The block $(i,j)$ represents the group $U^1_{i,j}$. The size of the blocks is large enough to prevent overflow between the blocks. Note that the maximum weight $W$ now satisfies $\log_2(W) = \Theta(hr\log_2(2|B|C))$ and hence $W \approx O(N_H^{\beta k'})$.

  Let us now actually specify the weights of the vertices, beginning with the vertices in $V_a$ for $a\in A_1$. For each $(i,j) \in [r]\times [h]$, we relabel the elements of each $U^1_{i,j}$ as $\{1, \ldots, N_H^{\beta k'/(hr)}\}$. Now we define the weight of the vertex $V_a \ni u = (u_1, \ldots, u_h)$ to have, for each $i \in [h]$, the value $|B|C - |N(a)|\cdot \varrho_S(u_i)$ in the block $(a,i)$ of its binary representation. Now we move on to the vertices in $V_b$ for $b = ((a_1, j_1), \ldots, (a_h, j_h))$. We define the weight of the vertex $V_b \ni u' = (u'_1, \ldots, u'_h)$ to have, for each $i \in [h]$ such that $a_i \in A_1$, the value $\varrho_S(u'_i)$ in the block $(a_i, j_i)$ of its binary representation.

  All blocks and vertices which have not been assigned a weight yet are assigned a value of zero. This concludes the construction of $G$.

\item[Correctness] Note that any colored subgraph isomorphism of $H$ in $G$ chooses vertices $v_a$ in $V_a$ for all $a$ in $A_1\cup A_2$. This corresponds to choosing vertices $u_i$ in $U_i$ for each $i \in [k']$. Moreover, the edges of a $k'$-hyperclique are in one-to-one correspondence with the set $B$. We simply need to show that the choice of hyperclique vertices induced by the choice of vertices in $A_1 \cup A_2$ agrees with the choice of hyperclique edges induced by the choice of vertices in $B$. For the vertices in $A_2$, this is easily seen to be ensured by the edges. For the vertices in $A_1$, we need to prove that the weights encode the same constraint. This, however, is simply the definition of a $|B|$-average free set: Consider the block $(i,j)$ (where $(i,j) \in [r]\times [h]$) in the binary representation of the total weight of the subgraph. Suppose that for $i \in A_1$ the vertex $V_i \ni u = (u_1, \ldots, u_h)$ was selected, and that for each $B \ni b = ((a_1,j_1), \ldots, (a_h, j_h))$ with $\exists \ell: (a_\ell, j_\ell) = (i,j)$ the vertex $V_b \ni u' = (u'_1, \ldots, u'_h)$ was selected. Then by construction, the total value in the block $(i,j)$ is the value $|B|C - |N(i)|\varrho_S(u_j)$ (where $N(i)$ is the neighbourhood of $i \in A_1$), plus the value $\varrho_S(u'_\ell)$ for all $b$ as above. Note that the latter term has exactly $|N(i)|$ summands, hence in order for the value in the block to be equal to $|B|C$ as specified by the target weight, we must have that the value of $\varrho_S(u_j)$ is equal to the value of each of the $\varrho_S(u'_\ell)$ by the definition of $|B|$-average free sets. Since $\varrho_S$ is a bijection, this ensures that the choice of hyperclique vertices in the sets appearing in the Cartesian product defining $U^1_{i,j}$ -- i.e. $U_{(ir + j -1)\frac{beta k'}{hr} + 1}, \ldots, U_{(ir+j)\frac{\beta k'}{hr}}$ -- agree with the choice of hyperclique edges. This is true for all $i,j$ and hence for each $U_\ell$ for $\ell \in [k']$.

  The other direction is easy to see via a similar, simpler argument. This concludes the correctness proof.

\item[Running Time] It can be seen that the running time of this reduction is $O(N_H^{2h})$, up to the running time of the algorithm for the construction of the $|B|$-average free set, which we will ignore here for sake of simplicity. Now suppose we can solve \textsc{Exact Weight Subgraph Isomorphism} in time $O(n^{t+1-\varepsilon}W)$. We use the reduction above to convert a \textsc{Hyperclique} instance with $n_H = O(N_H)$ nodes to an \textsc{Exact Weight Subgraph Isomorphism} instance where $W \approx \Theta(N_H^{\beta k'})$ and $n = O(N_H^{(1-\beta)k'/(t+1)})$. Choosing $\beta$ carefully, we get $W = \Theta(N_H^\gamma)$; note that we are ignoring some intricacies in the choice of $\beta$ that arise when you consider the running time of the algorithm that constructs the $|B|$-average free set -- the details are available in Section~\ref{section-hardness}. Now via the algorithm for \textsc{Exact Weight Subgraph Isomorphism}, we can solve this instance and hence the original \textsc{Hyperclique} problem in time $O(N_H^{2h} + N_H^{((1-\beta)k'/(t+1))(t+1-\varepsilon)} \cdot N_H^{\beta k'}) = O(N_H^{k'-\varepsilon'}) = O(n_H^{k'-\varepsilon'})$. 
  \end{description}
\end{proof}

It is easy to see that the same proof also rules out algorithms running in time $O(n^{t+1}W^{1-\varepsilon})$. Similar as with the proof for the unweighted problem, basically the same techniques can be used to prove the other parts of Theorem~\ref{corollary-lower-bound-weighted}.

\section{Preliminaries}\label{section-preliminaries}

\subsection{General Notation and Nomenclature}\label{notation-and-nomenclature-general}

We denote by $\mathbb{N}$ the set of positive integers. For $p \in \mathbb{N}$, we use $[p]$ to denote the set $\{1, \ldots, p\}$. For a statement or predicate \(P\), we define the Iverson bracket \([P]\) as 1 if \([P]\) is true, and zero otherwise. To declutter notation that relies heavily on the Iverson bracket, we will often use truth values and 0/1 interchangeably, where true will be indicated by 1 and false by 0.

For a function $f: \mathcal{A} \to \mathcal{B}$ and a set $S \subseteq \mathcal{A}$, we denote with $f|_S: S \to \mathcal{B}$ the function $f$ restricted to $S$. That is, $\forall s \in S: f|_S(s) = f(s)$. Furthermore, for $u \notin \mathcal{A}$ and $v \notin \mathcal{B}$ we define the function extension $(f \cup \{u \mapsto v\}): \mathcal{A} \cup \{u\} \to \mathcal{B} \cup \{v\}$ as
\begin{align*}
(f \cup \{u \mapsto v\})(c) := \begin{cases} v & \text{ if $c = u$} \\ f(c) & \text{ otherwise}\end{cases}
\end{align*}

We use standard notation for graphs. In particular, for a graph $G$, we let $V(G)$ be its set of vertices and $E(G)$ its set of edges. For a set $X \subseteq V(G)$, we denote the induced subgraph by $G[X]$. For a vertex $v \in V(G)$, we denote its neighbourhood as \(N_G(v)\), or as $N(v)$ when \(G\) is clear from context. We denote the treewidth and pathwidth (see Section~\ref{prelim-treewidth}) of $G$ as $\tw(G)$ and $\pw(G)$, respectively. All graphs are, unless otherwise stated, simple, undirected and without self-loops.

We use $\poly(n)$ to denote functions which are upper-bounded by $O(n^c)$ for some \(c \in \mathbb{N}\), and $\polylog(n)$ to denote functions upper-bounded by $O(\log^c(n))$ for some $c \in \mathbb{N}$. In running times, we use $O^*(\cdot)$ to suppress factors that are polynomial in the input size, and \(\widetilde{O}(\cdot)\) to suppress factors that are polylogarithmic in the input size.

In all weighted problems, we assume without further mention that the target weight or maximum absolute weight is at least 1. This is to avoid special cases with the running time.

\subsection{Notation and Nomenclature for Colored Subgraph Isomorphism}\label{notation-and-nomenclature-colsubiso}

We now define some nomenclature for the \textsc{(Exact Weight) Colored Subgraph Isomorphism} problem.
Not that instead of talking about colors, we will talk about a ``color homomorphism'' \(f\). Specifically, the \textsc{Colored Subgraph Isomorphism} is defined as follows: Given a pattern graph $H$ and a host graph $G$ along with a graph homomorphism $f: V(G) \to V(H)$ , is it possible to pick a set $S$ with exactly one vertex from the preimage of each $v \in V(H)$ such that the subgraph induced by $S$ is isomorphic to $H$? The homomorphism \(f\) simulates the colors, with all the vertices in a preimage of \(f\) being of equal color (which is unique over all preimages). The \textsc{Exact Weight Colored Subgraph Isomorphism} is defined analogously\footnote{Depending on whether the instance is node-or edge weighted, we require of the solution subgraph that the sum of either its node or its edge weights is zero.}. The weight function is always be denoted by \(w\).
  
Fix an instance $(G,H,f)$ or $(G,H,f,w)$. $H$ is the pattern graph which is to be found in the large graph $G$ when given color homomorphism $f: V(G) \to V(H)$ and weight function $w$. For a subset $I \subseteq V(H)$, we call a function $R: I \to f^{-1}(I)$ a \textbf{configuration of $I$} if $\forall v \in I: R(v) \in f^{-1}(v)$. We define $\mathcal{Conf}(I) \subseteq (I \to f^{-1}(I))$ to be the set of configurations of $I$. A configuration of $I$ is called \textbf{valid configuration} if $\forall uv \in E(H[I]): R(u)R(v) \in E(G)$. Finally, we call a configuration $R$ of $I$ a \textbf{partial solution of $I$ in $J$}, for some $I \subseteq J \subseteq V(H)$, if there is a valid configuration $S$ of $J$ such that $S|_I = R$. We may shorten this to $R$ being a \textbf{partial solution for $J$} if $I$ is clear from context.

For a valid configuration $R$ of $I \subseteq V(H)$, we call $w(R)$ its \textbf{weight}. The exact definition of the weight $w(R)$ depends on whether $G$ is node-weighted or edge-weighted. If $G$ is node-weighted with weight function $w: V(G) \to \mathbb{Z}$, we define $w(R) := \sum_{u \in I} w(f^{-1}(u))$. If it is edge-weighted with weight function $w: E(G) \to \mathbb{Z}$, we define $w(R) := \sum_{uv \in E(H[I])}w(R(u)R(v))$, where it is guaranteed that $R(u)R(v) \in E(G)$ because $R$ is a valid configuration. Furthermore, we say that a partial solution $R$ of $I$ in $J$ \textbf{has an extension of weight $W'$} if there is a valid configuration $S$ of $J$ such that $S|_I = R$ and the nodes (respectively: the edges) of S which are not in $I$ have combined weight $W'$, i.e.\ $w_{ext}(S,R) := w(S) - w(R) = W'$.

For brevity, we further define the following predicates:
\begin{itemize}
\item $\ParSol(R;I;J) =$ $R$ is a partial solution of $I$ in $J$
\item $\ParSolE(R;I;J;W) = $ $R$ is a partial solution of $I$ in $J$ with an extension of weight $W$
\item $\ValConf(R;I) =$ $R$ is a valid configuration of $I$
\end{itemize}

\subsection{Equivalence of the Colored and Uncolored Problems}

\label{sec:equivalence-statement}

All mentioned algorithms and conditional lower bounds are shown for the restricted problem of \textsc{Colored Subgraph Isomorphism}, where the nodes of $G$ and $H$ are colored with $|V(H)|$ colors and the isomorphism must preserve colors, as also studied in~\cite{marx2007can}. In section~\ref{section-equivalence}, we prove that the standard and the colored variant of \textsc{Subgraph Isomorphism} can be solved in essentially the same running time in almost all cases. Specifically, we show the following lemma.

\begin{lemma}\label{lemma-equiv} Let $\rho$ be any graph parameter.
  \begin{enumerate}
  \item If there is a $T(n,k,\rho(H))$ time algorithm for \textsc{Colored Subgraph Isomorphism}, then there is a \(\widetilde{O}(T(k n,k,\rho(H))g(k))\) time algorithm for \textsc{Subgraph Isomorphism}, where $g$ is some computable function.
    \item If there is a \(T(n,k,\rho(H),W)\) time algorithm for \textsc{Exact Weight Colored Subgraph Isomorphism}, then there is a \(\widetilde{O}(T(k n, k, \rho(H),W)g(k))\) time algorithm for \textsc{Exact Weight Subgraph Isomorphism}, where $g$ is some computable function.
  \item Let $\tw(H) \geq 2$. If there is a $T(n,k,\tw(H))$ time algorithm for \textsc{Subgraph Isomorphism}, then there is a $O(T(\poly(k)n,\poly(k),\tw(H)) + \poly(k)n^2)$ time algorithm for \textsc{Colored Subgraph Isomorphism}.
  \item If there is a $T(n,k,\rho(H),W)$ time algorithm for \textsc{Exact Weight Subgraph Isomorphism}, then there is a $O(T(2n,2k,\rho(H),2^kW) + \poly(k)n^2)$ time algorithm for \textsc{Exact Weight Colored Subgraph Isomorphism}.
  \end{enumerate}
\end{lemma}

This lemma enables us to prove results for \textsc{(Exact Weight) Subgraph Isomorphism} while only talking about the more structured colored variants of the problem.

  Regarding treewidth, the only case the above lemma does not cover is how to transform an algorithm for unweighted \textsc{Subgraph Isomorphism} to an algorithm for \textsc{Colored Subgraph Isomorphism} for $\tw(H) = 1$. For our purposes, this is not a problem, since \textsc{Subgraph Isomorphism} for trees already has a trivial unconditional lower bound of $\Omega(n^2)$, which is tight. Note that for this unconditional lower bound, we must assume that the graph is dense, i.e.\ has \(\Theta(n^2)\) edges.

  The lemma also cannot transform algorithms for unweighted \textsc{Subgraph Isomorphism} to algorithms for \textsc{Colored Subgraph Isomorphism} for bounded pathwidth. This means that we cannot show the same lower bounds for the unweighted \textsc{Subgraph Isomorphism} problem for bounded pathwidth as we can for the \textsc{Colored Subgraph Isomorphism} problem. This is a shortcoming of the lemma that we could not fix, and hence we leave it as an open problem whether the lower bounds for the \textsc{Subgraph Isomorphism} problem for bounded pathwidth can even be improved.

\subsection{Treewidth and Pathwidth}\label{prelim-treewidth}

We give a very short introduction to treewidth and pathwidth, and state some auxiliary definitions and notation used throughout the paper. For an thorough introduction to treewidth, pathwidth and their many applications, we refer the reader to~\cite[Chapter 7]{cygan2015parameterized}.
\begin{definition}[Tree Decomposition]
  Let $H$ be a graph. A \textbf{tree decomposition} of $H$ is a pair $\mathcal{T} = (T, \{X_t\}_{t\in V(T)})$ consisting of a tree $T$ and along with a set of ``bags'' $X_t\subseteq V(H)$, one for each vertex of $T$. It must satisfy the following properties:
  \begin{description}
  \item[(T1)] $\bigcup_{t \in V(T)} X_t = V(H)$
  \item[(T2)] $\forall uv \in E(H): \exists t \in V(T): \{u,v\} \subseteq X_t$
  \item[(T3)] $\forall u \in V(H):$ The subgraph induced by $\{t\in V(T) | u \in X_t\}$ is a connected subtree
  \end{description}
\end{definition}
\begin{definition}[Treewidth]
  Let $\mathcal{T} = (T,\{X_t\}_{t\in V(T)})$ be a tree decomposition of $H$. We define its \textbf{width} to be $max_{t\in V(T)}|X_t|-1$. We define the \textbf{treewidth} of $H$ to be the minimum width of all tree decompositions of $H$ and denote it as $\tw(H)$.
\end{definition}
\begin{definition}[Path Decomposition]
Let $H$ be a graph. A \textbf{path decomposition} of $H$ is a tree decomposition where $T$ is a path.
\end{definition}
\begin{definition}[Pathwidth]
  The \textbf{width} of path decompositions is defined as for tree decompositions. The \textbf{pathwidth} of $H$ is defined to be the minimum width of all path decompositions of $H$, and is denoted as $\pw(H)$.
\end{definition}

 A classic algorithm by Bodlaender~\cite{bodlaender1996linear} computes an optimal tree decomposition or path decomposition for an input graph $H$ in time $O(f(|\tw(H)|) |V(H)|)$ for a computable function~$f$. For our purposes, this is almost excessive: For our results, we only need an algorithm which computes an optimal tree decomposition in time $g(|V(H)|)$, for some computable function~$g$.

Clearly, the treewidth of a graph is always smaller than or equal to its pathwidth. It should also be noted that while graphs of treewidth one are exactly the class of tree graphs, graphs of pathwidth one encompass more than just paths. Rather, they are the class of graphs where each connected component is a caterpillar graph, i.e.\ consists of a single path with arbitrarily many degree-one nodes attached at any node of the path~\cite{proskurowski1999classes}. The latter is vital for our conditional lower bounds for pathwidth one.

 For a tree or path decomposition with underlying tree $T$ and for $u \in V(T)$, we define $T_u$ to be the subtree rooted at $u$. The \textbf{cone} $V_u$ is then defined to be $V_u := \bigcup_{v \in V(T_u)}X_v$.

\section{Hardness Results}
\label{section-hardness}

\subsection{Twin Water Lilies}

We will obtain our lower bounds by reducing \textsc{hyperclique} instances to (un)weighted \textsc{Colored Subgraph Isomorphism} instances, where the pattern graph \(H\) is of a special form defined below. In other words, the pattern graphs below are the ``maximally hard'' pattern graphs for the \textsc{Subgraph Isomorphism} problem. See Figure~\ref{fig:twin-water-lily} for an illustration.

\begin{definition}[Twin Water Lily]
  For any \(h,s_1,s_2\), we define the graph \(\TWL(h,s_1,s_2)\) as follows and call it a \textbf{\(h\)-wide Twin Water Lily of order \((s_1, s_2)\)}. The vertex set of \(\TWL(h,s_1,s_2)\) consists of two independent sets \(S_1\) and \(S_2\) with \(r_1\) and \(r_2\) vertices, respectively. Additionally, for every size \(h\) subset \(\{(v_1,s_1), \ldots, (v_h,s_h)) \in {(S_1 \cup S_2)\times [h] \choose h}\), it has a vertex \(v\) which is connected to all vertices in \(\{v_1, \ldots, v_h\}\cap S_2\). We define \(P\) to be the set of all such \(v\) created in this way.

  Note that \(S_1\) consists only of isolated vertices, and that \(H\) is bipartite.
\end{definition}

  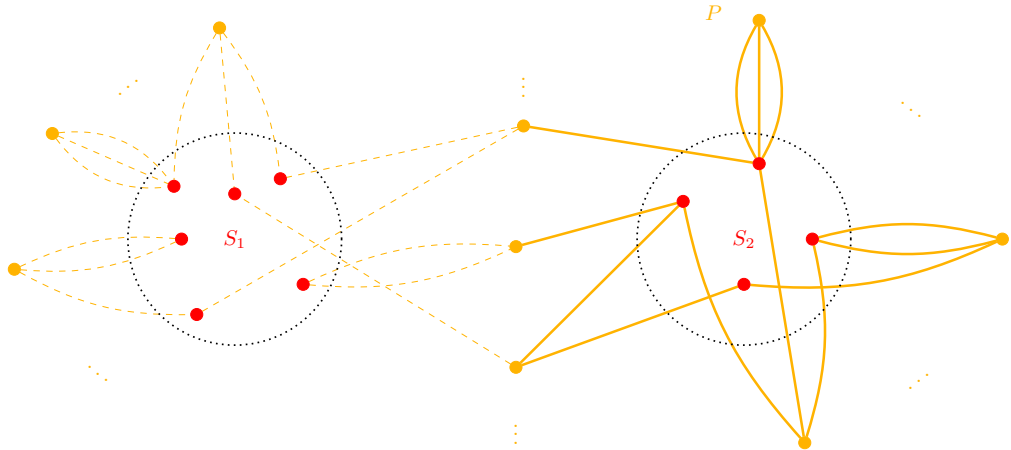
\begin{figure}[h]
    \centering
    \scalebox{0.8}{
      \tikzstyle{every picture}=[tikzfig]
      \begin{tikzpicture}
	\begin{pgfonlayer}{nodelayer}
		\node [style=center node] (0) at (-11.5, 0) {};
		\node [style=center node] (1) at (5.25, 0) {};
		\node [style=Rnodes] (2) at (3.25, 1.25) {};
		\node [style=Rnodes] (3) at (5.75, 2.5) {};
		\node [style=Rnodes] (4) at (5.25, -1.5) {};
		\node [style=Rnodes] (5) at (7.5, 0) {};
		\node [style=Rnodes] (6) at (-13.5, 1.75) {};
		\node [style=Rnodes] (7) at (-10, 2) {};
		\node [style=Rnodes] (8) at (-13.25, 0) {};
		\node [style=Rnodes] (9) at (-11.5, 1.5) {};
		\node [style=Rnodes] (10) at (-12.75, -2.5) {};
		\node [style=Rnodes] (11) at (-9.25, -1.5) {};
		\node [style=Rnodes 2] (12) at (-12, 7) {};
		\node [style=Rnodes 2] (13) at (5.75, 7.25) {};
		\node [style=Rnodes 2] (14) at (13.75, 0) {};
		\node [style=Rnodes 2] (15) at (-2, 3.75) {};
		\node [style=Rnodes 2] (16) at (-2.25, -0.25) {};
		\node [style=Rnodes 2] (17) at (-17.5, 3.5) {};
		\node [style=Rnodes 2] (18) at (-2.25, -4.25) {};
		\node [style=Rnodes 2] (19) at (7.25, -6.75) {};
		\node [style=Rnodes 2] (20) at (-18.75, -1) {};
		\node [style=none] (21) at (5.25, 0) {\textcolor{red}{$S_2$}};
		\node [style=none] (22) at (-11.5, 0) {\textcolor{red}{$S_1$}};
		\node [style=none] (23) at (4.25, 7.5) {\textcolor{rgb,255: red,255; green,179; blue,0}{$P$}};
		\node [style=none] (24) at (-2, 5.25) {\textcolor{rgb,255: red,255; green,179; blue,0}{$\vdots$}};
		\node [style=none] (25) at (-2.25, -6.25) {\textcolor{rgb,255: red,255; green,179; blue,0}{$\vdots$}};
		\node [style=none] (26) at (11, -4.5) {\textcolor{rgb,255: red,255; green,179; blue,0}{\reflectbox{$\ddots$}}};
		\node [style=none] (27) at (10.75, 4.5) {\textcolor{rgb,255: red,255; green,179; blue,0}{$\ddots$}};
		\node [style=none] (28) at (-16, -4.25) {\textcolor{rgb,255: red,255; green,179; blue,0}{$\ddots$}};
		\node [style=none] (29) at (-15, 5.25) {\textcolor{rgb,255: red,255; green,179; blue,0}{\reflectbox{$\ddots$}}};
	\end{pgfonlayer}
	\begin{pgfonlayer}{edgelayer}
		\draw [style=dashed Redge, bend right=15] (12) to (6);
		\draw [style=dashed Redge, bend left=15] (12) to (7);
		\draw [style=dashed Redge] (12) to (9);
		\draw [style=dashed Redge] (17) to (6);
		\draw [style=dashed Redge, bend right] (17) to (6);
		\draw [style=dashed Redge, bend left] (17) to (6);
		\draw [style=dashed Redge, bend left=15] (20) to (8);
		\draw [style=dashed Redge, bend right=15] (20) to (8);
		\draw [style=dashed Redge, bend right=15] (20) to (10);
		\draw [style=dashed Redge, bend right=15] (11) to (16);
		\draw [style=dashed Redge, bend left=15] (11) to (16);
		\draw [style=R-edge] (16) to (2);
		\draw [style=R-edge] (18) to (4);
		\draw [style=R-edge] (18) to (2);
		\draw [style=R-edge] (15) to (3);
		\draw [style=dashed Redge] (15) to (7);
		\draw [style=dashed Redge] (15) to (10);
		\draw [style=dashed Redge] (18) to (9);
		\draw [style=R-edge, bend left] (13) to (3);
		\draw [style=R-edge, bend right] (13) to (3);
		\draw [style=R-edge, bend right=15] (14) to (5);
		\draw [style=R-edge, bend left=15] (14) to (4);
		\draw [style=R-edge, bend right=15] (19) to (5);
		\draw [style=R-edge, bend left=15] (19) to (2);
		\draw [style=R-edge] (19) to (3);
		\draw [style=R-edge, bend right=15] (5) to (14);
		\draw [style=R-edge] (3) to (13);
	\end{pgfonlayer}
\end{tikzpicture}
    }
    \caption{Partial sketch of the \(3\)-wide Twin Water Lily of order \((6, 4)\). The dashed edges are not actual edges, they just represent which other vertices are in the set \(\{v_1, \ldots, v_h\}\) of the vertex \(u \in P\) that they are connected to.}
    \label{fig:twin-water-lily}
  \end{figure}

\begin{proposition}\label{prop:treewidth-of-lily}
  If a graph \(H\) is a \(h\)-wide Twin Water Lily of order \((s_1,s_2)\), then its treewidth is bounded by
  \begin{align*}
    \tw(H) \leq    \begin{cases} s_2-1 & \text{ if } s_2 > h\\ s_2 & \text{ otherwise}\end{cases}
  \end{align*}
  and its pathwidth is bounded by $\pw(H) \leq s_2$.
\end{proposition}

\begin{proof}
  There is a very useful characterization of treewidth using a graph-theoretic game: A graph \(G\) has treewidth \(\leq k\) if and only if \(k+1\) cops can catch\footnote{See the corresponding footnote in Section~\ref{section-sketches} for a description of the game.} a visible robber on \(G\)~\cite{seymour1993graph}.

  To show the bound on the treewidth of \(H\), simply place the cops on all vertices of \(S_2\). No matter where the robber starts, it is surrounded by cops or is on an isolated vertex. If \(s_2 > h\), then there must exist some cop which is not adjacent to the robber, whom we can use to catch him in a single step. If \(s_2 \leq h\), it is not guaranteed that there is a non-adjacent cop. However, we have an additional cop which can start at any vertex. As soon as the robber is positioned, we use the additional cop to capture them.

    A similar characterization exists for pathwidth: A graph \(G\) has pathwidth \(\leq k\) if and only if \(k+1\) cops can catch\footnote{The game can be formulated such that it is the same as the one for treewidth, but the cops simply cannot see the robber and must therefore have a universal strategy for catching him on \(G\). This is also sometimes referred to as the contamination cleansing or infection cleansing game.} an invisible robber on \(G\)~\cite{ellis1994vertex}.

  For the pathwidth of \(H\), we again place all cops on the vertices of \(S_2\) and have one left over. We use this cop to go through all vertices not in \(S_2\), one in each step. The robber is always surrounded by the cops in \(S_2\) an hence cannot move, so after going through all vertices with the additional cop, we must have caught him.
\end{proof}

\subsection{Unweighted Colored Subgraph Isomorphism}\label{sec:unweighted-uncol-subiso}

\begin{lemma}[The Unweighted Lemma]\label{lem:unweighted-lemma}
  For any \(h\geq 3\) and any \(r\geq 2\), if we can solve \textsc{Colored Subgraph Isomorphism} with pattern graph \(\TWL(h,0,r)\) in time \(O(N^{r-\varepsilon})\), then we can solve \(h\)-uniform \(hr\)-\textsc{hyperclique} in time \(O(n^{hr-\varepsilon})\).
\end{lemma}

\begin{proof}
  Let an instance \(\mathcal{I}_0 = G_0\) of \(h\)-uniform \(hr\)-\textsc{Hyperclique} be given. We convert this to a \textsc{Colored Subgraph Isomorphism} instance with a Twin Water Lily as pattern graph in three steps, each of which we explain in detail below: First, we convert it to a \textsc{Colored \(hr\)-Hyperclique} instance in a standard way. Second, we go from hypercliques to \textsc{Colored Subgraph Isomorphism} by replacing each hyperedge by an intermediate vertex. Finally, we merge preimages to homogenize preimage sizes. 


    \begin{description}
  \item[\underline{1. Converting to Colored Hyperclique:}]\label{step1-unweighted-lemma} We convert \(\mathcal{I}_0\) to a \(h\)-uniform \textsc{Colored Hyperclique} instance with \(hr\) colors. The converted instance should have a hyperclique where all vertices has different colors if and only if the old instance has a hyperclique. The new instance have the form \(\mathcal{I}_1 = (G_1, f_1)\), where the color homomorphism \(f_1: V(G_1) \to V(C_{hr})\) assigns each vertex of \(G_1\) a vertex in the \(h\)-uniform \(hr\)-hyperclique \(C_{hr}\).

    Let \(V(C_{hr}) = \{1,\ldots, hr\}\). For each \(i\), the preimage \(f_1^{-1}(i)\) is a copy of \(V(G_0)\). Let \(g: V(G_1) \to V(G_0)\) be a function between sets that indicates which vertex in \(G_0\) the vertex in \(G_1\) is a copy of. Now for each set \(\{v_1, \ldots, v_h\} \in {V(C_{hr})\choose h}\), we go through all tuples \((w_1, \ldots, w_h) \in f_1^{-1}(v_1)\times \ldots \times f_1^{-1}(v_h)\) and create the edge \(\{w_1, \ldots, w_h\} \in E(G_1)\) if and only if \(\{g(w_1), \ldots, g(w_h)\} \in E(G_0)\).
    
    The correctness of this construction is easy to see. 

  \item[\underline{2. Representing Hyperedges by Intermediate Vertices:}]\label{step2-unweighted-lemma} We now go from the \textsc{Colored Hyperclique} instance \(\mathcal{I_1} = (G_1, f_1)\) to a (structured) \textsc{Colored Subgraph Isomorphism} instance \(\mathcal{I_2} = (H_2, G_2, f_2)\). The reduction is done in a standard way: We replace each hyperedge with a vertex connected to all its endpoints.

    Formally, we need to construct \(H_2\) and \(G_2\). \(H_2\) has two sets of vertices \(S'_2\) and \(P\). \(S'_2\) is a copy of \(V(C_{hr})\) from the last step, including its preimages. Accordingly, we write \(S'_2 = \{1, \ldots, hr\}\). In \(P\), we have one vertex \(u\) for every subset \(\{w_1, \ldots, w_h\} \in {C_{hr} \choose h}\), and we have \(\forall \ell \in [h]: uw_\ell \in E(H_2)\). Now for every hyperedge \(\{w'_1, \ldots, w'_h\} \in E(G_1)\) with \(\forall \ell \in [h]: w'_\ell \in f_1^{-1}(G_1)\), we add a vertex \(u' \in f_2^{-1}(u)\) which is connected to all vertices \(w'_1, \ldots, w'_h\). This concludes the construction of \(\mathcal{I}_2\).

    Correctness of this construction is again easy to see. As for the size, note that preimages of vertices in \(S'_2\) still have size \(n\). The preimages of vertices in \(P\), however, have at most \(n^h\) vertices.

  \item[\underline{3. Merging Preimages in \(S'_2\):}]\label{step3-unweighted-lemma} Lastly, we go from \(\mathcal{I}_2 = (H_2, G_2, f_2)\) to the final instance \(\mathcal{I}_3 = (H_3, G_3, f_3)\) where \(H_3\) is a Twin Water Lily of order \((0,r)\). Note that \(\mathcal{I}_2\) is ``almost'' the instance we want, save for the fact that the preimages \(P\) are much larger (size up to \(n^h\)) than the preimages of \(S'_2\) (size \(n\)). We rectify this by merging groups of vertices within \(S'_2\). These groups have size \(h\).

    We split \(S'_2\) into \(r\) groups \(X_{1}, \ldots, X_{r}\) of size \(h\). In \(H_3\), we have for every \(i \in [r]\) a vertex \(x_i\) representing \(X_i\). Each vertex \(x'_i \in f_3^{-1}(x_i)\) corresponds to a configuration \(\mathrm{conf}(x'_i) \in \mathcal{Conf}(X_i)\). In accordance with the definition of a Twin Water Lily, we define \(S_1 = \emptyset\) and \(S_2 = \{x_{1}, \ldots, x_{r}\}\).

    The set \(P \subseteq V(H_3)\) remains the same as in the preceding step, including its preimages. For each \(u \in P\), we connect \(u\) to all \(x_i\) such that \(X_i \cap N_{H_2}(u) \neq \emptyset\) (recall that \(N_{H_2}(u)\) is defined to be the neighbourhood of \(u\) in \(H_2\)). Note that each \(u \in P\) is still connected to at most \(h\) other vertices, but that it can be less if multiple vertices of its neighborhood came from the same group. Finally, for an edge \(ux_i \in E(H_3)\), we connect a vertex \(u' \in f_3^{-1}(u)\) to a vertex \(x'_i \in f_3^{-1}(x_i)\) if and only if \(\forall v \in N_{H_2}(u)\cap X_i: u'(\mathrm{conf}(x'_i)(v)) \in E(G_2)\).

    Correctness is easy to see. Note that \(H_3\) is an \(h\)-wide Twin Water Lily of order \((0,r)\) now, and all preimages are of size at most \(N := n^h\).
  \end{description}

  This completes the construction of the reduction algorithm. Each of the steps runs in \(O(n^{2h-1})\) time: In the intermediate-vertex step, each of the vertices \(u'\) in a preimage \(f^{-1}(u)\) represents a hyperedge and is hence only connected to \(h\) vertices. After the next step, it is connected to all vertices which represent compatible configurations. Each vertex of \(S_2\) that \(u\) is now connected to must represent a non-empty intersection of the neighbourhood of \(u\). Hence it has cardinality at least \(1\). Hence there are at most \(n^{h-1}\) compatible vertices in its preimage. Hence overall, there are at most \(O(n^{2h-1})\) edges, and the graph can also be constructed in this time.

  Now suppose there is an algorithm solving the \textsc{Colored Subgraph Isomorphism} problem with pattern \(\TWL(h,0,r)\) in time \(O(N^{r-\varepsilon})\). Then for any \(h\)-uniform \(hr\)-\textsc{Hyperclique} instance, we run the reduction in time \(O(n^{2h-1})\) and then solve the new instance in time \(O(N^{r-\varepsilon}) = O(n^{hr-h\varepsilon})\). Since \(r \geq 2\), the algorithm hence takes total time \(O(n^{hr-h\varepsilon}) = O(n^{hr-\varepsilon'})\).
\end{proof}

We now use the~\hyperref[lem:unweighted-lemma]{Unweighted Lemma} to prove the lower bounds for \textsc{Colored Subgraph Isomorphism}. In particular, we prove the following theorem, which implies Theorem~\ref{corollary-lower-bound-detection} from the Results section via the Equivalence Lemma (Lemma~\ref{lemma-equiv}).
  
\begin{theorem}\label{lower-bound-detection}
  The following statements are true.
  \begin{enumerate} 
  \item For each \(t\geq 3\) and any \(3\leq h \leq t\), there exists a connected, bipartite pattern graph \(\mathcal{H}_{t,h}\) of treewidth \(t\) such that there cannot be an algorithm solving the \textsc{Colored Subgraph Isomorphism} problem on pattern graph \(\mathcal{H}_{t,h}\) in time \(O(n^{t+1-\varepsilon})\) unless the \(h\)-uniform \(h(t+1)\)-\textsc{hyperclique} hypothesis fails.
  \item For each \(t\geq 2\) and any \(h \geq 3\), there exists a connected, bipartite pattern graph \(\mathcal{H}_{t,h}\) of treewidth \(t\) such that there cannot be an algorithm solving the \textsc{Colored Subgraph Isomorphism} problem on pattern graph \(\mathcal{H}_{t,h}\) in time \(O(n^{t-\varepsilon})\) unless the \(h\)-uniform \(ht\)-\textsc{hyperclique} hypothesis fails.
  \item For each \(t \geq 2\), there exists a connected, bipartite pattern graph \(\mathcal{H}_{t}\) of treewidth \(t\) such that there cannot be an algorithm solving the \textsc{Colored Subgraph Isomorphism} problem on pattern graph \(\mathcal{H}_{t}\) in time \(O(n^{(t+1)\omega/3})\) unless the \((t+1)\)-\textsc{Clique} hypothesis fails.
  \end{enumerate}
  
\end{theorem}

\begin{theorem}[Theorem~\ref{lower-bound-detection} for pathwidth]\label{lower-bound-detection-pathwidth}
  Part 2 of Theorem~\ref{lower-bound-detection} also holds when replacing the treewidth \(t\) by the pathwidth \(p\). Part 3 only holds when replacing \(t+1\) by \(p+1\), and the pattern graph is not bipartite anymore. Part 1 does not hold.
\end{theorem}

Unfortunately, Lemma~\ref{lemma-equiv} cannot be applied to Theorem~\ref{lower-bound-detection-pathwidth}, and hence we have no lower bounds for the uncolored, unweighted case of bounded pathwidth. We believe it is unlikely that the techniques used to prove~\ref{lemma-equiv} generalize to pathwidth. 

We now prove the theorems above.
\begin{proof}[Proof (of Theorem~\ref{lower-bound-detection})]
  For the proof of this theorem, we use the~\hyperref[lem:unweighted-lemma]{Unweighted Lemma}.

  \textbf{Part 1:} Let \(t \geq 3\) and \(3\leq h \leq t\) be given. It suffices to apply the~\hyperref[lem:unweighted-lemma]{Unweighted Lemma} with \(h' := h\) and \(r := t+1\), and set \(\mathcal{H}_{t,h} = \TWL(h,0,t+1)\). As proven in Proposition~\ref{prop:treewidth-of-lily}, \(\TWL(h,0,t+1)\) has treewidth \(t\) since \(t+1 > h\).

  \textbf{Part 2:} Let \(t \geq 3\) and \(h \geq 3\) be given. It suffices to apply the ~\hyperref[lem:unweighted-lemma]{Unweighted Lemma} with \(h' := h\) and \(r := t\), and set \(\mathcal{H}_{t,h} = \TWL(h,0,t+1)\). The loss of the +1 in the exponent is due to the weaker bound in in Proposition~\ref{prop:treewidth-of-lily} for \(t \leq h\).

  \textbf{Part 3:} Let \(t \geq 2\) be given. We know that by the \(t+1\)-\textsc{clique} hypothesis, \textsc{Colored Subgraph Isomorphism} on pattern graph \(C_{t+1}\) cannot be solved in time \(O(n^{(t+1)\omega/3})\). We can make \(C_{t+1}\) bipartite in the obvious way by subdividing the edges. This subdivided graph is \(\mathcal{H}_t\).
  \end{proof}

  \begin{proof}[Proof (of Theorem~\ref{lower-bound-detection-pathwidth})]
    Completely analogous. The deviation in bounds with respect to Theorem~\ref{lower-bound-detection} is due to the difference in bounds for treewidth and pathwidth in Proposition~\ref{prop:treewidth-of-lily}.

    In part 3, we cannot subdivide the edges of \(\mathcal{H}_t\) without changing the pathwidth, hence the pattern graph stays a clique and is thus not bipartite.
  \end{proof}
  
\noindent
Indeed, the Unweighted Lemma can also be used to prove the results under SETH. SETH is beyond doubt the most widely used for conditional lower bounds for problems in P, which is why the following results are still interesting, even though they only give smaller lower bounds than the results under the \textsc{Hyperclique} hypothesis. For context on SETH and the many conditional lower bounds it enables, see e.g.~\cite{abboud2018more, abboud2014consequences, abboud2014popular, abboud2019seth, bringmann2014walking, bringmann2015quadratic, roditty2013fast}.

The following theorem implies Theorem~\ref{corollary-lower-bound-detection-seth} via the Equivalence Lemma (Lemma~\ref{lemma-equiv}).
  \begin{theorem}
  Assuming SETH, the following two statements are true.
  \begin{enumerate}
  \item For any \(t\geq 3\) and any \(\varepsilon > 0\) there exists a pattern graph \(\mathcal{H_{t,\varepsilon}}\) of treewidth \(t\) such that there cannot be an algorithm solving all instances of \textsc{Colored Subgraph Isomorphism} with pattern graph \(\mathcal{H}_{t,\varepsilon}\) in time \(O(n^{t-\varepsilon})\).
    \item For any \(\varepsilon > 0\) there exists a \(t\geq 3\) and a pattern graph \(\mathcal{H}_{\varepsilon}\) of treewidth \(t\) such that there cannot be an algorithm solving all instances of \textsc{Colored Subgraph Isomorphism} with pattern graph \(\mathcal{H}_{\varepsilon}\) in time \(O(n^{t+1-\varepsilon})\).
\end{enumerate}
\end{theorem}

\begin{proof}
  This result also follows via Unweighted Lemma and hence via the \textsc{hyperclique} problem. Specifically, the split-and-list technique from~\cite{williams2005new} allows one to reduce the Satisfiability problem to \textsc{Hyperclique}. Along the same lines, the following result was shown in~\cite[Lemma 9.1]{lincoln2018tight}.
  \begin{lemma}[\cite{lincoln2018tight}]\label{seth-to-hyperclique}
    Assuming SETH, for any \(\varepsilon>0\) there exists \(h\geq3\) such that for all \(k>h\), the \(h\)-uniform \(k\)-\textsc{Hyperclique} problem is not in time \(O(n^{k-\varepsilon})\).
  \end{lemma}
  Using this, we now prove parts 1 and 2 of the theorem.
  
  \textbf{Part 1:} Let \(t\geq 3\) and \(\varepsilon > 0\) be given. We use Lemma~\ref{seth-to-hyperclique} to obtain \(h \geq 3\) such that \(h\)-uniform \(ht\)-\textsc{Hyperclique} is not in time \(O(n^{ht-\varepsilon})\). W.l.o.g. assume \(h \geq t\). Now it suffices to apply the Unweighted Lemma with \(h' := h\) and \(r := t\), and set \(\mathcal{H}_{t,\varepsilon} = \TWL(h,0,t)\). Via Proposition~\ref{prop:treewidth-of-lily} and the fact that \(h \geq t\), we know that \(\TWL(h,0,t)\) has treewidth exactly \(t\).
  
  \textbf{Part 2:} Let \(\varepsilon > 0\) be given. Lemma~\ref{seth-to-hyperclique} gives a \(h\geq 3\) such for all \(k>h\), \(h\)-uniform \(k\)-\textsc{Hyperclique} is not in time \(O(n^{k-\varepsilon})\). We choose \(t := h\) and get that \(h\)-uniform \(ht\)-\textsc{Hyperclique} is not in time \(O(n^{ht-\varepsilon})\). Now if suffices to apply the Unweighted Lemma with \(h' := h\) and \(r := t+1 = h+1\) and set \(H_{t,\varepsilon} = \TWL(h,0,h+1)\). Via Proposition~\ref{prop:treewidth-of-lily}, we know that \(\TWL(h,0,t+1)\) has treewidth exactly \(h = t\).
\end{proof}

\subsection{Exact Weight Colored Subgraph Isomorphism}

First, we state the Weighted Lemma. This result enables us to prove lower bounds for \textsc{Exact Weight Colored Subgraph Isomorphism}.

\begin{lemma}[The Weighted Lemma]\label{lem:main-lemma}
  For any \(\varepsilon \in (0,1)\) and any constant parameters \(h \in \mathbb{N}\setminus \{1\}, r_1 \in \mathbb{N}, r_2 \in \mathbb{N}, \beta \in (0,1)\cap \mathbb{Q}\), there exists a \(k \in \mathbb{N}\) and an algorithm \(\mathcal{A}\) which
  \begin{enumerate}[(a)]
  \item accepts as input an instance \(\mathcal{I} = G\) of \(h\)-uniform \(k\)-\textsc{Hyperclique}.
  \item produces an equivalent instance \(\mathcal{I}' = (H',G',f',w')\) of \textsc{Exact Weight Colored Subgraph Isomorphism}, where \(H'\) is a \(h\)-wide Twin Water Lily of order \((r_1,r_2)\). The preimages of \(\mathcal{I}'\) have size at most \(\max\{n^{\beta k/r_1}, n^{(1-\beta)k/r_2}\}\), and the maximum weight is \(W = \Theta(n^{(1+\varepsilon)\beta k})\).
  \item runs in time \(O(n^{2h-1} + (n^{\beta k/(hr_1)})^{\hat{c}})\) for some universal constant \(\hat{c}\in \mathbb{N}\).
  \end{enumerate}
\end{lemma}

Intuitively, the parameter \(\beta\) indicates what percentage of the instance \(\mathcal{I}\) should be encoded in which part of the Twin Water Lily. A percentage of \(\beta\) is encoded in the weights, while the remaining percentage of \((1-\beta)\) is encoded in the edges.

We now prove the~\hyperref[lem:main-lemma]{Weighted Lemma}.

\begin{proof}
  Let an instance \(\mathcal{I_0}=G_0\) of \(h\)-uniform \(k\)-\textsc{Hyperclique} be given, where \(k\) is chosen later. We convert this to an \textsc{Exact Weight Colored Subgraph Isomorphism} instance with a Twin Water Lily as pattern graph in five steps, each of which we explain in detail below: First, we convert it to a \textsc{Colored \(k\)-Hyperclique} instance in a standard way. Second, we split the instance into the part that we want to encode in the weights and the part that we want to encode in the edges. In both of these parts, we merge large groups of preimages such that we are left with only \(hr_1\) in the weight part, and \(hr_2\) in the edges part. Third, we go from hypercliques to \textsc{Colored Subgraph Isomorphism} in a standard way while preserving the preimages. Fourth, we convert the weight part of the instance into actually using weights by replacing edge constraints by weight constraints, using a construction known as \(k\)-average free sets. Finally, we merge preimages in both parts again to obtain the final Twin Water Lily instance.
  \begin{description}
  \item[\underline{1. Converting to Colored Hyperclique:}] This step works exactly like~\hyperref[step1-unweighted-lemma]{step 1} in the proof of the~\hyperref[lem:unweighted-lemma]{Unweighted Lemma}. As described there, We convert \(\mathcal{I}_0\) to a \(h\)-uniform \textsc{Colored Hyperclique} instance. The latter has the form \(\mathcal{I}_1 = (G_1, f_1)\), where the color homomorphism \(f_1: V(G_1) \to V(C_k)\) assigns each vertex of \(G_1\) a vertex in the \(k\)-hyperclique \(C_k\).

  \item[\underline{2. Merging Preimages:}] We now convert \(\mathcal{I}_1\) to a \(h\)-uniform \textsc{Colored \((hr_1+hr_2)\)-Hyperclique} instance, by condensing groups of (small) preimages into single (large) preimages. In the converted instance \(\mathcal{I}_2 = (G_2,f_2)\) with color homomorphism \(f_2: V(G_2)\to V(C_{r_1+r_2})\), we ensure that \(V(C_{hr_1+hr_2})\) can be divided into two sets \(H_1,H_2\) such that \(|H_1| = hr_1, |H_2| = hr_2\) and \(\forall v \in H_1: |f_2^{-1}(v)| = n^{\beta k/(hr_1)}, \forall v \in H_2: |f_2^{-1}(v)| = n^{(1-\beta)k/(hr_2)}\). These two sets correspond to the two water lilies constructed in later steps.

    At this point, we must make our choice of \(k\). We will need that \(k_1 := \beta k\) is an integer and divisible by \(hr_1\). Furthermore, \(k_2 := (1-\beta)k\) must also be an integer and divisible by \(hr_2\). Letting \(\beta = \frac{p}{q} \in (0,1) \cap \mathbb{Q}\) where \(p,q \in \mathbb{N}\), it hence suffices to choose \(k = hr_1r_2q\).

    Now, we split \(V(C_k) = \{1, \ldots, k\}\) into two groups \(V_1 = \{1, \ldots, k_1\}\) and \(V_2 = \{k_1+1, \ldots, k_2\}\). Going further, we split each of these sets again: \(V_1\) is split into \(hr_1\) disjoint groups \(V_{1}, \ldots, V_{hr_1}\) of size \(\frac{k_1}{hr_1}\) each. Analogously, we split \(V_2\) into \(hr_2\) disjoint groups \(V_{hr_1+1}, \ldots, V_{hr_1+hr_2}\) of size \(\frac{k_2}{hr_2}\) each.

    Now define \(V(C_{hr_1+hr_2}) = \{v_{1},\ldots, v_{hr_1+hr_2}\}\) and furthermore let \(S'_1 = \{v_1, \ldots, v_{hr_1}\}, S'_2 = \{v_{hr_1+1}, \ldots, v_{hr_1+hr_2}\}\). The vertex \(v_{i}\) corresponds to \(V_{i}\), for every \(i\). The vertices in their preimages represent all the configurations of the sets. In particular, we let every \(v'_i \in f_2^{-1}(v_{i})\) represent a configuration \(\mathrm{conf}(v'_i) \in \mathcal{Conf}(V_{i})\). Hence, for \(v_i \in S'_1\) the preimage \(f_2^{-1}(v_{i})\) has size \(n^{k_1/(hr_1)}\), and for \(v_i \in S'_2\) the preimage \(f_2^{-1}(v_i)\) has size \(n^{k_2/(hr_2)}\).

    Now for the edges. For every subset \(\{v_{i_1}, \ldots, v_{i_h}\} \in {S'_1 \cup S'_2 \choose h}\), we iterate over all \((v'_{i_1}, \ldots, v'_{i_h}) \in f_2^{-1}(v_{i_1})\times \ldots, \times f_2^{-1}(v_{i_h})\). We combine the configurations that they represent by defining \(R \in \mathcal{Conf}(V_{i_1} \cup \ldots V_{i_h})\) as \(\forall \ell \in [h]: \forall v'_{i_\ell} \in f_2^{-1}(v_{i_\ell}): R(v'_{i_\ell}) := \mathrm{conf}(v'_{i_\ell})\). Now we add \(\{w_1, \ldots, w_h\}\) as a hyperedge to \(E(G_2)\) if and only if \(R\) is a valid configuration. That is, if the image of \(R\) induces a hyperclique in \(G_1\).

    Again, correctness is easy to see. 

  \item[\underline{3. Representing Hyperedges by Intermediate Vertices:}] We now go from the \textsc{Colored Hyperclique} instance \(\mathcal{I_2} = (G_2, f_2)\) to a (structured) \textsc{Colored Subgraph Isomorphism} instance \(\mathcal{I_3} = (H_3, G_3, f_3)\). The reduction is essentially the same as the one in~\hyperref[step2-unweighted-lemma]{step 2} in the proof of the~\hyperref[lem:unweighted-lemma]{Unweighted Lemma}.

    We construct \(H_3\) and \(G_3\). \(H_3\) has three sets of vertices \(S'_1, S'_2\) and \(P\). \(S'_1\) and \(S'_2\) copy \(S'_1\) and \(S'_2\) from the last step, including their preimages. Accordingly, we write \(S'_1 = \{v_{1}, \ldots, v_{hr_1}\}\) and \(S'_2 = \{v_{hr_1+1}, \ldots, v_{hr_1+hr_2}\}\). In \(P\), we have one vertex \(u\) for every subset \(\{w_1, \ldots, w_h\} \in {C_{hr_1+hr_2} \choose h}\), and we have \(\forall \ell \in [h]: uw_\ell \in E(H_3)\). For every hyperedge \(\{w'_1, \ldots, w'_h\} \in E(G_2)\) with \(\forall \ell \in [h]: w'_\ell \in f_2^{-1}(G_2)\), we add a vertex \(u' \in f_3^{-1}(u)\) which is connected to all vertices \(w'_1, \ldots, w'_h\). Note here that \(u'\) is only connected to one vertex from each preimage, which is a property that will be needed in the fourth step. 

    Correctness of this construction is easy to see. Note that preimages of vertices in \(S'_1\) and \(S'_2\) still have size \(n^{k_1/(hr_1)}\) and \(n^{k_2/(hr_2)}\), respectively. The preimages of vertices in \(P\), have at most \((\max\{n^{k_1/(hr_1)}, n^{k_2/(hr_2)}\})^h = \max\{n^{k_1/r_1}, n^{k_2/r_2}\}\) vertices.

  \item[\underline{4. Replacing Some of the Edges with Weights}] We now come to the crucial step of converting some of the edge constraints to weight constraints. We convert the \textsc{Colored Subgraph Isomorphism} instance \(\mathcal{I}_3 = (H_3, G_3, f_3)\) of the preceding step into an \textsc{Exact Weight Colored Subgraph Isomorphism} instance \(\mathcal{I}_4 = (H_4, G_4, f_4, w_4)\). To do this, we will need so-called \(k\)-average free sets.

    \begin{definition}[$k$-average free sets]
    A set \(S\subseteq \mathbb{Z}\) is called $k$-average-free if, for any $s_1, \ldots, s_{k'+1} \in S$ with $k' \leq k$, we have $s_1+ \ldots + s_{k'} = k'\cdot s_{k'+1}$ if and only if $s_1 = \ldots = s_{k'+1}$. In other words, the average of $s_1, \ldots, s_{k'}\in S$ is in $S$ if and only if all $s_i$ are equal.  \end{definition}

We use the following construction for $k$-average free sets, originally proven in~\cite{behrend1946sets}, modified into a more useful version in~\cite{abboud2014losing} and formulated in this form in~\cite{abboud2019seth}.

\begin{lemma}\label{k-average-free-lemma}
    There exists a universal constant $c > 0$ such that, for all constants $\varepsilon \in (0,1)$ and $k \geq 2$, a $k$-average-free set $S$ of size $n$ with $S \subseteq [0, k^{c/\varepsilon}n^{1+\varepsilon}]$ can be constructed in time $\poly(n)$.
  \end{lemma}

  Specifically, we use Lemma~\ref{k-average-free-lemma} with \(\varepsilon' = \varepsilon, k' = \lambda := |P|\) and \(n' = n^{k_1/(hr_1)}\). Letting \(B := \lambda^{c/\varepsilon}n^{(1+\varepsilon)k_1/(hr_1)}\), this yields a \(\lambda\)-average free set \(S \subseteq [0,B]\) of size \(n^{k_1/(hr_1)}\). From this, we can construct  an arbitrary bijection \(\varrho_S: [n^{k_1/(hr_1)}] \to S\).
  
  Now, to construct \(\mathcal{I}_4\), we first copy \(\mathcal{I_3}\), giving each node a default weight of ``infinity'' (i.e.\ something otherwise unobtainable, e.g. \(k^2W + 1\)). Now we delete all edges in \(H_3\) which are incident to a vertex in \(S'_1\), along with the corresponding edges in \(G_3\). These are the edges that we replace by weight constraints.

  Hence we now describe the weights. To make our construction easier, we specify a target value \(T\) (instead of the default target zero). We can easily get rid of this again by picking some vertex \(\hat{w} \in V(H_3)\) and subtracting \(T\) from the weights of all of its preimages. The binary representation of \(T\) consists of \(h\cdot a\) blocks of \(\lceil\log(2\lambda B)\rceil\) bits, each containing the binary representation of \(\lambda B\). The \(i\)-th block represents the vertex \(v_i \in S'_1\).

  We move to the weights of the vertices, starting with vertices in the preimages of \(S'_1\). Somewhat abusing notation, we define \(\forall i \in [hr_1]: f_4^{-1}(v_i) = \{1, \ldots, n^{k_1/(hr_1)}\}\). The weight of vertex \(v'_i \in f_4^{-1}(v_i)\) has a value of \(\lambda B - |N(v_i)|\cdot \varrho_S(v'_i)\) in the \(i\)-th block, and a value of zero in all other blocks. Now for vertices in the preimages of\(P\). Let \(u \in P\) correspond to the set \(\{w_1, \ldots, w_h\} \in {C_{hr_1+hr_2} \choose h}\). As observed in the preceding step, each \(u' \in f_3^{-1}(u)\) is connected to exactly one vertex \(w'_i\) from each preimage \(f_3^{-1}(w_i)\). In the current step, for each \(i\in [h]\) with \(w_i \in S'_1\), we have deleted the edges \(u'w'_i\). To replace them, for each such \(i\), we give \(u'\) a value of \(\varrho_S(w'_i)\) in the \(i\)-th block. We have just changed the weight of \(u'\) at \(|N(u) \cap S'_1|\) blocks of its binary representation. All other blocks have a value of zero.

  All vertices with so far unspecified weight have weight zero. This concludes the construction of \(\mathcal{I}_4\).

  We show correctness of this construction. It suffices to show that any configuration \(R\) that is a solution for \(\mathcal{I}_3\) is also a solution for \(\mathcal{I}_4\) and vice versa. Hence, suppose \(R\) is a solution for \(\mathcal{I}_3\). Then all edge constraints of \(H_4\) are trivially fulfilled and we need only show that the total weight is \(T\). For ease of discussion, we denote by \(\alpha[i]\) the value of the \(i\)-th block of a weight \(\alpha\). Consider the blocks of the binary representation of the sum of weights \(w(R) = \sum_{v \in \mathrm{Im}(R)}w_4(v)\), and let \(i\) be fixed. The large block size prevents overflow, so \(w(R)[i] = \sum_{v \in \mathrm{Im}(R)}w_4(v)[i]\). By construction, we have \(w(R)[i] = w(R(v_i))[i] + \sum_{u \in N(v_i)} w(R(u))[i]\). However since \(R\) is a valid configuration in \(\mathcal{I}_3\), we have that for each \(u\), \(\forall v_j \in N(u): R(u)R(v_j) \in E(G_3)\). In particular, \(R(u)R(v_i) \in E(G_3)\) and hence by construction \(w(R(u))[i] = \varrho_S(v_i)\). We conclude \(w(R)[i] = \lambda B - |N(v_i)|\cdot \varrho_S(R(v_i)) + \sum_{i \in N(v_i)} \varrho_S(R(v_i)) = \lambda B = T[i]\). Hence \(w(R)\) is equal to \(T\) in each of its blocks, which was to be proven.
  \todo{add stuff where the infinity was necessary (i.e.\ that you can't choose non-active vertices)}

  Conversely, suppose \(R\) is a solution for \(\mathcal{I}_4\). Then all edge constraints in \(G_3[V(G_3)\setminus S'_1]\) are trivially satisfied and we need only show that \(\forall v_i \in S'_1: \forall u \in N(v_i): R(v_i)R(u) \in E(G_3)\). Fix \(v_i \in S'_1\). We have that \(\lambda B = T[i] = w(R)[i] = w(R(v_i))[i] + \sum_{u \in N(v_i)} w(R(u))[i]\). Let \(N(v_i) = \{u_1, \ldots, u_{|N(v_i)|}\}\). For each \(\ell \in |N(v_i)|\), we have that \(R(u_\ell)\) is connected to some vertex \(v_i^{(\ell)} \in f_3^{-1}(v_i)\), and hence that \(w(R(u_\ell))[i] = \varrho_S(v_i^{(\ell)})\). Hence we have that \(\lambda B = \lambda B - |N(v_i)|\cdot \varrho_S(R(v_i)) + \sum_{\ell \in |N(v_i)|}\varrho_S(v_i^{(\ell)})\). Hence \(|N(v_i)|\cdot \varrho_S(R(v_i)) = \sum_{\ell \in |N(v_i)|}\varrho_S(v_i^{(\ell)})\). But because the values in the image of the bijection \(\varrho_S\) are a \(\lambda\)-average free set and \(N(v_i)\subseteq P\) certainly has size less than \(|P|=\lambda\), we have that \(\forall \ell: v_i^{(\ell)} = R(v_i)\). Thus for all \(\ell\), \(R(u_\ell)\) is connected to \(R(v_i)\), which was to be proven.

\item[\underline{5. Merging Preimages in \(S'_1\) and \(S'_2\):}] Lastly, we go from \(\mathcal{I}_4 = (H_4, G_4, f_4, w_4)\) to the final instance \(\mathcal{I}_5 = (H_5, G_5, f_5, w_5)\) where \(H_5\) is a Twin Water Lily of order \((r_1,r_2)\). This step is similar to step~\hyperref[step3-unweighted-lemma]{step 3} of the proof of the~\hyperref[lem:unweighted-lemma]{Unweighted Lemma}.

  We merge groups of vertices within \(S'_1\) and \(S'_2\). In both sets, these groups have size \(h\).

    We split \(S'_1\) into \(r_1\) groups \(X_1, \ldots, X_{r_1}\) of size \(h\), and we split \(S'_2\) into \(r_2\) groups \(X_{r_1+1}, \ldots, X_{r_1+r_2}\) of size \(h\). In \(H_5\), we have for every \(i \in [r_1+r_2]\) a vertex \(x_i\) representing \(X_i\). Each vertex \(x'_i \in f_5^{-1}(x_i)\) corresponds to a configuration \(\mathrm{conf}(x'_i) \in \mathcal{Conf}(X_i)\). The weight of \(x'_i\) is \(w(\mathrm{conf}(x'_i))\), i.e.\ the sum of the weights of the vertices in the image of the configuration. In accordance with the definition of a Twin Water Lily, we define \(S_1 = \{x_1, \ldots, x_{r_1}\}\) and \(S_2 = \{x_{r_1+1}, \ldots, x_{r_1+r_2}\}\).

    The set \(P \subseteq V(H_5)\) remains the same as in the preceding step, including its preimages and the weights of the vertices in the preimages. For each \(u \in P\), we go through the vertices \(v\) in the neighbourhood of \(u\) in \(H_4\), and connect \(u\) to \(x_i\) such that \(v \in X_i\). Note that each \(u \in P\) is still connected to at most \(h\) other vertices, but that it can be less if multiple vertices of its neighborhood came from the same group. Finally, we connect a vertex \(u' \in f^{-1}(u)\) to a vertex \(x'_i \in f_5^{-1}(x_i)\) if and only if \(\forall v \in N_{H_2}(u)\cap X_i: u'(\mathrm{conf}(x'_i)(v)) \in E(G_4)\).

    Correctness is easy to see. Note that \(H_5\) is a Twin Water Lily of order \((r_1,r_2)\) now. The set \(S_1\) has size \(r_1\) and \(S_2\) has size \(r_2\), with respective preimages of size \(n^{k_1/r_1}\) and \(n^{k_2/r_2}\). The preimages of \(P\) still have size \(\max\{n^{k1/r_1}, n^{k_2/r_2}\}\). Furthermore, the weights constructed in step 4 have \(hr_1\) blocks of \(\lceil\log(2\lambda B)\rceil\) bits, hence the maximum weight is \(\Theta(2^{\zeta})\) where \(\zeta = hr_1\cdot (\log(\lambda B) + O(1)) = hr_1 \cdot \log(|P|^{1+c/\varepsilon}n^{(1+\varepsilon)k_1/(hr_1)}) + O(1) = hr_1 \cdot \log(n^{(1+\varepsilon)k_1/(hr_1)}) + O(1) = (1+\varepsilon)k_1\log(n) + O(1)\), hence as promised in the statement of the lemma the maximum weight is \(\Theta(n^{(1+\varepsilon)k_1})\).

    \end{description}

    We have shown a reduction that has properties \((b)\) and \((c)\) from the lemma. We still have to analyze the running time. It is easy to see that steps 1, 2 and 3 run in time \(O(n^{2h-1})\), just as in the proof of the unweighted lower bound. In step 4, we need to construct the \(k\)-average free set \(S\), which is done in time \(poly(n^{k_1/(hr_1)}) = O((n^{k_1/(hr_1)})^{\hat{c}})\) for some universal constant \({\hat{c}} \in \mathbb{Z}\). The rest of step 4 as well as step 5 can again be done in time \(O(n^{2h-1})\). This concludes the proof.
\end{proof}

We now use this lemma to prove the lower bound for \textsc{Exact Weight Colored Subgraph Isomorphism}. Note that it only gives node-weighted instances. But as the following lemma shows, this is enough.. It shows that we can always convert node-weighted instances to edge-weighted instances, hence showing lower bounds for node-weighted instances immediately shows lower bounds for edge-weighted instances. It can also be reused later in the algorithms, as we can then always assume that the instances are edge-weighted.

\begin{proposition}\label{prop:node-to-edge-weight}
  Given an instance of \textsc{Exact Node Weight Colored Subgraph Isomorphism} where each vertex $v \in V(G)$ with $w(v) \neq 0$ has degree at least one, we can transform it into an equivalent instance of \textsc{Exact Edge Weight Colored Subgraph Isomorphism} in such a way that only the weight function changes. Furthermore, this reduction runs in time linear in the input size.
\end{proposition}

\begin{proof}  
  We construct a new weight function $w': E(G) \to \mathbb{Z}$. Initially, $w'(e) = 0$ for all $e \in E(G)$. The idea is to push the weight of each vertex of non-zero weight in $G$ onto one of its edges. Hence let $v \in V(G)$ with $w(v) \neq 0$. Then there must be $u \in V(G)$ with $uv \in E(G)$. We add $w(v)$ to $w'(uv)$. This completes the reduction.

  It can easily be seen that if there is a solution in the node-weighted instance with $w$, then that same solution must work with $w'$ and vice versa.
\end{proof}

This enables us to prove the following theorem. Note that it implies Theorem~\ref{corollary-lower-bound-weighted} via Lemma~\ref{lemma-equiv}.

\begin{theorem}\label{lower-bound-weighted}
  For both the node- and edge weighted variant of the problems, the following statements are true.
  \begin{enumerate}
  \item For each \(t \geq 3\), each \(\gamma \in \mathbb{R}^+\) and any \(3\leq h \leq t\), there exists a connected, bipartite graph \(\mathcal{H}_{t,h,\gamma}\) of treewidth \(t\) such that there cannot be an algorithm solving the \textsc{Exact Weight Colored Subgraph Isomorphism} problem on pattern graph \(\mathcal{H}_{t,h,\gamma}\) for instances with maximum weight \(W = \Theta(n^{\gamma})\) in time \(O(n^{t+1-\varepsilon}W)\), unless the \(h\)-uniform \textsc{Hyperclique} hypothesis fails.
  \item For each \(t\geq 1\), each \(\gamma \in \mathbb{R}^+\) and any \(h \geq 3\), there exists a connected, bipartite graph \(\mathcal{H}_{t,h,\gamma}\) of treewidth \(t\) such that there cannot be an algorithm solving the \textsc{Exact Weight Colored Subgraph Isomorphism} problem on pattern graph \(\mathcal{H}_{t,h,\gamma}\) for instances with maximum weight \(W = \Theta(n^{\gamma})\) in time \(O(n^{t-\varepsilon}W)\), unless the \(h\)-uniform \textsc{Hyperclique} hypothesis fails.
    \item For each \(t\geq 1\) and each \(\gamma \in \mathbb{R}^+\), there exists a connected, bipartite graph \(\mathcal{H}_{t,\gamma}\) of treewidth \(t\) such that there cannot be an algorithm solving the \textsc{Exact Weight Colored Subgraph Isomorphism} problem on pattern graph \(\mathcal{H}_{t,\gamma}\) for instances with maximum weight \(W=\Theta(n^{\gamma})\) in time \(O(n^{(t+1)\omega/3-\varepsilon}W^{\omega/3})\), unless the \textsc{Clique} hypothesis fails.
    \end{enumerate}
  \end{theorem}
    

\noindent
 The following implies Theorem~\ref{corollary-lower-bound-weighted-pathwidth} from the results section.

\begin{theorem}[Theorem~\ref{lower-bound-weighted} for pathwidth]\label{lower-bound-weighted-pathwidth}
  Parts 2 and 3 of Theorem~\ref{lower-bound-weighted} also hold when replacing the treewidth \(t\) by the pathwidth \(p\). Part 1 does not hold.
\end{theorem}

  We remark that by the algorithm presented in Theorem~\ref{upper-bound-detection-pathwidth}, we cannot hope to obtain a lower bound as in part 1 of Theorem~\ref{lower-bound-weighted} for the case of pathwidth.

\begin{proof}[Proof (of theorem~\ref{lower-bound-weighted})]
  Note that by Proposition~\ref{prop:node-to-edge-weight}, it suffices to prove lower bounds for the node-weighted case.

  We begin with \textbf{part 1} of the theorem. Let \(t \geq 3\) and \(\gamma \in \mathbb{R}^+\), as well as \(3\leq h \leq t\) be given. We apply the~\hyperref[lem:main-lemma]{Weighted Lemma} with
  \begin{itemize}
  \item some \(\varepsilon' > 0\) chosen later,
  \item some \(\beta' \in (0,1)\cap \mathbb{Q}\) chosen later,
  \item \(h' := h\),
  \item \(r_2' := t+1\) and
  \item some arbitrary \(r_1'\in \mathbb{N}\) with
    \begin{itemize}
    \item \(r_1' > \frac{\hat{c}\beta'}{h}\) (this ensures that the
      running time \(O(n^{2h-1} + n^{\hat{c}\beta' k/(hr_1')})\) of the
      reduction is equal to \(O(n^{k-\varepsilon})\) for
      some \(\varepsilon>0\), and can hence be ignored in the
      analysis) and
    \item \(r_1' > \frac{\beta'(t+1)}{1-\beta'}\) (this
      ensures that
      \(\max\{n^{\beta' k/r_1'}, n^{(1-\beta')k/(t+1)}\} =
      n^{(1-\beta')k/(t+1)}\)).
  \end{itemize}

  \end{itemize}
  This produces a \(k \in \mathbb{N}\) and a reduction algorithm \(\mathcal{A}\) with the properties from the lemma. In particular, the reduction algorithm produces instances where the pattern graph \(H\) is a Twin Water Lily of order \((r_1, r_2)\), which we define to be our graph \(\mathcal{H}_{t,\gamma}\).

  Now suppose there is an algorithm for the \textsc{Exact Weight Colored Subgraph Isomorphism} problem on pattern graph \(\mathcal{H}_{t,\gamma}\) running in time \(O(N^{t+1-\varepsilon}W)\) (the case \(O(N^{t+1}W^{1-\varepsilon})\) is analogous). We show that the \(h\)-uniform \textsc{Hyperclique} hypothesis fails by showing that there is an algorithm for \(h\)-uniform \(k\)-\textsc{Hyperclique} running in time \(O(n^{k-\varepsilon})\) for some \(\varepsilon > 0\).

  Given an \(h\)-uniform \(k\)-\textsc{Hyperclique} instance, we use algorithm \(\mathcal{A}\) to obtain an equivalent instance of \textsc{Exact Weight Colored Subgraph Isomorphism} where the pattern graph is a Twin Water Lily of order \((r_1',t+1)\), the preimages have size \(N = \max\{n^{\beta' k/r_1'}, n^{(1-\beta')k/(t+1)}\} = n^{(1-\beta')k/(t+1)}\), and the maximum weight is \(W = \Theta(n^{(1+\varepsilon')\beta' k})\).

  First, we make sure that \(W = \Theta(N^{\gamma})\) by choosing \(\beta'\) and \(\varepsilon'\) accordingly. Substituting, we get \(n^{(1+\varepsilon')\beta' k} = \Theta(n^{\gamma (1-\beta')k/(t+1)})\), which is true if and only if
  \begin{align*}
  (1+\varepsilon')\beta'k = \frac{\gamma(1-\beta')k}{t+1} \iff \frac{\beta'}{1-\beta'} = \frac{\gamma}{(t+1)(1+\varepsilon')} \iff \varepsilon' = \frac{\gamma(1-\beta')}{(t+1)\beta'}-1
  \end{align*}
  Hence we choose \(\varepsilon'\) as such. However, to apply the~\hyperref[lem:main-lemma]{Weighted Lemma}, we must have \(\varepsilon' \in (0,1)\). Hence we get the following two constraints for \(\beta'\):    

  \begin{align*}
    \frac{\gamma(1-\beta')}{(t+1)\beta'} - 1 > 0 \ \ &\iff\ \ \beta' < \frac{\gamma}{(t+1)+\gamma}\\
    \frac{\gamma(1-\beta')}{(t+1)\beta'} - 1 < 1 \ \ &\iff\ \ \beta' > \frac{\gamma}{2(t+1)+ \gamma}
  \end{align*}
  We incorporate these constraints later.

  Now we need to ensure that the new running time we get is also small. We solve the \textsc{Exact Weight Colored Subgraph Isomorphism} instance in time \(O(N^{t+1-\varepsilon}W) = O(N^{ t + 1 + \gamma - \varepsilon}) = O(n^{(t + 1 + \gamma -\varepsilon)(1-\beta')k/(t+1)})\). Hence we get the following additional constraints on \(\beta'\):
  \begin{align*}
    \frac{(t + 1 + \gamma -\varepsilon)(1-\beta')}{t+1} < 1 \iff 1-\beta' < \frac{t+1}{t+1+\gamma-\varepsilon} \iff \beta' > \frac{\gamma - \varepsilon}{(t+1)+\gamma-\varepsilon}
  \end{align*}
  Combining these three constraints on \(\beta'\), we get
  \begin{align*}
    \max\left\{\frac{\gamma - \varepsilon}{(t+1)+\gamma-\varepsilon}\ , \ \frac{\gamma}{2(t+1)+ \gamma}\right\} < \beta' < \frac{\gamma}{(t+1)+\gamma}
  \end{align*}
  Clearly, it is always possible to choose a \(\beta' \in (0,1)\cap \mathbb{Q}\) such that this is true.

  \textbf{Part 2} of the theorem is very much analogous. Note that the loss of the \(1\) in the exponent is due to the weaker bound in Proposition~\ref{prop:treewidth-of-lily}.

  \textbf{Part 3} is completely analogous for the case \(t\geq 2\); we simply always choose \(h = 2\).  The \(\omega/3\) in the bound comes from the \textsc{Clique} hypothesis.

  However, a small trick has to be used for the case \(t=1\), since a \(2\)-wide Twin Water Lily of order \((r_1, 2)\) has treewidth \(2\), not \(1\). To get the better lower bound, we have to slightly modify the proof of the~\hyperref[lem:main-lemma]{Weighted Lemma} for \(h=2\) in step 3. Instead of replacing the edges between vertices of \(S_2\) by intermediate vertices, we simply leave them as-is. Now the resulting graph is not a Twin Water Lily anymore, but does always have treewidth \(r_2\). The rest of the proof is analogous.
  
\end{proof}
Finally, we prove the same theorem for pathwidth.
\begin{proof}[Proof (of Theorem~\ref{lower-bound-weighted-pathwidth})]
  Proving part 2 is exactly analogous to part 2 of the theorem for treewidth.

  Now remember that we needed a slight modification of the proof of the~\hyperref[lem:main-lemma]{Weighted Lemma} for part 3 of the theorem for treewidth for \(t=1\). For part 3 of the theorem for pathwidth, we actually need that modification for all \(t\), i.e. we always leave $S_2$ as-is in step 3. The rest of the proof is analogous.  
\end{proof}

\subsubsection{Subset Sum}\label{sec:subset-sum}

We remark that with basically the same technique as is used to prove the~\hyperref[lem:main-lemma]{Weighted Lemma}, we can also prove a lower bound on the \textsc{Subset Sum} problem\footnote{Defined as: Given a set $A\subseteq \mathbb{N}$ of $n$ numbers and a target $T \in \mathbb{N}$, determine whether $\exists B \subseteq A: \sum_{b \in B} b = T$.}.


\begin{theorem}\label{thm-hyperclique-to-subsetsum}
  For no $\varepsilon > 0$ can there be an algorithm which solves \textsc{Subset Sum} in time $O(T^{1-\varepsilon}\poly(n))$ unless the \(h\)-uniform \textsc{Hyperclique} hypothesis fails for all \(h\geq 3\).
  \end{theorem}

  In~\cite{abboud2019seth}, a slightly better lower bound for \textsc{Subset Sum} is proven under SETH: They prove that unless SETH fails, \textsc{Subset Sum} cannot have an algorithm running in time $O(T^{1-\varepsilon}2^{o(n)})$.

  We briefly discuss other current algorithms and lower bounds for \textsc{Subset Sum}. The \textsc{Subset Sum} problem has a well-known $O(Tn)$ time algorithm using dynamic programming~\cite{richard1957dynamic}. Very recently, suprising new algorithms with running time $\widetilde{O}(\sqrt{n}T)$ \cite{koiliaris2017faster,koiliaris2018subset} and $\widetilde{O}(T+n)$~\cite{bringmann2017near} have been shown, the latter matching several conditional lower bounds from SETH~\cite{abboud2019seth}, \textsc{Set Cover}~\cite{cygan2016problems}, \textsc{$k$-clique} (observed in~\cite{bringmann2017near} via techniques from~\cite{abboud2014losing}), and now from \textsc{Hyperclique}. The algorithm in~\cite{bringmann2017near} is slightly simplified in~\cite{jin2018simple}, with improvements in the log factors of the running time. The \textsc{Subset Sum} problem is closely related to the \textsc{\(k\)-Sum} problem (see also Appendix~\ref{apx:k-sum-to-subset-sum}).

There are two ways to see why Theorem~\ref{thm-hyperclique-to-subsetsum} is true, and they both more or less lead to the same reduction. Both involve first reducing a \textsc{Hyperclique} instance to a \textsc{\(k\)-Sum} instance, after wich a well-known reduction from \textsc{\(k\)-Sum} to \textsc{Subset Sum} can be used. We describe the latter reduction formally in Appendix~\ref{apx:k-sum-to-subset-sum}.

The first way to see the result is a generalization of the reduction from \textsc{\(k\)-Clique} to \textsc{\(k\)-Sum} described by~\cite{abboud2014losing}. Instead of encoding edges with only two endpoints in the weights, we encode hyperedges. The second way to see the result (as stated, they lead to the same reduction) is via a slight modification of the~\hyperref[lem:main-lemma]{Weighted Lemma} to encode everything in the weighted part. We give details for the second way in Appendix~\ref{apx:subset-sum-via-weighted-lemma}.

\section{Algorithmic Results}\label{section-algos}

\subsection{\(k\)-Wise Matrix Products}

In our algorithms, the following generalization of matrix multiplication to tensors is both a crucial building block and a bottleneck. It was defined in its general form in~\cite{gnang2011spectral} and explored further algorithmically in~\cite{lincoln2018tight}.

Given $k$ tensors $A^1, \ldots, A^k$ of order k with dimensions $\overbrace{n \times \ldots \times n}^{\text{k times}}$, we define the \textbf{\(k\)-wise matrix product} $\MP_k(A^1, \ldots, A^k)$ to be the tensor given by
\begin{equation*}
  \MP_k(A^1, \ldots, A^k)[i_1, \ldots, i_k] := \sum_{\ell\in[n]}A^1[\ell,i_2,\ldots,i_k]\cdot A^2[i_1, \ell, i_3, \ldots, i_k]\cdots A^k[i_1,\ldots,i_{k-1},\ell]
\end{equation*}

Clearly, for $k=2$ this product is exactly matrix multiplication. There is also a boolean version of this generalized matrix product, just as there is a boolean version of the standard matrix product. In this boolean version, the tensors contain truth values (or equivalently 0/1 values) and the sum is replaced by an OR, while the products are replace by ANDs.

 We briefly discuss the computational complexity of \(k\)-wise matrix products. They can trivially be computed in time $O(n^{k+1})$ for all $k$, and in time \(O(n^\omega)\) for \(k=2\) via techniques originating from Strassen~\cite{strassen1969gaussian} (for a history and introduction, see Bläser~\cite{blaser2013fast}). Unfortunately, as Lincoln, Williams and Williams observe in~\cite{lincoln2018tight}, it is impossible that faster Strassen-like algorithms with running time $O(n^{k+1-\varepsilon})$ with $\varepsilon > 0$ exist for $k\geq 3$: Just as the operation of $n\times n$ by $n\times n$ matrix multiplication has a corresponding order 3 tensor of dimensions $n\times \ldots \times n$, there is an order $k+1$ tensor of dimension $n\times \ldots \times n$ corresponding to the $k$-wise matrix product. For $k\geq 3$, this tensor has border rank $n^{k+1}$, i.e.\ it cannot be expressed as the limit of a sequence of tensors of rank smaller than $n^{k+1}$.

\subsection{\(k\)-Wise Tree Decompositions}

In our algorithms, we use a structured form of tree decompositions that allow us to apply \(k\)-wise matrix products very easily to solve the \textsc{Colored Subgraph Isomorphism} problems (either weighted or unweighted) on them. Any tree decomposition can be converted to this structured form without changing its width. Its definition is very loosely based on the structured tree decompositions that~\cite{curticapean2017homomorphisms} describes for treewidth 2 graphs. See figure~\ref{fig:k-wise-tree-decomp} for a partial illustration of the structured form in the context of the algorithm for the unweighted problem.

\begin{definition}
  We call a tree decomposition \(\mathcal{T} = (T, \{X_t\}_{t\in V(T)})\) a \textbf{\(k\)-wise tree decomposition} if it satisfies the following requirements. \(T\) must be a rooted tree, and each of its nodes has one of three types: It is either an \textbf{intermediate-result node}, a \textbf{\(k\)-wise node}, or a \textbf{merge node}. Intermediate-result and merge nodes have bags of size \(k\), while \(k\)-wise nodes have bag size \(k+1\). We require the root to be an intermediate-result node, and all leaves to be merge nodes. Finally, the three types of nodes are defined as follows:
  \begin{enumerate}[(i)]
  \item If \(t \in V(T)\) is an intermediate-result node, it has two children: a \(k\)-wise node and a merge node \(t'\) with \(X_t = X_{t'}\).
  \item If \(t \in V(T)\) is a \(k\)-wise node, its parent \(\parent(t)\) is an intermediate-result node and its set of children \(\children(t)\) consists of exactly \(k\) merge nodes. Furthermore, we can rename the nodes in its bag to \(X_t = \{v, u_1, \ldots, u_k\}\) such that \(X_t = X_{\parent(t)} \cup \{v\}\) and \(\forall i \in [k]: \exists c(i) \in children(t): X_t = X_{c(i)} \cup \{u_i\}\).
  \item If \(t \in V(T)\) is a merge node, it has arbitrarily many children which must all be intermediate-result nodes. Furthermore, for each child \(t'\) we have \(X_t = X_{t'}\).
  \end{enumerate}
  We call intermediate-result and merge nodes \textbf{helper nodes}, and define \(\help(T)\subseteq V(T)\) to be the set of all helper nodes.
\end{definition}

\noindent
Note that in particular, a \(k\)-wise tree decomposition has width \(k\).

\begin{lemma}\label{lem:k-wise-tree-decomps}
  Let \(H\) be a graph and let \(\mathcal{T} = (T, \{X_t\}_{t \in V(T)})\) be a tree decomposition of \(H\) that has width \(\width(\mathcal{T})\). Then we can convert \(\mathcal{T}\) into a \(\width(\mathcal{T})\)-wise tree decomposition \(\mathcal{T}' = (T', \{X_t\}_{t \in V(T')})\). Furthermore, \(|V(T')| = \poly(|V(T)|)\cdot \width(\mathcal{T})\) and the conversion can be done in time \(O(\poly(|V(T)|)\cdot \poly(\width(\mathcal{T})))\).
\end{lemma}

\begin{proof}
  First, we go from the input tree decomposition \(\mathcal{T}\) to one where all bags have size \(\width(\mathcal{T})+1\), and where for any adjacent nodes \(t,t' \in V(T)\), we have \(X_t \cap X_{t'} = \width(\mathcal{T})\). To do this, we roughly follow the outline of an algorithm that~\cite{curticapean2017homomorphisms} describes for treewidth 2 graphs. First, we merge adjacent nodes \(t, t' \in V(T)\) with \(X_t = X_{t'}\). If there exists a node \(t\) with \(|X_t| \leq \width(\mathcal{T})+1\) and a neighbor \(t'\) such that \(X_{t'} \not\subseteq X_t\), then we simply add an element of \(X_{t'} \setminus X_t\) to \(X_t\). Applying this rule exhaustively, we obtain a tree decomposition where all bags have exactly \(\width(\mathcal{T})+1\). Applying this rule exhaustively, the resulting tree decomposition has the desired properties. Now we take any adjacent nodes \(t,t' \in V(T)\) with \(|X_t \cap X_{t'}| < \width(\mathcal{T})\) and, letting \(u \in X_t \setminus X_{t'}\) and \(v\in X_{t'} \setminus X_t\), insert a vertex with bag \((X_t \cup \{v\}) \setminus \{u\}\) between them.

  Now we root \(T\) in an arbitrary node. For any \(t \in V(T)\) and any \(t' \in \children(t)\), we subdivide the edge \(tt'\), and give the new node the bag \(X_t \cap X_{t'}\). We call the newly inserted vertices small, and all other vertices large. The root node \(r'\) must be a large node. We give it a small parent by selecting an arbitrary subset \(X \subseteq X_{r'}\) of size \(\width(\mathcal{T})\), adding a new root \(r\) to \(T\) with bag \(X\) and making \(r'\) a child of \(r\).

  We now iterate over the large nodes in \(T\) (as it is now) in depth-first search order. For each large vertex \(t\), we (temporarily) rename the vertices of its bag to \(X_t = \{v, u_1, \ldots, u_{\width(\mathcal{T})}\}\) such that its (small) parent has bag \(X_{\parent(t)} = \{u_1, \ldots, u_{\width(\mathcal{T})}\}\). Furthermore, for each \(i\), let \(C_i\subseteq \children(t)\) be the (possibly empty) set of all children \(t'\) such that \(X_{t'} = X_t \setminus \{u_i\}\). We add a node \(c(i)\) with bag \(X_t \setminus \{u_i\}\) to \(V(T)\) and connect all children \(t' \in C_i\) to \(c(i)\) instead of \(t\). Finally, we make \(c(i)\) a child of \(t\). We do something analogous for \(v\), namely letting \(C_v \subseteq \children(t)\) be the (possibly empty) set of all children \(t'\) such that \(X_{t'} = X_t \setminus \{v\}\), we add a node \(c(v)\) with bag \(X_t \setminus \{v\}\) to \(V(T)\) and connect all children \(t' \in C_v\) to \(c(v)\) instead of \(t\). However, we make \(c(v)\) a child of \(\parent(t)\), not \(t\). Now we are done: The newly added vertices \(c(i)\) (for all \(i\)) and \(c(v)\) are merge nodes, their children are intermediate-result nodes, and the large nodes \(t\) are \(k\)-wise nodes. Note also that because we allowed the sets \(C_i\) (for all \(i\)) and \(C_v\) to be empty, we now have that all leaves are merge nodes. 

  This concludes the construction. It is easily seen that the number of nodes of the new tree decomposition is \(\poly(|V(T)|)\cdot \width(\mathcal{T})\), and that this conversion algorithm works in time \(O(\poly(|V(T)|) \cdot \poly(\width(\mathcal{T})))\).
\end{proof}

\subsection{Colored Subgraph Isomorphism for Bounded Treewidth}\label{subsection-algos-detection}

We begin by looking at the unweighted version of \textsc{Colored Subgraph Isomorphism}. As was previously mentioned for Theorem~\ref{corollary-upper-bound}, these results essentially follow from~\cite{alon1995color} and~\cite{curticapean2017homomorphisms}, but are now unified via a single technique.

\begin{theorem}\label{upper-bound-detection}
  There is an algorithm which, given an arbitrary instance $\phi = (H,G,f)$ of \textsc{Colored Subgraph Isomorphism}, solves $\phi$ in time
  \begin{enumerate}
  \item $O(n^{\tw(H)+1}\poly(k) + g(k))$ when $\tw(H)\geq 3$,
  \item $O(n^\omega \poly(k) + g(k))$ when $\tw(H) = 2$, where $\omega$ is the exponent of matrix multiplication, and
  \item $O(n^2\poly(k) + g(k))$ when $\tw(H) = 1$.
  \end{enumerate}
  where $k := |V(H)|$, $n$ is the size of the preimages of $f$, and $g$ is a computable function.
\end{theorem}

Obviously, a proof of this theorem suffices to prove Theorem~\ref{corollary-upper-bound}, since we can simply plug the algorithm into Lemma~\ref{lemma-equiv}.

\begin{proof}[Proof (of part 1 of Theorem~\ref{upper-bound-detection})]
  We describe an algorithm which, given $G, H$ and $f: V(G) \to V(H)$, first calculates an optimal tree decomposition \(\mathcal{T}_{\text{initial}}\) of width \(\tw(H)\) for $H$ in time $g(k)$ (via the algorithm by Bodlaender~\cite{bodlaender1996linear}, see also Section~\ref{prelim-treewidth}), then finds a solution via dynamic programming over the tree decomposition. The algorithm that computes the optimal tree decomposition also ensures that its tree graph \(T_{\text{initial}}\) has size \(|T_{\text{initial}}| = \poly(k)\). By assumption, $\tw(H) \geq 3$. We shorten \(\tw(H)\) to \(\tw\) in the following. We use a slightly more complicated framework than necessary, because it generalizes nicely to a proof of part 2 and to a proof of parts 1 and 2 of Theorem~\ref{upper-bound-weighted}.

  To make our algorithm as easy as possible, we begin by applying Lemma~\ref{lem:k-wise-tree-decomps} to convert our tree decomposition into a \(\tw\)-wise tree decomposition $\mathcal{T} = (T, \{X_t\}_{t\in V(T)})$. Since \(|T_{\text{initial}}| = \poly(k)\), the time this conversion takes is certainly negligible. Furthermore, \(|T| = \poly(|T_{\text{initial}}|)\cdot \poly(\width(\mathcal{T}_{\text{initial}})) = \poly(\poly(k))\cdot \tw = \poly(k)\).

  We now do dynamic programming over the \(\tw\)-wise tree decomposition. Using notation and nomenclature from Section~\ref{notation-and-nomenclature-colsubiso}, we only store values for each configuration of the bags of helper nodes, not for \(k\)-wise nodes. In particular, for each helper node \(t \in \help(T)\) we store the following function of finite domain from configurations of $X_t$ to truth values. Remember that $\mathrm{ParSol}(S; I; J)$ is true if and only if $S$ is a partial solution of $I$ in $J$, for $I \subseteq J \subseteq V(H)$. 
  \begin{align*}\label{eq-dp-for-normal-subgraph-iso}
    d_{t} &: \mathcal{Conf}(X_t) \to \{\mathrm{true},\mathrm{false}\}\\
    d_{t}(R) &:= \ParSol(R; X_t; V_{t})\tag{1}
  \end{align*}
  i.e.\ we store for each configuration whether it is a partial solution of \(X_t\) in the cone \(V_{t}\). We call these functions DP functions (where DP stands for dynamic programming). Since there are $n^{\tw}$ many configurations for $X_t$, each DP function $d_{t}$ can be specified using $n^{\tw}$ many bits.

  We calculate the DP functions $d_{t}$ for all $t\in\help(T)$ in a bottom-up manner. The overall picture of the algorithm is very simple: At a merge node \(t\), we take the DP functions of all children and do a pointwise AND. At an intermediate-vertex node \(t\), we first calculate a \(\tw\)-wise matrix product for its \(\tw\)-wise node child, then AND the result with the DP function of its merge node child. See figure~\ref{fig:k-wise-tree-decomp} for a conceptual illustration.
  \begin{figure}[h]
    \centering
    \scalebox{1.3}{
      \tikzstyle{every picture}=[tikzfig]
      \begin{tikzpicture}
	\begin{pgfonlayer}{nodelayer}
		\node [style=small-node] (1) at (-2.5, 8) {};
		\node [style=small-node] (2) at (-8, 2) {};
		\node [style=small-node] (3) at (-5, 2) {};
		\node [style=small-node] (5) at (-2, 2) {};
		\node [style=small-node] (6) at (2.5, 2) {};
		\node [style=small-node] (7) at (0.25, -2.5) {};
		\node [style=none] (8) at (-9.5, -0.5) {};
		\node [style=none] (9) at (-8, -0.5) {};
		\node [style=none] (10) at (-6.5, -0.5) {};
		\node [style=none] (11) at (-3.25, -0.5) {};
		\node [style=none] (12) at (1.5, -0.5) {};
		\node [style=none] (13) at (3.5, -0.5) {};
		\node [style=k-wise-border] (14) at (-5, 5) {};
		\node [style=k-wise-node] (15) at (-5, 5) {};
		\node [style=none] (16) at (-6, 5.5) {};
		\node [style=none] (17) at (-4, 5.5) {};
		\node [style=none] (18) at (-1.25, 2.75) {};
		\node [style=none] (19) at (-2.5, 1) {};
		\node [style=none] (20) at (-7.5, 1) {};
		\node [style=none] (21) at (-8.75, 2.75) {};
		\node [style=none] (22) at (-2.75, 7.25) {};
		\node [style=none] (23) at (-2.25, 7.25) {};
		\node [style=none] (24) at (1.475, 2.75) {};
		\node [style=none] (25) at (-9.25, -0.5) {};
		\node [style=none] (26) at (-7.775, -0.5) {};
		\node [style=none] (27) at (-6.25, -0.5) {};
		\node [style=none] (28) at (-8.1, 1.5) {};
		\node [style=none] (29) at (-7.775, 1.475) {};
		\node [style=none] (30) at (-7.475, 1.475) {};
		\node [style=none] (31) at (-3, -0.5) {};
		\node [style=none] (32) at (-2, 1.5) {};
		\node [style=none] (33) at (-1.5, 1.5) {};
		\node [style=none] (34) at (0.25, -2) {};
		\node [style=none] (35) at (1.775, -0.5) {};
		\node [style=none] (36) at (3.775, -0.5) {};
		\node [style=none] (37) at (2.525, 1.5) {};
		\node [style=none] (38) at (2.95, 1.5) {};
		\node [style=none] (39) at (-4, 5.5) {};
		\node [style=none] (40) at (-4, 5.5) {};
		\node [style=none] (41) at (-4.1, 5.575) {};
		\node [style=none] (42) at (-6, 4.25) {};
		\node [style=none] (43) at (-4, 4.25) {};
		\node [style=text node] (44) at (-4.75, 8) {intermediate-result};
		\node [style=text node] (45) at (-6.75, 5) {tw-wise};
		\node [style=text node] (46) at (-8.5, 2) {merge};
		\node [style=text node] (47) at (-2, -2.5) {intermediate-result};
		\node [style=text node] (48) at (-1.3, 4.575) {\textcolor{gray}{tw}};
		\node [style=none] (49) at (-2.5, 9.5) {};
		\node [style=none] (50) at (-0.75, -4) {};
		\node [style=text node] (51) at (-5.25, -1.75) {$\vdots$};
		\node [style=text node] (52) at (5.5, -1.75) {$\vdots$};
		\node [style=none] (53) at (1.25, -4) {};
		\node [style=text node] (54) at (-3.25, 3.25) {$\cdots$};
		\node [style=text node] (55) at (-1.75, 3.25) {$\cdots$};
		\node [style=none] (56) at (-14, -1.75) {};
	\end{pgfonlayer}
	\begin{pgfonlayer}{edgelayer}
		\draw [style=tree-edge] (1) to (15);
		\draw [style=tree-edge] (1) to (6);
		\draw [style=tree-edge] (6) to (12.center);
		\draw [style=tree-edge] (13.center) to (6);
		\draw [style=tree-edge] (7) to (5);
		\draw [style=tree-edge] (11.center) to (5);
		\draw [style=tree-edge] (15) to (5);
		\draw [style=tree-edge] (3) to (15);
		\draw [style=tree-edge] (2) to (15);
		\draw [style=tree-edge] (10.center) to (2);
		\draw [style=tree-edge] (9.center) to (2);
		\draw [style=tree-edge] (8.center) to (2);
		\draw [style=area-border] (17.center) to (18.center);
		\draw [style=area-border, in=0, out=-45, looseness=1.25] (18.center) to (19.center);
		\draw [style=area-border] (20.center) to (19.center);
		\draw [style=area-border] (21.center) to (16.center);
		\draw [style=area-border, bend left=45, looseness=1.25] (16.center) to (17.center);
		\draw [style=area-border, in=180, out=-135, looseness=1.25] (21.center) to (20.center);
		\draw [style=new edge style 0] (25.center) to (28.center);
		\draw [style=new edge style 0] (26.center) to (29.center);
		\draw [style=new edge style 0] (27.center) to (30.center);
		\draw [style=new edge style 0] (35.center) to (37.center);
		\draw [style=new edge style 0] (36.center) to (38.center);
		\draw [style=new edge style 0] (34.center) to (33.center);
		\draw [style=new edge style 0] (31.center) to (32.center);
		\draw [style=new edge style 0] (24.center) to (23.center);
		\draw [style=area-pointer] (41.center) to (22.center);
		\draw [style=new edge style 1, bend left=60] (43.center) to (42.center);
		\draw [style=tree-edge] (1) to (49.center);
		\draw [style=tree-edge] (50.center) to (7);
		\draw [style=tree-edge] (53.center) to (7);
	\end{pgfonlayer}
\end{tikzpicture}
    }
    \caption{Partial sketch of the tree \(T\) of the \(\tw\)-wise tree decompositions. Colored arrows represent operations of the algorithms, with red indicating that the result of the \(\tw\)-wise node is calculated via a \(\tw\)-wise matrix product, and green indicating that the result of that subtree is ANDed pointwise with the result of all other subtrees.}
    \label{fig:k-wise-tree-decomp}
  \end{figure}
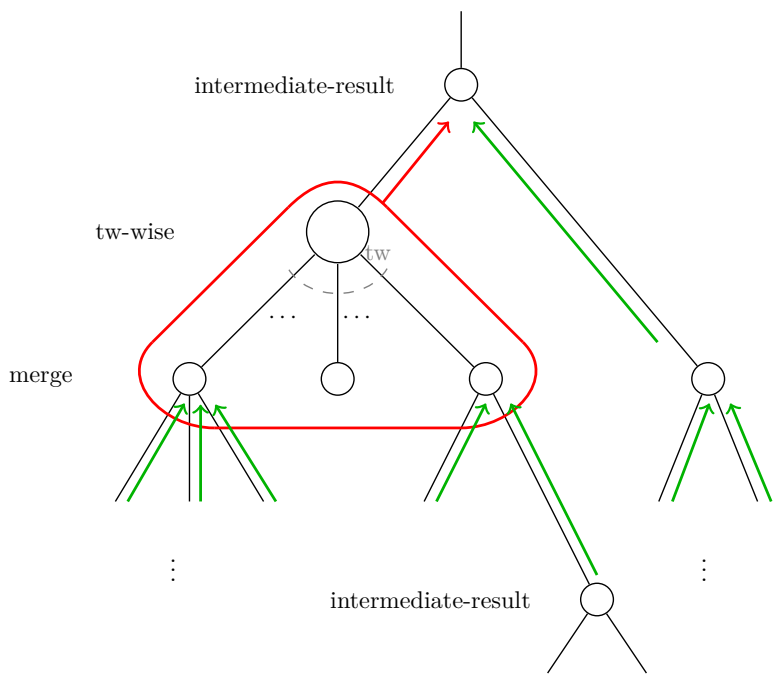
  
  
We now describe the algorithm in detail. We begin with the leaves of \(T\), which must be merge nodes. Hence let \(t\) be a merge node with no children. Since $X_t = V_{t}$, any configuration of $X_t$ is a configuration of $V_{t}$. Thus, to calculate $d_{t}(R)$ as in equation~\ref{eq-dp-for-normal-subgraph-iso}, we simply have to check whether $R$ is a valid configuration. The latter can be done in time $\poly(k)$, which leads to a total time of $O(n^{\tw}\poly(k))$ per leaf.

  Now let $t\in\help(T)$ be an inner node of $T$. There are two cases: either \(t\) is an intermediate-result node, or \(t\) is a merge node.
  \begin{description}
  \item[\underline{\(t\) is a merge node:}] Let \(t\) be a merge node. We continue to denote the set of children of \(t\) by \(\children(t)\). For each configuration \(R\) we have
    \begin{align*}
      d_t(R) &= \ParSol(R;X_t;V_t) = \ParSol\left(R; X_t; \bigcup_{c \in \children(t)} V_c\right)\\
             &= \bigwedge_{c\in\children(t)}\ParSol(R;X_t;V_c) = \bigwedge_{c\in\children(t)}d_c(R)
    \end{align*}
    Since \(t\) has at most \(\poly(k)\) children, and since there are at most \(n^{\tw}\) many possible configurations \(R\), this can be done in time \(O(n^{\tw}poly(k))\) per merge node.
    
  \item[\underline{\(t\) is an intermediate-result node:}] Let \(t\) be an intermediate-result node with \(\tw\)-wise child node \(t'\) and merge node \(t''\). By definition we have \(X_t = X_{t''}\). Furthermore \(|\children(t')| = \tw\), and we can rename the nodes in \(X_{t'}\) to \(X_{t'} = \{v, u_1, \ldots, u_{\tw}\}\) such that \(X_{t'} = X_t \cup \{v\}\) and \(\forall i \in [k]: \exists c(i) \in \children(t'): X_{t'} = X_{c(i)} \cup \{u_i\}\). We have that
    \begin{align*}
      d_t(R) &= \ParSol(R;X_t;V_t)\\
             &= \ParSol(R; X_{t''}; V_{t''})\land {}\\
      &\phantom{{}={}}\exists v' \in f^{-1}(v): \forall i \in [\tw]: \ParSol((R\cup \{v \mapsto v'\})|_{X_{c(i)}};X_{c(i)};V_{c(i)})\\
             &= d_{t''}(R) \land \exists v' \in f^{-1}(v): \forall i \in [\tw]: d_{c(i)}((R\cup \{v \mapsto v'\})|_{X_{c(i)}}) \tag{\(\ast\)}
    \end{align*}
    We show how to calculate the values of \(d_t\) via a \(\tw\)-wise boolean matrix product with 0-1-tensors of dimension \(n\times \ldots \times n\). Note that we will use truth values and 0/1 interchangeably to declutter notation. True is interchangeable with 1, false with 0.

    For convenience, all tensors from this point onward are indexed via configurations, where each dimension is indexed by a single vertex. Formally, we call a tensor \(A\) \textbf{indexed by configurations of} \(X = \{v_1, \ldots, v_{h}\} \subseteq V(H)\) \textbf{with ordering} \((v_{i_1}, \ldots, v_{i_h})\) (where \(\{v_{i_1}, \ldots, v_{i_h}\} = \{v_1, \ldots, v_{h}\}\)) when it is an order \(h\) tensor of dimension \(n\times \ldots \times n\). Abusing notation, we rename the vertices of \(f^{-1}(v_i)\) for each \(v_i \in X\) to \(\{1, \ldots, n\}\) and use them as if they were numbers. Now for a configuration \(R\) of \(X\) we define \(A[R] = A[R(v_{i_1}), \ldots, R(v_{i_h})]\).

     For our \(\tw\)-wise matrix product, we have \(\tw\) input tensors, specifically one tensor \(p_i\) for each \(i \in [\tw]\). \(p_i\) is indexed by configurations of \(X_{c(i)} = X_{t'} \setminus \{u_i\}\) with ordering \((u_1, \ldots, u_{i-1}, v, u_{i+1}, \ldots, u_{\tw})\). Specifically, for a configuration \(R\) of \(X_{c(i)}\) we define \(p_i[R] := d_{c(i)}(R)\) (it is still a 0-1-tensor, remember that truth values and 0/1 are interchangeable).

     Now we calculate the \(\tw\)-wise matrix product \(p_{\text{res}} := \MP_{\tw}(p_1, \ldots, p_{\tw})\), which we use as a tensor indexed by configurations of \(X_t\) with ordering \((u_1, \ldots, u_{\tw})\). Hence, for a configuration \(R\) of \(X_t\), we have
    \begin{align*}
      p_{\text{res}}[R] &= \MP_{\tw}(p_1, \ldots, p_{\tw})[R(u_1), \ldots, R(u_{\tw})]\\
                        &= \bigvee_{\ell \in [n]} p_1[\ell, R(u_2), \ldots, R(u_{\tw})]\land \ldots \land p_{\tw}[R(u_1), \ldots, R(u_{\tw-1}), \ell]\\
                        &= \exists v' \in f^{-1}(v): \forall i \in [\tw]: p_i[R(u_1), \ldots, R(u_{i-1}), v', R(u_{i+1}), \ldots, R(u_{\tw})]\\
                        &= \exists v' \in f^{-1}(v): \forall i \in [\tw]: d_{c(i)}((R\cup \{v \mapsto v'\})|_{X_{c(i)}})
    \end{align*}
    From this, we can directly calculate \(d_t(R)\) via equation (\(\ast\)): We have \(d_t(R) = d_{t''}(R) \land p_{\text{res}}[R]\). Indeed, this gives a very simple algorithm for computing \(d_t\). The \(\tw\)-wise matrix product for \(p_{\text{res}}\) can be done in time \(O(n^{\tw+1})\), and after that calculating \(d_t\) is a simple pointwise AND with \(d_{t''}\), which takes time \(O(n^{\tw})\). Hence overall we have running time \(O(n^{\tw+1})\).
  \end{description}

    Since the tree decomposition has \(\poly(k)\) nodes, the overall time the algorithm takes is \(O(n^{\tw+1}\poly(k))\). Now we simply have to extract the answer from DP function of the root node. Let \(r\) be the root node. We have that \(d_r(R)\) is true if and only if there is a partial solution of \(X_r\) in the cone \(V_r = V(H)\). Hence it is true only if \(R\) can be expanded to a solution. Thus, the input is a YES-instance for \textsc{Colored Subgraph Isomorphism} if and only if there is an \(R\) such that \(d_r(R)\) is true. 
\end{proof}

  So far, we have only looked at the case that $\tw(H) \geq 3$. However, this exact algorithm also achieves the second result.

\begin{proof}[Proof (of part 2 of Theorem~\ref{upper-bound-detection})]
 Note that in the algorithm for part 1, all steps run in time $O(n^{\tw(H)}\poly(k))$, except for the $\tw(H)$-wise matrix product, which can be done in time $O(n^{\tw(H)+1})$. However, for $\tw(H) = 2$, $\tw(H)$-wise matrix product is exactly matrix multiplication, which runs in time $O(n^\omega)$. Thus we obtain our second result.
\end{proof}

The third result with $\tw(H)=1$ cannot be achieved by this algorithm directly. We shortly outline why. Rooting $H$ in some arbitrary node, consider the tree decomposition that takes exactly the edges of $H$ as bags, and constructs $T$ such that two nodes are connected by an edge if and only if their bags have a non-empty intersection. Now consider a graph $H$ which contains nodes $u,v,w$ such that $u$ is the parent of $v$ and $v$ is the parent of $w$. Let $t,t'\in T$ be such that $X_t = \{u,v\}$, $X_{t'} = \{v,w\}$. But now the edge $vw$ is not covered by any of the subsets $X_{t'} \setminus \{v\}$ or $X_{t'} \setminus \{w\}$ and is thus not considered in the algorithm at all. This leads to the algorithm failing. For $\tw(H) \geq 2$, this does not happen, since any pair of nodes is contained in a bag $X_t$ of size $\geq 3$, hence there is some $u$ such that the edge is covered by the subset $X_{t} \setminus \{u\}$.

Thus, to obtain our third result, we must employ a different technique. However, this part of the theorem turns out to be easy.

\begin{proof}[Proof (of part 3 of Theorem~\ref{upper-bound-detection})]
  We only sketch the result, since it is easy to see. It suffices to employ the trivial dynamic programming solution on the tree $G$, which already has a running time of $O(n^2\poly(k))$.
\end{proof}

\subsection{Exact Weight Colored Subgraph Isomorphism for Bounded Treewidth}

We now move on to the weighted version of \textsc{Colored Subgraph Isomorphism}. Specifically, we show how to solve the \textsc{Exact Weight Colored Subgraph Isomorphism} problem for bounded-treewidth pattern graphs using dynamic programming. Remember that the instances we consider here may be either node- or edge-weighted.

We also remark that the restriction of the target weight \(T\) to zero in the \textsc{Exact Weight Subgraph Isomorphism} problems is simply for ease of discussion. We use this version to both simplify our algorithms and circumvent any problems that might arise from having $T$ be part of the input. Trivially, all of our results also hold for the problem where a target $T = O(W)$ is given.

\begin{theorem}\label{upper-bound-weighted}
  There is an algorithm which, given an arbitrary instance $\phi = (H,G,f,w)$ of the \textsc{Exact Weight Colored Subgraph Isomorphism} problem, solves $\phi$ in time
  \begin{enumerate}
  \item $O((n^{\tw(H)+1}W + n^{\tw(H)}W\log W)\poly(k) + g(k))$ when $\tw(H) \geq 3$,
  \item $O((n^\omega W + n^2W\log W)\poly(k) + g(k))$ when $\tw(H) = 2$, and
  \item $O((n^2W + nW \log W)\poly(k) + g(k))$ when $\tw(H) = 1$.
  \end{enumerate}
  where $k := |H|$, $n$ is the size of the preimages of $f$, and $W$ is the maximum absolute weight in the image of $w$.
\end{theorem}

Again, note that this implies Theorem~\ref{corollary-upper-bound-weighted} via a simple application of Lemma~\ref{lemma-equiv} to the resulting algorithm.

Before we prove this, we establish an important lemma, showing that you can basically do the \(k\)-wise matrix product of tensors of polynomials faster than naively by utilizing the Fast Fourier Transform. Recall that a Laurent polynomial $p \in \mathbb{C}[X,X^{-1}]$ is simply a polynomial which may have negative powers of the $X$.

\begin{lemma}\label{tensor-fft}
  Let $q \in \mathbb{N}$ tensors $A^1, \ldots, A^q$ of order $q$ be given, each of dimensions $n \times \ldots \times n$ and such that each of their entries $A^i_{j_1, \ldots, j_q} \in \mathbb{C}[Z,Z^{-1}]$ ($i \in [q]$ and $\forall \ell \in [q]: j_\ell \in [n]$) is a Laurent polynomial of degree bounded by $W$ in both the positive and negative direction. Then their $q$-wise matrix product can be computed in time
  \begin{enumerate}
  \item $O(n^{q+1}W + qn^qW\log W)$ for $q \geq 3$ and
  \item $O(n^{\omega}W + n^2W\log W)$ for $q = 2$.
\end{enumerate}
\end{lemma}

\begin{proof}
  It suffices to prove the result for standard polynomials of degree bounded by \(2W\), since we can shift the exponents of the polynomials such that they only have positive exponents, do the \(q\)-wise matrix product, then shift back.
  
  We assume for now that $q \geq 3$. Our algorithm is a generalization of the Fast Fourier Transform algorithm for standard polynomials (e.g. ~\cite{cooley1965algorithm}). Specifically, we evaluate each $A^i$ at the set $S$ of the $2W$-th roots of unity, obtaining $q\cdot |S|$ tensors with complex entries. Since each entry of $A^i$ is a polynomial, this can be done separately for each entry. Then, for each $s \in S$, we compute $\MP_{q}(A^1(s), \ldots, A^q(s))$, obtaining $|S|$ tensors with complex entries. Obviously, these are exactly the evaluations of $\MP_{q}(A^1, \ldots, A^q)$ at $S$. At this point we can use interpolation via the inverse Fast Fourier Transform separately for each entry to recover the result.

  Evaluating all entries of the tensors $A^i$ at the roots of unity can be done using the classic Fast Fourier Transform algorithm. Since there are $q \cdot n^{q}$ such entries, each with polynomials of degree bounded by $W$, this runs in time $O(qn^{q}W \log W)$. Similarly, the interpolation of the result can be done by the inverse Fast Fourier Transform algorithm. Since there are $n^{q}$ entries to interpolate, each with a degree bound of $O(W)$, this runs in time $O(qn^{q}W\log W)$. Finally, computing the $q$-wise matrix products for each root of unity can be done in $|S|\cdot n^{q+1} = O(n^{q+1}W)$. Thus, our total running time is $O(n^{q+1}W + qn^{q}W\log W)$.

  For the case that $q = 2$, note that in the above algorithm, all steps except for the $q$-wise matrix product run in time $O(n^qW\log W)$. For $q=2$, the $q$-wise matrix product is exactly matrix multiplication, which can be done in $O(n^{\omega})$. Thus the $O(W)$ matrix multiplications can be done in $O(n^{\omega}W)$.
\end{proof}

We use this lemma as an important tool in our proof of Theorem~\ref{upper-bound-weighted}. The framework of the proof is somewhat analogous to that of Theorem~\ref{upper-bound-detection}, but instead of storing values for every configuration of bags of nodes in $\help(T)$, they we store the values for each configuration and each achievable weight. The computation of the dynamic programming table entries is slightly more complex, with some shifting of the entries being required.

\begin{proof}[Proof (of part 1 of Theorem~\ref{upper-bound-weighted})]
  Due to Proposition~\ref{prop:node-to-edge-weight}, we only need to consider the case of edge weights. We describe an algorithm which takes as inputs $G$, $H$, $f:V(G) \to V(H)$ and a weight function $w$ describing the edge weights. It calculates an optimal tree decomposition $\mathcal{T}_1 := (T_1, \{X_t\}_{t\in V(T_1)})$ for $H$ in time $g(k)$ and then computes a solution via dynamic programming over the tree decomposition. By assumption, $\tw := \tw(H) \geq 3$. We continue using the notation and terminology (e.g. ``configuration'') as described in Section~\ref{notation-and-nomenclature-colsubiso}.

  We apply Lemma~\ref{lem:k-wise-tree-decomps} to convert the tree decomposition into a \(\tw\)-wise tree decomposition $\mathcal{T}_2 = (T_2, \{X_t\}_{t\in V(T_2)})$. Since \(|T_1| = \poly(k)\), the time this conversion takes is certainly negligible. Furthermore, \(|T_2| = \poly(k)\). To simplify our algorithm further, we also modify \(\mathcal{T}_2\) further. Specifically, we introduce a new type of helper node, a \textbf{binary-merge node}, which replaces merge nodes. It is a node that has either zero or two children, both with the same bag as itself. Those children are either binary-merge nodes or intermediate-result nodes. We convert \(\mathcal{T}_2\) into a new tree decomposition $\mathcal{T} := (T, \{X_t\}_{t\in V(T)})$ that has binary-merge nodes instead of merge nodes by repeatedly taking any remaining merge node \(t\) with more than two children, splitting its set of children into two non-empty sets \(A\) and \(B\), creating two new children \(t'\) and \(t''\) of \(t\) and making the vertices from \(A\) children of \(t'\) and the vertices of \(B\) children of \(t''\). For any merge node \(t\) with only a single child, we create a new leaf with the same bag and let it be a child of \(t\). Correctness and negligibility of the conversion time is immediate, and the new tree decomposition still has size \(\poly(k)\). We now do dynamic programming over \(\mathcal{T}\). 

  In deviation from the proof of Theorem~\ref{upper-bound-detection}, our dynamic programming table is structured differently. Instead of having only a single function of finite domain for each helper node $t$, we have one function for each $t$ and each achievable weight. Since the instance is edge-weighted, the achievable weights must all lie in $\mathcal{W} := \{-k^2W,\ldots,k^2W\}$.

  Specifically, for each helper node \(t \in \help(T)\) and each weight \(W' \in \mathcal{W}\)  we store the following function of finite domain from configurations of $X_t$ to truth values. Remember that $\ParSolE(R; I; J; W)$ is true if and only if $R$ is a partial solution of $I$ in $J$ with an extension of weight $W$.
  \begin{align*}
    d_{t,W'} &: \mathcal{Conf}(X_t) \to \{\mathrm{true},\mathrm{false}\}\\
    d_{t,W'}(R) &:= \ParSolE(R; X_t; V_t; W')
  \end{align*}
  i.e.\ we store for each configuration whether it is a partial solution of \(X_t\) in the cone \(V_{t}\) that has an extension of weight \(W'\). We call these functions DP functions (where DP stands for dynamic programming). Since there are $n^{\tw}$ many configurations for $X_t$, each DP function $d_{t,W'}$ can be specified using $n^{\tw}$ many bits.

  We calculate the DP functions \(d_t\) for all \(t\in\help(T)\) in a bottom-up manner. The overall picture of the algorithm is as follows: At a non-leaf binary-merge node \(t\), we take the DP functions of both children and do a boolean convolution. At an intermediate-vertex node \(t\), we first convert the DP functions of the children of its \(\tw\)-wise child to tensors with Laurent polynomials as entries, calculate their \(\tw\)-wise matrix product, convert it back to a DP function and then do a boolean convolution with the DP function of its merge node child.

  We now describe the algorithm in detail, beginning with the leaves of \(T\). Hence \(t\) be a merge node with no children. Since \(X_t = V_t\), we have that \(d_{t,W'}(R)\) is true if and only if \(R\) is a valid configuration and \(W' = 0\). Clearly this takes total time at most \(O(n^{\tw}W\poly(k))\).

  Now let \(t \in \help(T)\) be an inner node of \(T\). There are two cases: either \(t\) is an intermediate-result node, or \(t\) is a merge node.

  \begin{description}
  \item[\underline{\(t\) is a binary-merge node:}] Let \(t\) be a binary-merge node with children \(t'\) and \(t''\). For each configuration \(R\) we have
    \begin{align*}
      d_{t,W'}(R) &= \ParSolE(R;X_t;V_t;W') = \ParSolE(R; X_t;  V_{t'}\cup V_{t''}; W')\\
                  &= \exists W_1: \ParSolE(R; X_{t'}; V_{t'}; W_1) \land \ParSolE(R; X_{t''}; W'-W_1)\\
                  &= \exists W_1: d_{t', W_1}(R) \land d_{t'', W'-W_1}(R)
    \end{align*}
    This can be calculated using a boolean convolution for each configuration. Specifically, for each configuration \(R\) we define the finitely supported functions \(f_R, g_R: \mathbb{Z} \to \{\mathrm{true},\mathrm{false}\}\) as \(f_R(x) := d_{t',x}(R)\) if \(x \in \mathcal{W}\) and \(\mathrm{false}\) otherwise, and analogously \(g_R(x) := d_{t'',x}(R)\) if \(x \in \mathcal{W}\) and \(\mathrm{false}\) otherwise. We now use the boolean convolution of these functions, defined as \((f_R \ast g_R)(x) := \bigvee_{z=-\infty}^\infty f_R(z)\land g_R(x-z)\). We have that \((f_R\ast g_R)(W') = (\exists W_1: f_R(W_1) \land g_R(W'-W_1)) = d_{t,W'}(R)\). Since \(f_R\) and \(g_R\) only have non-zero values in a range of size \(O(W\poly(k))\), the boolean convolution \(f_R \ast g_R\) can be calculated in time \(O(W \log W\poly(k))\)\footnote{This is a standard result which can be achieved e.g. by using the Fast Fourier Transform.}.
    
    Since there are at most \(n^{\tw}\) many possible configurations \(R\), \(d_{t,W'}\) can be calculated for all \(W'\) in time \(O(n^{\tw}W \log W\poly(k))\). Hence this is the maximum running time we need per binary-merge node.
    
  \item[\underline{\(t\) is an intermediate-result node:}] Let \(t\) be an intermediate-result node with \(\tw\)-wise child node \(t'\) and merge node \(t''\). By definition we have \(X_t = X_{t''}\). Furthermore \(|\children(t')| = \tw\), and we can (temporarily) rename the nodes in \(X_{t'}\) to \(X_{t'} = \{v, u_1, \ldots, u_{\tw}\}\) such that \(X_{t'} = X_t \cup \{v\}\) and \(\forall i \in [k]: \exists c(i) \in \children(t'): X_{t'} = X_{c(i)} \cup \{u_i\}\).

First, we deal with calculating the weight of an extension for a configuration of the bag of the \(\tw\)-wise node. Let $R$ be a partial solution of $X_{t'}$ in $V_{t'}$ and let $S$ be an extension of $R$ (i.e.\ a compatible configuration of \(V_t\)). Notice that we get the following formula for the weight of the extension $S$:
  \begin{equation*}
    w_{ext}(S,R) \overset{\text{def.}}{=} w(S) - w(S|_{X_{t'}}) = \underbrace{\left(\sum_{i \in [\tw]} w_{ext}(S|_{V_{c(i)}},S|_{X_{c(i)}})\right)}_{\text{Part A}} + \underbrace{w_{ext}(S|_{X_{t'}},S|_{X_{t'}\setminus \{v\}})}_{\text{Part B}}
  \end{equation*}
  Part A is the combined weight of the edges which are not in $X_{t'}$. Compared to the left hand side, it is missing the weight of all edges going from $v$ to $X_{t'}\setminus\{v\}$, which is exactly what is then added in Part B. Recall that $w_{ext}(S|_{X_{t'}}, S|_{X_{t'}\setminus\{v\}}) = \sum_{u \in X_{t'} \setminus \{v\}} w(S(v)S(u))$.
  
    Hence we can calculate the DP function for \(t\) as follows for each configuration \(R\):
    \begin{align*}
      d_{t,W'}(R) &= \ParSolE(R;X_t;V_t;W')\\
                  &= \exists W_v: \ParSolE(R; X_{t''}; V_{t''}; W_v)\land \exists W_1, \ldots, W_{\tw}:\exists v' \in f^{-1}(v):\\
                  &\phantom{{}={}}w_{ext}((R \cup \{v\mapsto v'\}),R) + \sum_{i=1}^{\tw}W_i = W' - W_v \land {}\\
                  &\phantom{{}={}}\forall i \in [\tw]: \ParSolE((R\cup \{v \mapsto v'\})|_{X_{c(i)}};X_{c(i)};V_{c(i)}; W_i)\\
                  &= \exists W_v: d_{t'',W_v}(R)\land \exists W_1, \ldots, W_{\tw}: \exists v' \in f^{-1}(v):\\
                  &\phantom{{}={}}w_{ext}((R \cup \{v\mapsto v'\}),R) + \sum_{i=1}^{\tw}W_i = W' - W_v \land {}\\
      &\phantom{{}={}}\forall i \in [\tw]: d_{c(i), W_i}((R\cup \{v \mapsto v'\})|_{X_{c(i)}}) \tag{\(\dag\)}
    \end{align*}
    We show how to calculate the values of \(d_{t,W'}\) for all \(W'\) via a \(\tw\)-wise matrix product of Laurent polynomials. We continue using the notation and nomenclature for indexing tensors via configurations, as described in the proof of part 1 of Theorem~\ref{upper-bound-detection}, and continue using truth values and 0/1 interchangeably.

    We now define one 0-1-tensors \(p_{i,W'}\) for each \(i \in [\tw]\) and \(W' \in \mathcal{W}\), which is indexed by configurations of \(X_{c(i)} = X_{t'} \setminus \{u_i\}\) with ordering \((u_1, \ldots, u_{i-1}, v, u_{i+1}, \ldots, u_{\tw})\). For a configuration \(R\) of \(X_{c(i)}\) we define
      \begin{align*}
        p_{i,W'}[R] &:= \begin{cases}
          d_{c(1),W'-\sum_{i=2}^{\tw}w(R(v)R(u_i))}(R) &\text{ if $i=1$} \\
          d_{c(2),W'-w(R(v)R(u_1))}(R) &\text{ if $i=2$}\\
          d_{c(i),W'}(R) &\text{ otherwise}
      \end{cases}
      \end{align*}
      Note that the weight of the part B above is now encoded in these tensors, specifically in the ones for \(i=1\) and \(i=2\). In particular, we can expand equation (\(\dag\)) as
      \begin{align*}
        &= \exists W_v: d_{t'',W_v}(R)\land \exists W_1, \ldots, W_{\tw}: \sum_{i=1}^{\tw}W_i = W' - W_v \land{}\\
        &\phantom{{}={}}\exists v' \in f^{-1}(v): \forall i \in [\tw]: p_{i, W_i}((R\cup \{v \mapsto v'\})|_{X_{c(i)}})
      \end{align*}
      Analogously to the computation of \(p_{\text{res}}[R]\) in the proof of part 1 of Theorem~\ref{upper-bound-detection}, we have
      \begin{align*}
        \hspace{-0.4cm}= \exists W_v: d_{t'',W_v}(R)\land \exists W_1, \ldots, W_{\tw}: \sum_{i=1}^{\tw}W_i = W' - W_v \land  MP_{\tw}(p_{1, W_1}, \ldots, p_{\tw, W_{\tw}})[R]
      \end{align*}
      We now show how to compute tensors \(D_{W'}\) for each \(W'\in \mathcal{W}\) such that \(D_{W'}(R) = \exists W_1, \ldots, W_{\tw}: \sum_{i=1}^{\tw}W_i = W' - W_v \land  MP_{\tw}(p_{1, W_1}, \ldots, p_{\tw, W_{\tw}})[R]\), which would simplify the above to
      \begin{align*}
       \hspace{-0.4cm}= \exists W_v: d_{t'',W_v}(R)\land D_{W' - W_v}(R)
      \end{align*}
      and hence make it computable via a boolean convolution in time \(O(n^{\tw}W \log W\poly(k))\) (given the tensors \(D_{W'}\)).

     We want to compute $D_{\widetilde{W}}$ for each $\widetilde{W} \in \mathcal{W}$. Note how, when defining the OR of tensors to be calculated entrywise, we have
  \begin{align*}
D_{W'} &= \bigvee_{\substack{W_1, \ldots, W_{\tw} \in \mathcal{W}\\W_1 + \ldots + W_{\tw} = W'}}\MP_{\tw}(p_{1,W_1}, \ldots, p_{\tw,W_{\tw(H)}})
  \end{align*}
  
  Indeed, the right-hand side can be calculated using a single $\tw$-wise matrix product of tensors with polynomials as entries. Let $T_{\tw}$ be the group of order-$\tw(H)$ tensors of dimensions $n \times \ldots \times n$ with entries from $\{\mathrm{true}, \mathrm{false}\}$. Define $T_{\tw}[X,X^{-1}]$ as the group of Laurent polynomials with elements of $T_{\tw(H)}$ as coefficients. Note how elements of $T_{\tw}[X,X^{-1}]$ may also be viewed as tensors with Laurent polynomials as entries, or as functions $f: \{0, 1\} \to T_{\tw(H)}$ when using the usual definition of scalar-tensor AND.

  We define $p_{i}(X) := \sum_{j\in \mathcal{W}}\,p_{i,j}\cdot X^{j} \in T_{\tw(H)}[X,X^{-1}]$ for all $i$. Viewing them as tensors of polynomials, we may calculate their $\tw(H)$-wise matrix product. Viewing $\MP_{\tw}(p_{1}, \ldots p_{\tw})$ as polynomial again, it can be easily seen that the coefficient tensor for $X^{W'}$ is exactly
  \begin{align*}
  \sum_{\substack{W_1, \ldots, W_{\tw} \in \mathcal{W}\\W_1 + \ldots + W_{\tw} = W'}}\MP_{\tw}(p_{1, W_1}, \ldots, p_{\tw, W_{\tw}})
  \end{align*}
Thus, to compute $D_{W'}$ for all \(W'\), it suffices to compute $\MP_{\tw}(p_{1}, \ldots, p_{\tw})$.

Hence, we have reduced the problem of calculating the tensors $D_{\widetilde{W}}$ to computing the $\tw(H)$-wise matrix product of tensors whose entries are Laurent polynomials of degree bounded (in both directions) by $O(W\poly(k))$. By Lemma~\ref{tensor-fft}, this can be done in time $O((n^{\tw(H)+1}W + n^{\tw(H)}W\log W)poly(k))$.

Hence we have an overall running time of $O((n^{\tw(H)+1}W + n^{\tw(H)}W\log W)poly(k))$ per intermediate-result node.
\end{description}

Again, since the tree decomposition has \(\poly(k)\) nodes, the overall time the algorithm takes is $O((n^{\tw(H)+1}W + n^{\tw(H)}W\log Wpoly(k))$.

After calculating $d_{t',W'}$ for each $t' \in T, W' \in \mathcal{W}$, outputting the result is simple. We simply extract the answer from the DP function of the root node. Let \(r\) be the root of \(T\). We output YES if and only if $\exists R: d_{r,-w(R)}(R) = \mathrm{true}$. By the definition of $d_{t',-w(R)}$, this is the case if and only if $R$ is a partial solution of $X_r$ in the cone $V_{r} = V(H)$ with an extension of weight $-w(R)$. This is the case if and only if $R$ can be expanded to a configuration for all of $V(H)$ such that its total weight is $-w(R)+w(R) = 0$. Hence, $\exists R: d_{t',-w(R)}(R) = \mathrm{true}$ if and only if the instance has a solution. Checking whether such an $R$ exists can obviously be done in time $O(n^{\tw(H)}\poly(k))$.

\end{proof}

  We now consider the second part of the theorem. Similarly to the weighted case, it suffices to employ the algorithm from part 1.

\begin{proof}[Proof (of part 2 of Theorem~\ref{upper-bound-weighted})]
  In the algorithm for part 1, all running times except for the $\tw(H)$-wise matrix product are bounded by $O(n^{\tw}W\log W \poly(k))$. For $\tw = 2$, the $\tw$-wise matrix product is simply matrix multiplication, which can be done in $O(n^\omega)$. Thus, for $\tw=2$ the running time is $O(n^\omega W \poly(k)+n^2W\log W \poly(k))$.
\end{proof}

Finally, we come to the third part of the theorem. For reasons outlined in the proof of the previous theorem, the algorithm from part 1 does not work for $\tw(H)=1$. Again, however, the result turns out to be quite simple for this case.

\begin{proof}[Proof (of part 3 of Theorem~\ref{upper-bound-weighted})]
   Again, we only sketch the result, since it is easy to see. A simple dynamic programming algorithm on trees can be applied, storing for each node which configurations together with which weights can be achieved.
  
\end{proof}

\subsection{Colored Subgraph Isomorphism for Bounded Pathwidth}\label{sec:algo-col-subiso-bounded-pathwidth}

Surprisingly, \textsc{Colored Subgraph Isomorphism} can be solved slightly faster on graphs of bounded pathwidth. This algorithm leverages rectangular matrix multiplication.

 We briefly discuss (rectangular) matrix multiplication and current algorithms solving it. Multiplication of $n \times n$ by $n\times n$ matrices is perhaps the most ubiquitous open problem in computer science, with the central question being whether it can be done in \(O(n^2)\) time. The matrix multiplication exponent $\omega$ has been slowly inching toward, but not quite reaching, a value of 2 over the last few decades. In our algorithms, however, we also multiply rectangular matrices. It turns out that the techniques used in these fast algorithms for square matrix multiplication can also be generalized to the rectangular case. In the following, let $\MM(s,r,t)$ denote the time needed to multiply a matrix of size $r\times s$ with a matrix of size $s\times t$. We are mostly interested in the case that $s=n, r=n$ and $t=n^k$ for some $k\in \mathbb{R}^+$.

A simple, well-known result is that $\MM$ is both convex and symmetrical in its arguments (see e.g.~\cite{lotti1983asymptotic,stothers2010complexity}). In particular, we have $\forall x \in \mathbb{R}^+: \MM(n,n,n^{k+x}) \leq n^x\cdot MM(n,n,n^k)$ and $\MM(n,n,n^k) = \MM(n,n^k,n)$. Letting $\omega(k) := \log_n MM(n,n,n^k)$, we immediately get $\forall k \geq 1: \omega(k) < k + 1.373$ via the current bounds of $\omega(1) = \omega < 2.373$~\cite{le2014powers}.

We can, however, do better: Le Gall~\cite{gall2018improved} has shown that there are faster algorithms based on the Coppersmith-Winograd method~\cite{coppersmith1987matrix, le2014powers} used for square matrix multiplication. Among other values, he shows $\omega(0.31) = 2$, $\omega(2) < 3.26$, $\omega(3) < 4.2$, $\omega(4) < 5.18$ and $\omega(5) < 6.16$ (see~\cite{gall2018improved} for an extensive table of such values).

We now show how rectangular matrix multiplication can be used for bounded-pathwidth pattern graphs.

\begin{theorem}\label{upper-bound-detection-pathwidth}
  There is an algorithm which, given an arbitrary instance $\phi = (H,G,f)$ of \textsc{Colored Subgraph Isomorphism}, solves $\phi$ in time
  \begin{enumerate}
  \item $O(n^{\omega(\pw(H)-1)}\poly(k) + g(k))$ when $\pw(H)\geq 2$, and
  \item $O(n^2\poly(k)) + g(k)$ when $\pw(H)=1$
  \end{enumerate}
  where $k := |H|$, $n$ is the size of the preimages of $f$, and $g$ is a computable function.
\end{theorem}

This theorem implies Theorem~\ref{corollary-upper-bound-detection-pathwidth} from the results section.

\begin{proof}
  \textbf{Part 2} of the theorem is trivial, since for any graph, its treewidth is smaller than its pathwidth. Hence, by application of Theorem~\ref{upper-bound-detection}, we achieve the desired running time.
  
  \medskip
  
  For \textbf{part 1}, let $G$, $H$ and $f: V(G) \to V(H)$ be given. The algorithm first computes an optimal path decomposition $\mathcal{P} = (P, \{X_t\}_{t\in V(T)})$ of $H$ in time $g(k)$ (see preliminaries), then does dynamic programming over $\mathcal{P}$.

  To unify nomenclature and notation with the case of treewidth, we talk about a path as a tree rooted at one of its endpoints. As in the proof of Theorem~\ref{upper-bound-detection}, we modify the path decomposition $\mathcal{P}$ to satisfy certain properties. Specifically, we wish to obtain the following properties:
  \begin{enumerate}
  \item The bags of the root and leaf of the path have size $\pw(H)$, and every other bag has size $\pw(H)+1$
  \item For every $t \in V(P)$ with child $t'$, we have $|X_t \cap X_{t'}| = \pw(H)$
  \end{enumerate}
  These properties can be obtained with similar techniques as described in the proof of part 1 of Theorem~\ref{upper-bound-detection}.

  We now do dynamic programming on this modified path decomposition. We only store values for each configuration of the separators $X_t \cap X_{t'}$. In particular, let $t'$ be a node other than the root, and let $t$ be its parent. As in Theorem~\ref{upper-bound-detection}, we only store
  \begin{align*}
    d_{t'} &: \mathcal{Conf}(X_t \cap X_{t'}) \to \{\mathrm{true},\mathrm{false}\}\\
    d_{t'}(R) &:= \ParSol(R; X_t \cap X_{t'}; V_{t'})\tag{1}
  \end{align*}
  We calculate these functions bottom-up. The case that $t'$ is the leaf is analogous to the corresponding case in Theorem~\ref{upper-bound-detection}, taking time $O(n^{\pw(H)}\poly(k))$.

  Now let $t'$ with parent $t$ and child $t''$ be an inner node of $P$. We define \(\hat{v}\) to be the unique element with $\hat{v} \in X_{t'} \setminus X_t$ and similarly, $\hat{w} \in X_{t'} \setminus X_{t''}$ (or, if we would have \(\hat{v} = \hat{w}\), we take \(\hat{w}\) to be some vertex from \(X_{t'}\setminus \{\hat{v}\}\) instead), and finally $E := X_{t'} \setminus \{\hat{v},\hat{w}\}$. Now if \(\hat{v}\hat{w} \notin E(H)\), the calculation is easy.
  \todo{details}
  Hence assume \(\hat{v}\hat{w} \in E(H)\). For a configuration $R$ of $X_t \cap X_{t'}$, we get the following alternate characterization of $d_{t'}(R)$:
  \begin{align*}
    d_{t'}(R) = &\ValConf(R; \{\hat{w}\}\cup E) \land {}\\
    &\exists v' \in f^{-1}(\hat{v}): \ParSol((R \cup \{\hat{v} \mapsto v'\})|_{X_{t'}\setminus\{\hat{w}\}}; \{\hat{v}\} \cup E; V_{t''}) \land v'R(w) \in E(G)
  \end{align*}

  We describe how to calculate $d_{t'}$ via rectangular matrix multiplication. Much like in Theorem~\ref{upper-bound-detection}, the matrices are indexed by configurations. In contrast to the former, however, one of the two dimensions of the matrix might correspond to the configuration of multiple vertices. Formally, for a vertex subset $Y \subseteq V(H)$ and a vertex $x \in V(H), x \notin Y$, we call a matrix $A$ \textbf{indexed by configurations of $(x,Y)$} when it is of dimensions $n \times n^{|Y|}$. We use two arbitrary bijections $g_{\{x\}}: \mathcal{Conf}(\{x\}) \to [n]$ and $g_Y: \mathcal{Conf}(Y) \to [n^{|Y|}]$ which will help us map configurations of $x$ and $Y$ to indices of the matrix. Hence, for a configuration $R$ of $\{x\} \cup Y$, we define $A[R] := A[g_x(R|_{\{x\}}),g_Y(R|_Y)]$.

  For our rectangular matrix product, we define a $n\times n^{\pw(H)-1}$ matrix $B_{t'}$ indexed by configurations of $(\hat{v},E)$. For a configuration $R'$ of $\{\hat{v}\} \cup E$, we define
  \begin{align*}
    B_{t'}[R'] := d_{t''}(R') = \ParSol(R'; \{\hat{v}\} \cup E;V_{t''})
  \end{align*}
  Now consider the adjacency matrix $Adj_{\hat{v},\hat{w}}$ of $f^{-1}(\hat{v})$ and $f^{-1}(\hat{w})$, indexed by configurations of $(\hat{w},\{\hat{v}\})$. We make sure that the indexing bijection $g_{\{\hat{v}\}}$ as defined above is the same for both $Adj_{\hat{v},\hat{w}}$ and $B_{t'}$ and then calculate the matrix product $Adj_{\hat{v},\hat{w}}\cdot B_{t'}$. Naturally, the product is indexed by configurations of $(\hat{w},E)$ and can be expressed as
  \begin{align*}
    &(Adj_{\hat{v},\hat{w}}\cdot B_{t'})[R']\\
    &\phantom{stuff}= \exists v' \in f^{-1}(\hat{v}): \ParSol((R' \cup \{\hat{v} \mapsto v'\})|_{X_{t'}\setminus\{\hat{w}\}}; \{\hat{v}\} \cup E;V_{t''}) \land v'R(w) \in E(G)
  \end{align*}
  Hence, we may write $d_{t'}(R)$ as
  \begin{align}\label{dp-formula-pathwidth}
    d_{t'}(R) = \ValConf(R;\{\hat{w}\} \cup E) \land (Adj_{\hat{v},\hat{w}}\cdot B_{t'})[R]
  \end{align}
  
  The algorithm to calculate $d_{t'}$ is immediate. First, we calculate the rectangular matrix product of $Adj_{\hat{v},\hat{w}}$ and $B_{t'}$ in time $O(n^{\omega(\pw(H)-1)})$, then calculate $d_{t'}$ via formula~\ref{dp-formula-pathwidth}. Checking whether $R$ is a valid configuration of $\{\hat{w}\} \cup E$ can be done in time $\poly(\pw(H)) \leq \poly(k)$, giving us a total time of at most  $O(n^{\omega(\pw(H)-1)}\poly(k))$ per inner node.

  To output the answer, consider the child $r'$ of the root $r$. We have by definition that $d_{r'}(R)$ is 1 if and only if $R$ is a partial solution for $X_r$ in $V(H)$. Thus, the input is a YES-instance for \textsc{Colored Subgraph Isomorphism} if and only if there is an $R$ such that $d_{r'}(R)$ is 1.

  Since there is only a single leaf and $\poly(k)$ inner nodes, total running time of the dynamic programming algorithm is $O(n^{\omega(\pw(H)-1)}\poly(k))$.
\end{proof}

\subsection{Exact Weight Colored Subgraph Isomorphism for Bounded Pathwidth}

The techniques of using rectangular matrix multiplication for the case of pathwidth can also be applied to the weighted case, and they lead to improvements in the expected way. Again, the instances may be either node- or edge-weighted.

\begin{theorem}\label{upper-bound-weighted-pathwidth}
  There is an algorithm which, given an arbitrary instance $\phi = (H,G,f,w)$ of the \textsc{Exact Weight Colored Subgraph Isomorphism} problem, solves $\phi$ in time
  \begin{enumerate}
  \item $O((n^{\omega(\pw(H)-1)}W + n^{\pw(H)}W\log W)\poly(k) + g(k))$ when $\pw(H) \geq 2$, and
  \item $O((n^2 W + nW\log W)\poly(k) + g(k))$ when $\pw(H) = 1$
  \end{enumerate}
  where $k := |V(H)|$, $n$ is the size of the preimages of $f$, $W$ is the maximum absolute weight in the image of $w$, and $g$ is a computable function.
\end{theorem}

\noindent
From this, Theorem~\ref{corollary-upper-bound-weighted-pathwidth} from the results section follows directly via Lemma~\ref{lemma-equiv}.

\begin{lemma}\label{lemma-fft-rectangular-mm}
  Let two matrices $A,B$ of dimensions $r\times s$ and $s\times t$ be given, such that each of their entries $A_{i,j}, B_{j,k} \in \mathbb{C}[X, X^{-1}]$ (for all $i \in [r], j \in [s], k \in [t]$) is a Laurent polynomial of degree bounded by $W$ in both the positive and the negative direction. Then their product can be computed in time $O(MM(r,s,t)W + (rs + st)W \log W)$
\end{lemma}

\begin{proof}
  Analogous to Lemma~\ref{tensor-fft}.
\end{proof}

\begin{proof}[Proof (of theorem~\ref{upper-bound-weighted-pathwidth})]
Again, we only need to prove an algorithm for the edge-weighted case due to Proposition~\ref{prop:node-to-edge-weight}.
  
  \textbf{Part 2} is a corollary of Theorem~\ref{upper-bound-weighted}, since for any graph, its treewidth is smaller than its pathwidth.

  \medskip

  For \textbf{part 1}, we only sketch the proof, since it is a straightforward combination of the techniques used in the proofs of Theorems~\ref{upper-bound-weighted} and~\ref{upper-bound-detection-pathwidth}. We use notation and phrasing from both of those proofs without further mention.

  Given $G, H$ and $f$, we compute an optimal path decomposition $(P, \{X_t\}_{t\in V(P)})$ for $H$ and do dynamic programming on the modified path decomposition as described in the proof of Theorem~\ref{upper-bound-detection-pathwidth}. For each non-root node $t' \in P$ with parent $t$, for each weight $W' \in \mathcal{W}$ and for each configuration $R$ of $X_t \cap X_{t'}$, we store
  \begin{align*}
    d_{t',W'} := \ParSolE(R; X_t \cap X_{t'}; V_{t'}; W')
  \end{align*}

  The case that $t'$ is a leaf is clear. For the case that $t'$ is an inner node with parent $t$ and child $t''$, let \(\hat{v},\hat{w}\) be the unique elements with $\hat{v} \in X_{t'} \setminus X_t, \hat{w} \in X_{t'}\setminus X_{t''}$ (or, if we would have \(\hat{v} = \hat{w}\), we take \(\hat{w}\) to be some vertex from \(X_{t'}\setminus \{\hat{v}\}\) instead) and $E = X_{t'} \setminus \{\hat{v},\hat{w}\}$. If \(\hat{v}\hat{w} \notin E(H)\), the calculation is easy, hence assume \(\hat{v}\hat{w} \in E(H)\). For each weight $W'$, we build a rectangular matrix $A_{t'}^{W'}$ indexed by $(\hat{v},E)$. The entry corresponding to a configuration $R$ of $\{\hat{v}\} \cup E$ tells us whether $R$ has an extension of weight $W' + x$, where $x$ is the weight contributed by $\hat{v}$ in $\{\hat{v}\} \cup E$. In particular $x := \sum_{\hat{u} \in E} w(R|_{\{\hat{u},\hat{v}\}})$.

  We also use, for each weight $W'$, an adjacency matrix $Adj_{\hat{v},\hat{w}}^{W'}$ defined for node-weighted instances as
  \begin{align*}
    Adj_{\hat{v},\hat{w}}^{W'}[v',w'] = v'w' \in E(G) \land w(v'w') = W'
  \end{align*}

  We then create two Laurent polynomials of degree $|\mathcal{W}|$, one with the matrices $Adj_{\hat{v},\hat{w}}^{W'}$ as coefficients, one with the matrices $A_{t'}^{W'}$. These can also be seen as matrices with Laurent polynomials as entries. Using Lemma~\ref{lemma-fft-rectangular-mm}, we then calculate their matrix product, which tells us for each weight $W'$ and each configuration $R$ of $\{\hat{w}\} \cup E$ if $R$ is a potential solution of $\{\hat{w}\} \cup E$ in $V_{t'}$ with an extension of weight $W'$, potentially missing edges between the preimage of $w$ and the preimages of $E$. The latter can be checked for each entry. This leads to the desired running time.

\end{proof}

\subsection{Improvements for the Node-Weighted Case}

In the case of node weights instead of edge weights, some of the algorithms can be slightly improved using rectangular matrix multiplication. However, these improvements only work for the case that that the treewidth of $H$ is 1 and for the case of bounded pathwidth.

Specifically, we show the following two results, which imply Theorems~\ref{corollary-node-weighted-algo-trees} and~\ref{corollary-node-weighted-algo-pathwidth} from the results section.

\begin{theorem}\label{node-weighted-algo-trees}
  There is an algorithm which, given an arbitrary instance $\phi = (H,G,f,w)$ of the node-weighted \textsc{Exact Weight Colored Subgraph Isomorphism} problem where $H$ is a tree, solves $\phi$ in time $O(\MM(n,n,W)\poly(k) + nW\log W\poly(k))$.
\end{theorem}

\begin{theorem}\label{node-weighted-algo-pathwidth}
  There is an algorithm which, given an arbitrary instance $\phi = (H,G,f,w)$ of the node-weighted \textsc{Exact Weight Colored Subgraph Isomorphism} problem, solves $\phi$ in time $O(MM(n,n,n^{\pw(H)-1}W)\poly(k) + g(k))$.
\end{theorem}

\begin{proof}[Proof (of Theorem~\ref{node-weighted-algo-trees})]
  Let $G,f: V(G) \to V(H)$ and $w: V(G) \to \mathbb{Z}$ be given. All achievable total weights must lie in $\mathcal{W} := \{-kW, \ldots, kW\}$.

  We do dynamic programming on the tree $H$. Rooting $H$ in an arbitrary vertex, we define $T_u$ to be the subtree rooted at a node $u \in H$. For each weight $W \in \mathcal{W}$, each node $v \in H$, we store the following function of finite domain:
  \begin{align*}
    d_{v,W}&: \mathcal{Conf}(\{v\}) \to \{\mathrm{true},\mathrm{false}\}\\
    d_{v,W}(R) &:= \ParSolE(R; \{v\}; T_v; W)
  \end{align*}
  We calculate the entries of these functions bottom-up, starting at the leaves of $H$. If $v$ is a leaf, $d_{v,W}(R)$ is 1 if and only if $W = 0$. Now suppose $v$ is a non-leaf node of $H$ with set of children $C_v$. For each child $u \in C_v$, we construct a rectangular matrix $p_{v}^u$. This matrix is indexed by $w$ and $\mathcal{W}$. Formally, we call a matrix $A$ \textbf{indexed by $w$ and $\mathcal{W}$} if it has dimensions $n\times |\mathcal{W}|$. We use two arbitrary bijections $g_{\{v\}}: \mathcal{Conf}(\{v\}) \to [n]$ and $g_{\mathcal{W}}: \mathcal{W} \to [|\mathcal{W}|]$ to help us map weights from $\mathcal{W}$ and configurations of $v$ to indices of $A$. Correspondingly, we define $A[R,W] := A[g_v(R),g_W(W)]$.

  We define $p_{v}^u$ as
  \begin{align*}
    p_{v}^u[R,W] := \ParSolE(R;\{v\};\{v\} \cup T_u;W)
  \end{align*}
  We may calculate this as follows. Let $Adj_{\{v,u\}}$ be the adjacency matrix of $f^{-1}(v)$ and $f^{-1}(u)$ indexed by $(v,\{u\})$ (as defined in the proof of Theorem~\ref{upper-bound-detection-pathwidth}), and let $\widetilde{d}_u$ be the rectangular matrix indexed by $u$ and $\mathcal{W}$ and defined as $\widetilde{d}_u[R,W] := d_{u,W-w(R)}(R)$. We make sure that the indexing bijection $g_{\{u\}}$ as defined above is the same for both $Adj_{v,u}$ and $\widetilde{d}_u$ and then calculate the matrix product $Adj_{v,u}\cdot \widetilde{d}_u$. The product is indexed by configurations of $v$ and $\mathcal{W}$ and can be expressed as
  \begin{align*}
    (Adj_{v,u}\cdot \widetilde{d}_u)[R,W] &= \exists u' \in f^{-1}(u): R(v)u' \in E(G) \land \ParSolE(R;\{u\}; T_u; W-w(u'))\\
    &= \ParSolE(R;\{v\};\{v\} \cup T_u; W)\\
    &= p_v^u[R,W]\\
\intertext{Writing $C_v = \{u_1, \ldots, u_c\}$ with $c = |C_v|$, the function $d_{v,W}$ may then be expressed as}
    d_{v,W}(R) &= \exists W_1, \ldots, W_c: \sum_{i=1}^c W_i = W \land \forall i: p_v^{u_i}[R,W_i]
  \end{align*}
  Similarly to the proof of Theorem~\ref{upper-bound-weighted}, this may be calculated using a boolean convolution. Accordingly, we define for each configuration $R$ of $v$ the finitely supported functions $f_{v,R}^1, \ldots, f_{v,R}^c: \mathbb{Z} \to \{\mathrm{true}, \mathrm{false}\}$ as $f_{v,R}^i(x) := p_v^{u_i}(R)$ if $x \in \mathcal{W}$, and 0 otherwise. By a simple calculation, we get $d_{v,W}(R) = (f_{v,R}^1\ast \ldots \ast f_{v,R}^c)(W)$.

  Finally, after having calculated all values of $d_{t',W}(R)$ for all $t',W$ and $R$, we wish to output the result. Let $r$ be the root of $H$. By definition of $d_{r,W}$, there is some configuration $R$ of $r$ such that $d_{r,-w(R)}(R) = \mathrm{true}$ if and only if the instance has a solution.

  It remains to analyze the running time. For the leaves of $H$, the calculation takes time $O(nW)$. For inner nodes, the calculation of the matrix product $Adj_{\{v,u\}} \cdot \widetilde{d}_u$ takes time $\MM(n,n,W)$. Finally, calculating the discrete convolutions takes time $O(nW\log W)$, since any vertex of $v$ is involved as a child in at most one discrete convolution. Hence, we arrive at the running time from the theorem.  
\end{proof}

\begin{proof}[Proof (of Theorem~\ref{node-weighted-algo-pathwidth})]
  The proof uses a combination of the techniques from the proofs of Theorem~\ref{node-weighted-algo-trees} and Theorem~\ref{upper-bound-detection-pathwidth}.

  Let $G, H, f:V(G) \to V(H)$ and $w: V(G) \to \mathbb{Z}$ be given. We compute an optimal path decomposition in time $g(k)$, modify it as described in Theorem~\ref{upper-bound-detection-pathwidth}, obtaining a modified path decomposition $\mathcal{P} = (P, \{X_t\}_{t \in P})$, and then do dynamic programming on $\mathcal{P}$. Note that all achievable weights must lie in $\mathcal{W} := \{-kW, \ldots, kW\}$.

  We store, for each non-root node $t'$ with parent $t$ of $P$, the following function of finite domain:
  \begin{align*}
    d_{t',W}(R) := \ParSolE(R;X_t \cap X_{t'}; V_{t'};W)
  \end{align*}

  For $t'$ a leaf, the calculations is clear. Let $t'$ be an inner node with parent $t$ and child $t''$ and define \(\hat{v}\) to be the unique element with $\hat{v} \in X_{t'} \setminus X_t$ and similarly $\hat{w} \in X_{t'} \setminus X_{t''}$ (or, if we would have \(\hat{v} = \hat{w}\), we take \(\hat{w}\) to be some arbitrary vertex from \(X_{t'}\setminus \{v\}\) instead) and finally $E := X_{t'} \setminus \{\hat{v},\hat{w}\}$. If \(\hat{v}\hat{w} \notin E(H)\), the calculation is easy, hence assume \(\hat{v}\hat{w} \in E(H)\). For a configuration $R$ of $X_t \cap X_{t'}$, we get the following alternate characterization of $d_{t',W}(R)$:
  \begin{align*}
    d_{t',W}(R) =& \ValConf(R; \{\hat{w}\} \cup E) \land \exists v' \in f^{-1}(\hat{v}): v'R(\hat{w}) \in E(G)\ \land {} \\
                 &\ParSolE((R \cup \{\hat{v} \mapsto v'\})|_{X_{t'}\setminus \{\hat{w}\}}; \{\hat{v}\} \cup E; V_{t''}; W-w(v'))
  \end{align*}

  We now set up our rectangular matrix product. For a vertex $x \in V(H), x \notin E$, we call a matrix $A$ indexed by $x,E$ and $\mathcal{W}$ if it has dimensions $n\times n^{|E|} |\mathcal{W}|$. We use two arbitrary bijections $g_{\{x\}}: \mathcal{Conf}(\{x\}) \to [n]$ and $g_{E,\mathcal{W}}: \mathcal{Conf}(E) \times \mathcal{W}) \to [n^{|E|}|\mathcal{W}|]$ to help index the matrix and define, for a configuration $R$ of $\{x\} \cup E$, $A[R,W] := A[g_x(R|_{x}),g_{E,\mathcal{W}}(R|_E,W)]$.

  We define the $n\times n^{\pw(H)-1}|\mathcal{W}|$ matrix $B_{t'}$, which is to be indexed by $\hat{v},E$ and $\mathcal{W}$, as follows:
  \begin{align*}
    B_{t'}[R',W'] := d_{t',W'-w(R'|_{\hat{v}})}(R')
  \end{align*}
  As in the proof of Theorem~\ref{upper-bound-detection-pathwidth}, we also use the adjacency matrix $Adj_{\hat{v},\hat{w}}$ indexed by configurations of $(\hat{w},\{\hat{v}\})$. Ensuring that the indexing bijections $g_{\{\hat{v}\}}$ are the same for both $B_{t'}$ and $Adj_{\hat{v},\hat{w}}$, we calculate $Adj_{\hat{v},\hat{w}}\cdot B_{t'}$ and obtain a matrix indexed by $\hat{w},E$ and $\mathcal{W}$. Its entries can be expressed as
  \begin{align*}
    (Adj_{\hat{v},\hat{w}}\cdot B_{t'})[R',W'] = &\exists v' \in f^{-1}(\hat{v}): v'R'(w) \in E(G)\land {}\\
    &\ParSolE((R' \cup \{\hat{v} \mapsto v'\})|_{X_{t'}\setminus \{\hat{w}\}}; \{\hat{v}\} \cup E; V_{t''}; W'-w(v'))
  \end{align*}
  Thus, $d_{t'}(R)$ can be expressed as
  \begin{align*}
    d_{t',W}(R) = \ValConf(R;\{\hat{w}\} \cup E) \land (Adj_{\hat{v},\hat{w}}\cdot B_{t'})[R,W]
  \end{align*}

  This concludes the description of the computation of $d_{t',W}$. For the computation of the answer, consider the child $r'$ of the root $r$. By definition of $d_{r',W}(R)$, we have that there exists a configuration $R$ such that $d_{t',-w(R|_{X_{r}}}(R) = \mathrm{true}$ if and only if the instance has a solution.

  By a simple argument, this dynamic programming algorithm has a running time of $O(MM(n,n,n^{\pw(H)-1}W)\poly(k))$.
  
\end{proof}

\section{Interconnections Between Subgraph Isomorphism, Boolean \(k\)-Wise Matrix Products and Hyperclique}\label{sec:interconnections}

We have seen in the proof of the conditional lower bound for the unweighted \textsc{Colored Subgraph Isomorphism} problem that for any \(t\in\mathbb{N}\) there is a polynomial-time reduction from \(h\)-uniform \(h(t+1)\)-\textsc{Hyperclique} with \(n\) nodes to \textsc{Colored Subgraph Isomorphism} on pattern graphs of treewidth \(t\) with \(n^h\) nodes. Hence we have the following corollary, which follows directly from the~\hyperref[lem:unweighted-lemma]{Unweighted Lemma}.
\begin{corollary}\label{cor:col-subiso-to-hyperclique}
  If there is a \(t\geq 3\) such that the \textsc{Colored Subgraph Isomorphism} problem for pattern graphs of treewidth \(t\) can be solved in time \(O(n^{t+1-\varepsilon})\) (for some \(\varepsilon > 0\)), then for any \(3\leq h\leq t\) the \(h\)-uniform \(h(t+1)\)-\textsc{Hyperclique} problem can be solved in time \(O(n^{h(t+1)-\varepsilon'})\) (for some \(\varepsilon' > 0\)).
\end{corollary}

But what about the other direction? Can we also provide a conditional lower bound for \textsc{Hyperclique} under the hypothesis that \textsc{Colored Subgraph Isomorphism} cannot be solved faster, hence proving an equivalence? Indeed we can! However, the ``equivalence'' we get is not as strong as one might hope.

Specifically, it turns out that the algorithm we described for \textsc{Colored Subgraph Isomorphism} already gives a Turing reduction from \textsc{Colored Subgraph Isomorphism} with treewidth \(t\) to the \textsc{Boolean \(t\)-wise Matrix Product} problem\footnote{Defined as: Given \(t\) tensors \(A^1, \ldots, A^t\) of order \(t\) with dimensions \(n\times \ldots \times n\), calculate their boolean \(t\)-wise matrix product \(\MP_t(A^1, \ldots, A^t)\).}. Indeed, it can also be seen that the \textsc{Boolean \(t\)-wise Matrix Product} problem is equivalent to the \(t\)-uniform \((t+1)\)-\textsc{Hyperclique} problem. More formally, we have the following two lemmas.
\begin{lemma}
  For any \(t\geq 3\), if the \textsc{Boolean \(t\)-wise Matrix Product} problem can be solved in time \(O(n^{t+1-\varepsilon})\), then \textsc{Colored Subgraph Isomorphism} problem on pattern graphs of treewidth \(t\) can be solved in time \(O(n^{t+1-\varepsilon})\).
\end{lemma}
\begin{proof}
  As was already mentioned in the algorithms for \textsc{Colored Subgraph Isomorphism} (see the proof of Theorem~\ref{upper-bound-detection}), the boolean \(\tw(H)\)-wise matrix product is the bottleneck for the running time. All other operations run in time \(O(n^{\tw(H)}\poly(|V(H)|))\). Hence a \(O(n^{t+1-\varepsilon})\) algorithm for \textsc{Boolean \(t\)-wise Matrix Product} translates directly to a \(O(n^{t+1-\varepsilon})\) algorithm for \textsc{Colored Subgraph Isomorphism} with \(\tw(H) = t\).
\end{proof}
\begin{lemma}\label{lem:hyperclique-to-k-wise-mp}
  If the \(t\)-uniform \((t+1)\)-\textsc{Hyperclique} problem can be solved in time \(O(n^{t+1-\varepsilon})\) (for some \(\varepsilon > 0\)), then the \textsc{Boolean \(t\)-wise Matrix Product} problem can be solved in time \(O(n^{t+1-\varepsilon'})\) (for some \(\varepsilon' > 0\)).
\end{lemma}
We defer the proof of Lemma~\ref{lem:hyperclique-to-k-wise-mp} to Appendix~\ref{apx:boolean-to-hyperclique}. Indeed, in appendix~\ref{apx:hyperclique-equiv-boolean}, we also show the other direction, i.e.\ we show that the existence of fast algorithms for these two problems is actually equivalent. This is a natural generalization of methods from~\cite{williams2010subcubic}, where this result is proven for combinatorial algorithms for the case \(t=2\).

\noindent
Composing these two lemmas, we get the following theorem.
\begin{theorem}\label{thm:hyperclique-to-col-subiso}
  If the \(t\)-uniform \((t+1)\)-\textsc{Hyperclique} problem can be solved in time \(O(n^{t+1-\varepsilon})\) (for some \(\varepsilon > 0\)), then the \textsc{Colored Subgraph Isomorphism} problem on pattern graphs of treewidth \(t\) can be solved in time \(O(n^{t+1-\varepsilon'})\) (for some \(\varepsilon' > 0\)).
\end{theorem}

Looking at Corollary~\ref{cor:col-subiso-to-hyperclique} and Theorem~\ref{thm:hyperclique-to-col-subiso}, we have reductions in both directions, but they do not give a full equivalence. This is because the reduction from Corollary~\ref{cor:col-subiso-to-hyperclique} only gives an algorithm for \(h\)-uniform \(h(t+1)\)-\textsc{Hyperclique} (for any \(h\geq 3\)), but Theorem~\ref{thm:hyperclique-to-col-subiso} needs an algorithm for the much ``denser'' \(t\)-uniform \((t+1)\)-\textsc{Hyperclique} problem.

\todo{maybe stuff about \(\exists r,s,t\)-equivalence and why this doesn't work due to input size problems}

\todo{maybe stuff about weighted equivalence}

\todo{if not weighted equivalence, remove from title of this section -- done}

\section{The Colored Problems are Equivalent to the Uncolored Problems}\label{section-equivalence}

We now show Lemma~\ref{lemma-equiv} from the preliminaries (restated in a modified form below), which shows that the \textsc{(Exact Weight) Subgraph Isomorphism} and \textsc{(Exact Weight) Colored Subgraph Isomorphism} problems are essentially equivalent with respect to running times. Hence, for most of our purposes, we can treat them as equal, which simplifies both the proofs of the algorithms and the lower bounds since the colored version is much more structured.

The reductions from \textsc{(Exact Weight) Subgraph Isomorphism} to \textsc{(Exact Weight) Colored Subgraph Isomorphism} and the reduction from \textsc{Exact Weight Colored Subgraph Isomorphism} to \textsc{Exact Weight Subgraph Isomorphism} leaves $H$ unmodified. The reduction from \textsc{Colored Subgraph Isomorphism} to \textsc{Subgraph Isomorphism}, however, modifies $H$ in such a way that preserves treewidth, but may modify pathwidth.

We say that \textsc{Subgraph Isomorphism} or \textsc{Colored Subgraph Isomorphism} have a $T(n,k,\rho(H))$ algorithm (for some graph parameter $\rho$) if there is an algorithm $\mathcal{A}$ which decides a given instance $\phi = (H,G,f)$ of either problem in time \(T(n,k,\rho(H))\). Analogously, we define the phrase that \textsc{Exact Weight Subgraph Isomorphism} or \textsc{Exact Weight Colored Subgraph Isomorphism} has a $T(n,k,\rho(H),W)$ algorithm, the only difference being that $\phi = (H,G,f,w)$. $W$ denotes the maximum absolute value of the weight function $w$.

Parts 1 and 2 of the lemma follows directly from the Color Coding technique~\cite{alon1995color}.

\begin{lemma}[reformulation of Lemma~\ref{lemma-equiv}]
  Let $\rho$ be any graph parameter.
  \begin{enumerate}
  \item If there is a $T(n,k,\rho(H))$ time deterministic algorithm for \textsc{Colored Subgraph Isomorphism}, then there is a $O(T(k n,k,\rho(H))g(k))$ expected time algorithm and furthermore a \(\widetilde{O}(T(k n,k,\rho(H))g(k))\) time deterministic algorithm for \textsc{Subgraph Isomorphism}, for some computable function $g$.
    \item If there is a \(T(n,k,\rho(H),W)\) time deterministic algorithm for \textsc{Exact Weight Colored Subgraph Isomorphism}, then there is a \(O(T(k n, k, \rho(H),W)g(k))\) expected time algorithm and furthermore a \(\widetilde{O}(T(k n, k, \rho(H),W)g(k))\) time deterministic algorithm for \textsc{Exact Weight Subgraph Isomorphism}, for some computable function \(g\).
  \item Let $\tw(H) \geq 2$. If there is a $T(n,k,\tw(H))$ time algorithm for \textsc{Subgraph Isomorphism}, then there is a $O(T(\poly(k)n,\poly(k),\tw(H)) + \poly(k)n^2)$ time algorithm for \textsc{Colored Subgraph Isomorphism}.
  \item If there is a $T(n,k,\rho(H),W)$ time algorithm for \textsc{Exact Weight Subgraph Isomorphism}, then there is a $O(T(2n,2k,\rho(H),2^kW) + \poly(k)n^2)$ time algorithm for \textsc{Exact Weight Colored Subgraph Isomorphism}.
  \end{enumerate}
\end{lemma}

\begin{proof}[Proof (of Lemma~\ref{lemma-equiv})]
  We start by showing \textbf{part 1} and \textbf{part 2}. As mentioned, this follows directly from a standard application of the Color Coding technique~\cite{alon1995color}. Briefly speaking, they use random colorings of the vertices of $G$ to make the potential solution subgraph multicolored with some probability depending only on $k$. In our case, we may then try all $k!$ mappings from colors to vertices of $H$ to obtain the randomized algorithm; we delete any monochromatic edges to make sure that $f$ is a homomorphism. The authors of~\cite{alon1995color} also explain how to derandomize the algorithm using $k$-perfect hash functions, which results in the deterministic algorithm with an additional factor of $2^{O(k)}\log(n)$.

  The factor of $k$ in front of $n$ in $T(k n,k,\tw(H))$ comes from the fact that in the \textsc{Colored Subgraph Isomorphism} problem, we consider $n$ to be the size of the preimages of $f$, while in the \textsc{Subgraph Isomorphism} problem, we consider it to be the size of $V(G)$.
  
\medskip
  
Now we show \textbf{part 3}. Given an instance $\phi$ of \textsc{Colored Subgraph Isomorphism}, where $H$ has $k$ vertices, with preimages in $G$ of $n$ vertices each, we construct an equivalent \textsc{Subgraph Isomorphism} instance $\phi'$. This is done in two steps. First, we modify $\phi$ into an equivalent, but more structured \textsc{Colored Subgraph Isomorphism} instance $\widetilde{\phi}$, which we then reduce to $\phi'$. See Figure~\ref{sketch-equiv} for an example of this reduction.

\begin{description}
  \item[\underline{Subdividing all edges:}] In the first step, we construct $\widetilde{H}$, which consists of a subdivided copy of $H$, where each vertex has a unique ``signature'' structure attached to it. These signatures have a triangle as a key component. Abusing notation, we write $V(H) = \{1,\ldots, k\}$ and use the vertices as numbers. First, for each $i \in V(H)$, we add a vertex $\widetilde{i}$ to $\widetilde{H}$. Then, for each edge $ij \in E(H)$ with $i<j$, we add a vertex $\widetilde{x}_{ij}$ and create two edges $\widetilde{i}\widetilde{x_{ij}}$ and $\widetilde{x}_{ij}\widetilde{j}$, hence subdividing the edge $ij$. This ensures that for now, the new graph has no triangles. In $\widetilde{G}$, we populate the preimages as follows: For each $i \in V(H)$, let $f^{-1}(i) = \{a_i^1, \ldots, a_i^n\}$ and add $n$ vertices $\{\widetilde{a}_i^1, \ldots, \widetilde{a}_i^n\}$ to $f^{-1}(\widetilde{i})$. For each $\ell$, the vertex $\widetilde{a}_i^\ell$ corresponds to $a_i^\ell$. The preimages of $\widetilde{x}_{ij}$ are populated with $n$ vertices $\{b_{ij}^1, \ldots, b_{ij}^n\}$ via $\widetilde{f}$. For each edge $ij \in E(H)$ with $i<j$, we add an edge $a_i^\ell b_{ij}^\ell$ for every $\ell \in [n]$. We also go through each edge $a_i^\ell a_j^m \in E(G)$ and add a corresponding edge $b_{ij}^\ell \widetilde{a}_j^m$ to $E(\widetilde{G})$.

  \item[\underline{Signatures:}] We now add the signatures. For each $i \in [k]$, we add a new vertex $\widetilde{t}_i$ and a new triangle $\widetilde{u}_i\widetilde{v}_i\widetilde{w}_i$ to $\widetilde{H}$, and connect $\widetilde{t}_i$ to both $\widetilde{u}_i$ and $\widetilde{i}$ from $V(\widetilde{H})$. Furthermore, we connect $\widetilde{v}_i$ to $i+1$ other newly created vertices $\widetilde{y}_i^1, \ldots \widetilde{y}_i^{i+1}$. Let the set of all newly created vertices $\widetilde{t}_i, \widetilde{u}_i, \widetilde{v}_i, \widetilde{w}_i, \widetilde{y}_i^\ell$ ($i \in [n], \ell \in [i+1]$) be named $X$. In $\widetilde{G}$, we populate the preimages of these new vertices by adding $n$ vertices $\{z_1, \ldots, z_n\}$ to $V(\widetilde{G})$ for each vertex $v \in X$, with $\forall i \in [n]: \widetilde{f}(z_i) = v$. Now, for each $u \in X$, we pick an arbitrary node from $f^{-1}(u)$ and call it active. Furthermore, for all $\widetilde{i} \in V(\widetilde{H})$, we call all vertices of $f^{-1}(i)$ active. Now for each $v \in X$, we connect its active vertex in $f^{-1}(v)$ to all active vertices from the neighbourhood $f^{-1}(N(v))$. Note that of the vertices in $f^{-1}(X)$, only the active ones have edges at all. Indeed, for each $i \in [k]$, $\widetilde{G}$ contains exactly one triangle such that one of its vertices has degree $i+3$.
  \end{description}
  
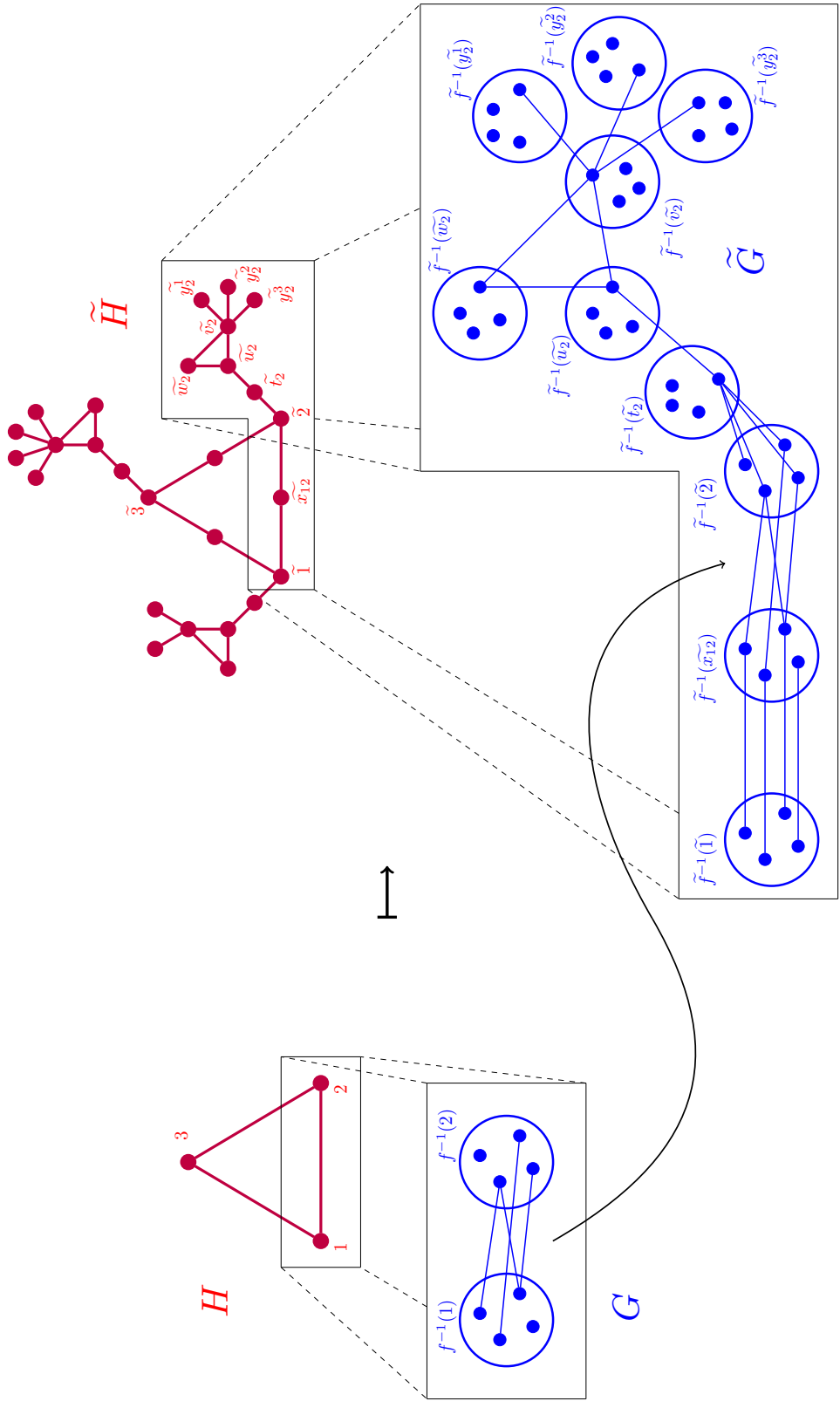
\begin{sidewaysfigure}
  \scalebox{0.8}{
    \tikzstyle{every picture}=[tikzfig]
    \begin{tikzpicture}
	\begin{pgfonlayer}{nodelayer}
		\node [style=H-node] (0) at (-13.25, 9.5) {};
		\node [style=H-node] (3) at (-7.25, 9.5) {};
		\node [style=H-node] (4) at (-10.25, 14.5) {};
		\node [style=none] (6) at (-14.25, 8) {};
		\node [style=none] (7) at (-14.25, 11) {};
		\node [style=none] (8) at (-6.25, 11) {};
		\node [style=none] (9) at (-6.25, 8) {};
		\node [style=none] (10) at (-19.25, 5.5) {};
		\node [style=none] (11) at (-19.25, -0.5) {};
		\node [style=none] (14) at (-7.25, 5.5) {};
		\node [style=none] (15) at (-7.25, -0.5) {};
		\node [style=Hollow H] (17) at (-10.25, 2.5) {};
		\node [style=Hollow H] (18) at (-16.25, 2.5) {};
		\node [style=G-node] (19) at (-16, 3.5) {};
		\node [style=G-node] (20) at (-17, 2.75) {};
		\node [style=G-node] (21) at (-15.25, 2) {};
		\node [style=G-node] (22) at (-16.5, 1.5) {};
		\node [style=G-node] (23) at (-10, 3.5) {};
		\node [style=G-node] (24) at (-11, 2.75) {};
		\node [style=G-node] (25) at (-9.25, 2) {};
		\node [style=G-node] (26) at (-10.5, 1.5) {};
		\node [style=H-node] (27) at (12, 11) {};
		\node [style=H-node] (28) at (15, 11) {};
		\node [style=H-node] (29) at (18, 11) {};
		\node [style=H-node] (30) at (15, 16) {};
		\node [style=H-node] (31) at (16.5, 13.5) {};
		\node [style=H-node] (32) at (13.5, 13.5) {};
		\node [style=H-node] (33) at (16, 17) {};
		\node [style=H-node] (34) at (17, 18) {};
		\node [style=H-node] (35) at (17, 19.5) {};
		\node [style=H-node] (36) at (18.5, 18) {};
		\node [style=H-node] (37) at (15.75, 20.25) {};
		\node [style=H-node] (38) at (16.5, 21) {};
		\node [style=H-node] (39) at (17.5, 21) {};
		\node [style=H-node] (40) at (16, 17) {};
		\node [style=H-node] (41) at (18.25, 20.25) {};
		\node [style=H-node] (42) at (19, 12) {};
		\node [style=H-node] (43) at (20, 13) {};
		\node [style=H-node] (44) at (20, 14.5) {};
		\node [style=H-node] (45) at (21.5, 13) {};
		\node [style=H-node] (46) at (22.5, 14) {};
		\node [style=H-node] (47) at (23, 13) {};
		\node [style=H-node] (48) at (22.5, 12) {};
		\node [style=H-node] (49) at (11, 12) {};
		\node [style=H-node] (50) at (10, 13) {};
		\node [style=H-node] (51) at (10, 14.5) {};
		\node [style=H-node] (52) at (8.5, 13) {};
		\node [style=H-node] (53) at (9.25, 15.75) {};
		\node [style=H-node] (54) at (10.75, 15.75) {};
		\node [style=none] (55) at (11.5, 12.25) {};
		\node [style=none] (56) at (11.5, 9.75) {};
		\node [style=none] (57) at (24, 9.75) {};
		\node [style=none] (58) at (24, 15.5) {};
		\node [style=none] (59) at (18, 15.5) {};
		\node [style=none] (60) at (18, 12.25) {};
		\node [style=Hollow H] (61) at (2, -7.5) {};
		\node [style=Hollow H] (62) at (9, -7.5) {};
		\node [style=Hollow H] (63) at (16, -7.5) {};
		\node [style=Hollow H] (64) at (19, -4.5) {};
		\node [style=Hollow H] (65) at (22, -1.5) {};
		\node [style=Hollow H] (66) at (22, 3.5) {};
		\node [style=Hollow H] (67) at (27, -1.5) {};
		\node [style=Hollow H] (68) at (29.5, 2) {};
		\node [style=Hollow H] (69) at (29.5, -5) {};
		\node [style=Hollow H] (70) at (31.5, -1.75) {};
		\node [style=G-node] (72) at (31.75, -0.75) {};
		\node [style=G-node] (73) at (31, -1.25) {};
		\node [style=G-node] (74) at (31.25, -2.5) {};
		\node [style=G-node] (75) at (32.25, -1.5) {};
		\node [style=G-node] (76) at (29.75, 3) {};
		\node [style=G-node] (77) at (28.5, 2) {};
		\node [style=G-node] (78) at (30.5, 2) {};
		\node [style=G-node] (79) at (28.75, 3) {};
		\node [style=G-node] (80) at (26.25, -1.75) {};
		\node [style=G-node] (81) at (27.5, -2) {};
		\node [style=G-node] (82) at (27.25, -0.75) {};
		\node [style=G-node] (83) at (26.75, -2.5) {};
		\node [style=G-node] (84) at (29, -6) {};
		\node [style=G-node] (85) at (28.75, -4.75) {};
		\node [style=G-node] (86) at (30, -4.75) {};
		\node [style=G-node] (87) at (30, -5.75) {};
		\node [style=G-node] (88) at (22, -0.75) {};
		\node [style=G-node] (89) at (21.25, -1.25) {};
		\node [style=G-node] (90) at (23, -1.5) {};
		\node [style=G-node] (91) at (21.5, -2.25) {};
		\node [style=G-node] (92) at (16.25, -6.5) {};
		\node [style=G-node] (93) at (15.25, -7.25) {};
		\node [style=G-node] (94) at (17, -8) {};
		\node [style=G-node] (95) at (15.75, -8.5) {};
		\node [style=G-node] (96) at (2.25, -6.5) {};
		\node [style=G-node] (97) at (1.25, -7.25) {};
		\node [style=G-node] (98) at (3, -8) {};
		\node [style=G-node] (99) at (1.75, -8.5) {};
		\node [style=G-node] (100) at (19.25, -3.75) {};
		\node [style=G-node] (101) at (18.5, -3.75) {};
		\node [style=G-node] (102) at (19.5, -5.5) {};
		\node [style=G-node] (103) at (18.25, -4.75) {};
		\node [style=G-node] (108) at (22, 4.25) {};
		\node [style=G-node] (109) at (21.25, 3.75) {};
		\node [style=G-node] (110) at (23, 3.5) {};
		\node [style=G-node] (111) at (21.75, 2.75) {};
		\node [style=G-node] (112) at (9.25, -6.5) {};
		\node [style=G-node] (113) at (8.25, -7.25) {};
		\node [style=G-node] (114) at (10, -8) {};
		\node [style=G-node] (115) at (8.75, -8.5) {};
		\node [style=none] (116) at (-0.25, -4) {};
		\node [style=none] (117) at (-0.25, -10) {};
		\node [style=none] (118) at (33.75, -10) {};
		\node [style=none] (119) at (33.75, 5.75) {};
		\node [style=none] (120) at (16, 5.75) {};
		\node [style=none] (121) at (16, -4) {};
		\node [style=none] (122) at (-15.75, 5.5) {};
		\node [style=none] (123) at (3, -4) {};
		\node [style=none] (124) at (18, 9.75) {};
		\node [style=none] (125) at (16, -10) {};
		\node [style=none] (126) at (17.6, 5.75) {};
		\node [style=none] (127) at (26, 5.75) {};
		\node [style=none] (128) at (-13.25, 0.75) {};
		\node [style=none] (129) at (12.5, -5.75) {};
		\node [style=none] (130) at (-1, -3) {};
		\node [style=big invisible] (131) at (-15.5, 13.5) {};
		\node [style=big invisible] (132) at (-15.5, 13.5) {\textcolor{red}{$H$}};
		\node [style=big invisible] (133) at (-15.75, -2) {\textcolor{blue}{$G$}};
		\node [style=big invisible] (134) at (22, 17.5) {\textcolor{red}{$\widetilde{H}$}};
		\node [style=big invisible] (135) at (24, -6.5) {\textcolor{blue}{$\widetilde{G}$}};
		\node [style=none] (136) at (-16.5, 4.75) {\textcolor{blue}{$f^{-1}(1)$}};
		\node [style=none] (137) at (-9.25, 4.75) {\textcolor{blue}{$f^{-1}(2)$}};
		\node [style=none] (138) at (-13.5, 8.75) {\textcolor{red}{$1$}};
		\node [style=none] (139) at (-7.5, 8.75) {\textcolor{red}{$2$}};
		\node [style=none] (140) at (-9.25, 14.75) {\textcolor{red}{$3$}};
		\node [style=none] (141) at (12.25, 10.25) {\textcolor{red}{$\widetilde{1}$}};
		\node [style=none] (142) at (15.25, 10.25) {\textcolor{red}{$\widetilde{x_{12}}$}};
		\node [style=none] (143) at (18.25, 10.25) {\textcolor{red}{$\widetilde{2}$}};
		\node [style=none] (144) at (14.5, 16.5) {\textcolor{red}{$\widetilde{3}$}};
		\node [style=none] (145) at (19.5, 11.25) {\textcolor{red}{$\widetilde{t_2}$}};
		\node [style=none] (146) at (20.5, 12.25) {\textcolor{red}{$\widetilde{u_2}$}};
		\node [style=none] (147) at (19.25, 14.75) {\textcolor{red}{$\widetilde{w_2}$}};
		\node [style=none] (148) at (21.5, 13.75) {\textcolor{red}{$\widetilde{v_2}$}};
		\node [style=none] (149) at (23, 14.75) {\textcolor{red}{$\widetilde{y_2^1}$}};
		\node [style=none] (150) at (23.5, 12.25) {\textcolor{red}{$\widetilde{y_2^2}$}};
		\node [style=none] (151) at (22.75, 11) {\textcolor{red}{$\widetilde{y_2^3}$}};
		\node [style=none] (152) at (1.25, -5) {\textcolor{blue}{$\widetilde{f}^{-1}(\widetilde{1})$}};
		\node [style=none] (153) at (8.5, -5) {\textcolor{blue}{$\widetilde{f}^{-1}(\widetilde{x_{12}})$}};
		\node [style=none] (154) at (14.75, -5) {\textcolor{blue}{$\widetilde{f}^{-1}(\widetilde{2})$}};
		\node [style=none] (155) at (17.75, -2.25) {\textcolor{blue}{$\widetilde{f}^{-1}(\widetilde{t_2})$}};
		\node [style=none] (156) at (20, 0.5) {\textcolor{blue}{$\widetilde{f}^{-1}(\widetilde{u_2})$}};
		\node [style=none] (157) at (24.75, 5) {\textcolor{blue}{$\widetilde{f}^{-1}(\widetilde{w_2})$}};
		\node [style=none] (158) at (31.25, 4.25) {\textcolor{blue}{$\widetilde{f}^{-1}(\widetilde{y_2^1})$}};
		\node [style=none] (160) at (31, -7.25) {\textcolor{blue}{$\widetilde{f}^{-1}(\widetilde{y_2^3})$}};
		\node [style=none] (161) at (32.5, 0.75) {\textcolor{blue}{$\widetilde{f}^{-1}(\widetilde{y_2^2})$}};
		\node [style=none] (162) at (25.25, -3.75) {\textcolor{blue}{$\widetilde{f}^{-1}(\widetilde{v_2})$}};
		\node [style=none] (163) at (-1, 7) {};
		\node [style=none] (164) at (1, 7) {};
	\end{pgfonlayer}
	\begin{pgfonlayer}{edgelayer}
		\draw [style=H-edge] (4) to (3);
		\draw [style=H-edge] (0) to (3);
		\draw [style=H-edge] (0) to (4);
		\draw (7.center) to (6.center);
		\draw (9.center) to (6.center);
		\draw (7.center) to (8.center);
		\draw (8.center) to (9.center);
		\draw (14.center) to (15.center);
		\draw (15.center) to (11.center);
		\draw (11.center) to (10.center);
		\draw (10.center) to (14.center);
		\draw [style=dashedline] (7.center) to (10.center);
		\draw [style=dashedline] (14.center) to (8.center);
		\draw [style=dashedline] (15.center) to (9.center);
		\draw [style=G-edge] (19) to (24);
		\draw [style=G-edge] (20) to (25);
		\draw [style=G-edge] (21) to (26);
		\draw [style=G-edge] (21) to (24);
		\draw [style=H-edge] (27) to (32);
		\draw [style=H-edge] (32) to (30);
		\draw [style=H-edge] (30) to (31);
		\draw [style=H-edge] (31) to (29);
		\draw [style=H-edge] (29) to (28);
		\draw [style=H-edge] (28) to (27);
		\draw [style=H-edge] (30) to (40);
		\draw [style=H-edge] (40) to (34);
		\draw [style=H-edge] (34) to (35);
		\draw [style=H-edge] (35) to (36);
		\draw [style=H-edge] (36) to (34);
		\draw [style=H-edge] (37) to (35);
		\draw [style=H-edge] (35) to (38);
		\draw [style=H-edge] (35) to (39);
		\draw [style=H-edge] (35) to (41);
		\draw [style=H-edge] (29) to (42);
		\draw [style=H-edge] (42) to (43);
		\draw [style=H-edge] (43) to (44);
		\draw [style=H-edge] (44) to (45);
		\draw [style=H-edge] (45) to (43);
		\draw [style=H-edge] (45) to (46);
		\draw [style=H-edge] (45) to (47);
		\draw [style=H-edge] (45) to (48);
		\draw [style=H-edge] (53) to (51);
		\draw [style=H-edge] (51) to (54);
		\draw [style=H-edge] (51) to (50);
		\draw [style=H-edge] (50) to (52);
		\draw [style=H-edge] (52) to (51);
		\draw [style=H-edge] (50) to (49);
		\draw [style=H-edge] (49) to (27);
		\draw (55.center) to (60.center);
		\draw (60.center) to (59.center);
		\draw (59.center) to (58.center);
		\draw (58.center) to (57.center);
		\draw (57.center) to (56.center);
		\draw (56.center) to (55.center);
		\draw [style=G-edge] (112) to (93);
		\draw [style=G-edge] (113) to (94);
		\draw [style=G-edge] (114) to (95);
		\draw [style=G-edge] (114) to (93);
		\draw [style=G-edge] (96) to (112);
		\draw [style=G-edge] (97) to (113);
		\draw [style=G-edge] (98) to (114);
		\draw [style=G-edge] (99) to (115);
		\draw [style=G-edge] (92) to (102);
		\draw [style=G-edge] (93) to (102);
		\draw [style=G-edge] (94) to (102);
		\draw [style=G-edge] (95) to (102);
		\draw [style=G-edge] (102) to (90);
		\draw [style=G-edge] (90) to (110);
		\draw [style=G-edge] (90) to (82);
		\draw [style=G-edge] (82) to (78);
		\draw [style=G-edge] (82) to (74);
		\draw [style=G-edge] (82) to (86);
		\draw [style=G-edge] (110) to (82);
		\draw (116.center) to (121.center);
		\draw (121.center) to (120.center);
		\draw (120.center) to (119.center);
		\draw (119.center) to (118.center);
		\draw (118.center) to (117.center);
		\draw (117.center) to (116.center);
		\draw [style=dashedline] (122.center) to (6.center);
		\draw [style=dashedline] (55.center) to (116.center);
		\draw [style=dashedline] (123.center) to (56.center);
		\draw [style=dashedline] (59.center) to (120.center);
		\draw [style=dashedline] (126.center) to (124.center);
		\draw [style=dashedline] (58.center) to (119.center);
		\draw [style=dashedline] (57.center) to (127.center);
		\draw [style=pointed, in=105, out=30, looseness=1.25] (130.center) to (129.center);
		\draw [style=thick edge, in=-150, out=-60, looseness=1.25] (128.center) to (130.center);
		\draw [style=implies] (163.center) to (164.center);
	\end{pgfonlayer}
\end{tikzpicture}
    }
  \caption{An example of the construction of the instance $\widetilde{\phi}$ from $\phi = (H,G,f)$ where $H$ is a triangle\label{sketch-equiv}}
\end{sidewaysfigure}

\noindent
We set $\widetilde{\phi}$ to be the instance $(\widetilde{H},\widetilde{G},\widetilde{f})$. This concludes the construction of $\widetilde{\phi}$.

  Clearly, the new instance $\widetilde{\phi}$ has a solution if and only if $\phi$ has one, and the size of the new instance is only a $\poly(k)$ factor larger than the size of $\phi$. We must also show that the reduction preserves treewidth. Note that $\widetilde{H}$ is obtained from $H$ via two operations: Subdividing edges and connecting a graph of smaller or equal treewidth via a single edge. It is easy to see that both operations do not change the treewidth.

Now, to construct $\phi'$, we simply get rid of the mapping $\widetilde{f}$. In other words, $\phi' = (\widetilde{H}, \widetilde{G})$. Obviously, if $\widetilde{\phi}$ has a solution, then $\phi'$ has one. For the other direction, suppose $\phi'$ has a solution, i.e.\ a subgraph $S'$ of $\widetilde{G}$ along with an isomorphism $h: V(\widetilde{H}) \to V(G[S'])$ of $\widetilde{G}$. Since both $\widetilde{H}$ and $\widetilde{G}$ have, for each $i$, exactly one triangle with a vertex of degree $i$, $h$ must map these triangles to their respective counterparts in $\widetilde{G}$. In particular, each of the vertices in $X$ is mapped to the active vertex in its preimage. Since the active node in $f^{-1}(\widetilde{t}_i)$ is connected only to nodes of $f^{-1}(\widetilde{i})$ (apart from the active node of $f^{-1}(\widetilde{u}_i)$, which is already in the image of $h(\widetilde{u}_i)$), we know that $h(\widetilde{i}) \in f^{-1}(\widetilde{i})$. Analogously, $h(\widetilde{x}_{ij}) \in f^{-1}(x_{ij})$. Hence, $S'$ takes exactly one vertex from each preimage of vertices from $\widetilde{H}$. Thus, $\widetilde{\phi}$ also has a solution.

We thus obtain a way to reduce $\phi$ to $\phi'$ with a size factor of only $\poly(k)$. The reduction obviously runs in $O(\poly(k)n^2)$ time. This shows part 3.

\medskip

  Finally, we show \textbf{part 4}. We begin with the node-weighted version. Given an instance $\phi$ of \textsc{Exact Weight Colored Subgraph Isomorphism} where the pattern graph $H$ has $k$ vertices, we create an instance $\phi'$ of \textsc{Exact Weight Subgraph Isomorphism} by simply dropping $f$ and modifying the weights. We have to ensure that a solution of $\phi'$ takes exactly one node from each preimage of $H$. To do this, we encode a checklist in the weights of the nodes. Again, let $V(H) = \{1,\ldots, k\}$. Let $u \in f^{-1}(i)$ for $i \in V(H)$, and consider its weight $w(u)$. We modify it by multiplying it with $2^k$ and adding $2^i$. We call the added weight its ``signature''. Now, since any solution must pick exactly $k$ vertices, the only way that the signatures of the solution vertices sum up to $2^k-1$ is to pick vertices which have a sum of weight 0 according to the original weight function and furthermore have exactly one vertex with added weight $2^i$ for each $i=1,\ldots, k$. To complete our reduction, we pick an arbitrary vertex $v \in V(G)$ and subtract $2^k-1$ from all vertices in $f^{-1}(v)$, making the new target zero. This reduction does not alter $G$ or $H$, and instead only modifies the weight function, resulting in the stated time bounds.

  For the edge-weighted version, we can use essentially the same construction as for the node-weighted version. Again, all weights are multiplied by $2^k$, and each vertex of $H$ has a unique ``signature''. However, this time, we have to add the signatures to the edge weights. Consequently, for each $i \in V(H) = \{1, \ldots, k\}$, we pick an arbitrary incident edge $e \in E(H)$. Let $e = \{i,j\}$. We add the signature $2^i$ to every edge in the preimage of $e$. That is, to every edge $\{e' \in E(G)\ |\ e' = \{u,v\} \text{ and } f(u) = i \text{ and } f(v) = j\}$. This way, we must still pick a node from each preimage to ensure that the signatures sum to $T := 2^k-1$. Again, we pick an arbitrary edge $e \in E(H)$ and subtract $T$ from all edges in its preimage. This almost completes the proof. However, we still have to handle nodes of degree 0 in $H$, since we cannot pick an incident edge for them. However, an isolated vertex $i$ in $H$ may be mapped to any vertex in $G$. Hence, we may simply skip the signature of $i$. We also have to modify the target $T$ to be $T-2^i$.
\end{proof}

\section{Open Problems}\label{sec:open-problems}

In this paper we discussed many different variants of the \textsc{Subgraph Isomorphism} problem. 
For some of these variants we leave gaps, which gives rise to several open problems:

\begin{enumerate}
\item Can the algorithms for weighted trees be improved? We have shown that some improvements can be made for node-weighted trees (see Theorem~\ref{corollary-node-weighted-algo-trees}), but are these optimal? What about edge-weighted trees?
\item Are there fast algorithms for unweighted \textsc{Subgraph Isomorphism} on graphs of bounded pathwidth that do not use rectangular matrix multiplication? Can the gap between exponent $\omega(p-1)$ and exponent $p$ be closed? Similar questions apply to the weighted case; see Theorems~\ref{corollary-upper-bound-detection-pathwidth} and~\ref{corollary-upper-bound-weighted-pathwidth}.
\item Relatedly, are there good lower bounds for unweighted \textsc{Subgraph Isomorphism} on graphs of bounded pathwidth? These could be attained via a modification of the proof of part 2 of Lemma~\ref{lemma-equiv} for pathwidth (though we see no way to do this), or with completely new techniques.
\end{enumerate}

We conclude with some more general open problems:

\begin{enumerate}
\item Do our algorithms and lower bounds also work for other types of graph homomorphisms, and for counting the number of solutions? It seems like techniques from~\cite{curticapean2017homomorphisms} should apply.
  \item In this work we demonstrated the existence of maximally hard patterns for which Subgraph Isomorphism requires time $n^{\tw(H)+1-o(1)}$. Can we classify which (classes of) patterns are maximally hard?
  \item Changing our focus from hard patterns to easy patterns, we can ask: do classes of patterns of unbounded treewidth exist for which Subgraph Isomorphism can be solved in time $n^{o(\tw(H))}$? Recall that a conditional lower bound rules out $n^{o(\tw(H)/\log \tw(H))}$~\cite{marx2007can}.
\end{enumerate}

\bibliography{references}

\newpage

\appendix

\section{Proof of the Conditional Lower Bound for Subset Sum}\label{apx:subset-sum}
In Section~\ref{sec:subset-sum}, we stated Theorem~\ref{thm-hyperclique-to-subsetsum}, which gives the following lower bound on \textsc{Subset Sum}:
\begin{center}
  \emph{
  For no $\varepsilon > 0$ can there be an algorithm which solves \textsc{Subset Sum} in time \\ $O(T^{1-\varepsilon}\poly(n))$ unless the \(h\)-uniform \textsc{Hyperclique} hypothesis fails for all \(h\geq 3\).}
\end{center}



\noindent
We now prove this.

\subsection{Reduction from \textsc{\(k\)-Sum} to \textsc{Subset Sum}}\label{apx:k-sum-to-subset-sum}

Before proving the result, we need a lemma that shows that we can reduce \(k\)-\textsc{sum} to \textsc{Subset Sum} with minimal overhead. This theorem is already known, but we could not find a formal proof of it in the literature. Therefore we provide one here.

\begin{lemma}[Reducing \(k\)-\textsc{Sum} to \textsc{Subset Sum}]\label{lemma-ksum-to-subsetsum}
  There is an algorithm $\mathcal{B}$ which, given as input a \textsc{$k$-sum} instance with $N$ values from $[0,D]$ per set, as well as a target $T$, constructs an equivalent \textsc{Subset Sum} instance with $k\cdot N$ values in $[0,D g(k)]$ and a target $T'$ bounded by $Tg(k)$ for a computable function $g$. Furthermore, for constant $k$, $\mathcal{B}$ runs in time linear in the input size.
\end{lemma} 

\begin{proof}
  The values of the \textsc{Subset Sum} instance are the union of the $k$ sets. However, we modify the weights and target as follows. At the front of the binary representation of the weights and the target, we add a buffer of $\lceil \log(k) \rceil$ zero bits to avoid overflow, then another $k$ bits constituting a ``checklist'', then in front of that another buffer of $\lceil \log(k) \rceil$ zero bits and finally another $\lceil\log(k)\rceil$ bits which contains a counter for the number of nodes of the solution.

  The target $T$ has the binary representation of $k$ in the counter bits and only ones in the checklist bits. Each weight has the binary representation of $1$ in the counter bits, ensuring that we take exactly $k$ weights. Furthermore, if the weight comes from the $i$-th set of the \textsc{$k$-sum} instances, its checklist bits is zero except for the $i$-th position.

  Now if one picks more than $2^{\lceil \log(k) \rceil}$ values, the counter at the front overflows the length of the target, so that selection cannot be a solution. Since $2^{\lceil \log(k) \rceil} < 2k$ and since there is a buffer of $\lceil \log(k) \rceil$ bits, the checklist cannot overflow into the counter. Hence any solution must pick exactly $k$ weights. Hence, the only way to achieve all ones in the checklist bits of a sum of $k$ weights is to pick exactly one weight from each of the $k$ sets. This completes the reduction.

  Since we add $3 \lceil \log(k) \rceil + k$ bits to the weights and the target, their value is multiplied by at most $2^k\cdot 2^3\cdot k^3$. Choosing $g(k) = 2^{k+3}k^3$, we obtain the bounds from the lemma.
\end{proof}

\subsection{Modifying the Weighted Lemma to Prove the Lower Bound}\label{apx:subset-sum-via-weighted-lemma}

To prove the theorem above, we modify the~\hyperref[lem:main-lemma]{Weighted Lemma} by adding a special case of parameters: All parameters are as before, but \(\beta = 1\) and \(r_2=0\). In this case, the preimages in the resulting instance \(\mathcal{I}'\) instead have size \(n^{k/r_1}\), and the maximum weight is \(W = \Theta(n^{(1+\varepsilon)k})\).

To prove this special case of the~\hyperref[lem:main-lemma]{Weighted Lemma}, we can almost use the reduction from its proof. We must simply purge all parts of the construction that relate to the edges part. That is, no construction of \(V_2\) in step 2\todo{this will be invalid after move to colored hyperclique}, no construction of \(S'_2\) in steps 3 and no construction of \(S_2\) in step 5. The reduction then yields an instance \(\mathcal{I}'\) as described above.

We now prove the theorem about \textsc{Subset Sum}. Essentially, the instances we get from the special case of the~\hyperref[lem:main-lemma]{Weighted Lemma} are \(k\)-\textsc{sum} instances that we can then reduce to \textsc{Subset Sum} via the lemma in Appendix~\ref{apx:k-sum-to-subset-sum} above.

\begin{proof}[Proof (of Theorem~\ref{thm-hyperclique-to-subsetsum})]
  Let \(h\) be given. Suppose there is an algorithm solving \textsc{Subset Sum} in time \(O(T^{1-\varepsilon}N^z)\) for some \(z\in\mathbb{N}\).
  We use the special case of the~\hyperref[lem:main-lemma]{Weighted Lemma} with
  \begin{itemize}
  \item some \(\varepsilon'\) chosen later,
  \item \(h' := h\), and
  \item some arbitrary \(r_1\in \mathbb{N}\) such that
    \begin{itemize}
    \item \(r_1 > \frac{\hat{c}}{h}\) (this ensures, again, that the
    running time \(O(n^{2h-1} + n^{\hat{c} k/(hr_1)})\) of the
    reduction is equal to \(O(n^{2h-1} + n^{k-\varepsilon})\) for some
    \(\varepsilon>0\), and can hence be ignored in the analysis), and
    \item \(r_1 > z\varepsilon'\) (this ensures that \(n^{zk/r_1} < n^{\varepsilon'k}\)).
  \end{itemize}
\end{itemize}

  This yields a \(k\in\mathbb{N}\) and a reduction algorithm \(\mathcal{A}\) with the properties from the description of the special case above. In particular, the reduction algorithm produces an \textsc{Exact Weight Colored Subgraph Isomorphism} where \(H\) consists only of isolated vertices with preimages of size \(O(n^{k/r_1})\) and maximum absolute weight \(W = \Theta(n^{(1+\varepsilon)k})\). We make all weights positive by adding a large number to each, such that the target is \(T = \Theta(n^{(1+\varepsilon)k})\). Note that this instance is also a \((r_1 + {hr_1 \choose h})\)-\textsc{sum} instance, with each set of numbers being the set of weights in a preimage.

  We now use the algorithm \(\mathcal{B}\) from Lemma~\ref{lemma-ksum-to-subsetsum} to convert this to a Subset Sum instance with \(N = O(n^{k/r_1})\) values in \([0,\Theta(n^{(1+\varepsilon')k})]\) and target \(T = \Theta(n^{(1+\varepsilon')k})\).
  
  We can now solve the instance in time \(O(n^{(1-\varepsilon)(1+\varepsilon')k}n^{zk/r_1}) = O(n^{((1-\varepsilon)(1+\varepsilon')+\varepsilon')k})\). Hence it suffices to choose \(\varepsilon'\) such that
  \begin{align*}
    (1-\varepsilon)(1+\varepsilon')+\varepsilon' < 1 \ \ \iff\ \ \varepsilon' < \frac{\varepsilon}{2-\varepsilon}
  \end{align*}
  Since \(\varepsilon \in (0,1)\) and hence \(\frac{\varepsilon}{2-\varepsilon} \in (0,1)\), this is always possible.

\end{proof}

\section{Hyperclique and Boolean \(k\)-Wise Matrix Product are Equivalent}\label{apx:hyperclique-equiv-boolean}

In this section, we show that \(k\)-uniform \((k+1)\)-\textsc{Hyperclique} has a fast algorithm if and only if the \textsc{\(k\)-Wise Matrix Product} problem has a fast algorithm. More formally:

\begin{theorem}
  The \(k\)-uniform \((k+1)\)-\textsc{Hyperclique} problem has an algorithm running in time \(O(n^{k+1-\varepsilon})\) (for some \(\varepsilon > 0\)) if and only if the \textsc{\(k\)-wise Matrix Product} problem has an algorithm running in time \(O(n^{k+1-\varepsilon'})\) (for some \(\varepsilon' > 0\)).
\end{theorem}

In the proofs of both directions, we will reduce from and to the \textsc{Colored Hyperclique} problem instead of the \textsc{Hyperclique} problem. This is because these problems are equivalent: In the one direction, we can use the reduction in~\hyperref[step1-unweighted-lemma]{step 1} in the proof of the~\hyperref[lem:unweighted-lemma]{Unweighted Lemma} (Section~\ref{sec:unweighted-uncol-subiso}). In the other direction, a trivial argument shows that simply throwing away the color homomorphism is enough.

Reducing from and to the \textsc{Colored Hyperclique} problem simplifies the proofs of both directions. Both are a straightforward generalization of the known proof that the standard boolean matrix product and the \textsc{triangle} problem\footnote{Given a 3-colored graph, check if it contains a triangle. Note that this is the colored \(k\)-uniform \((k+1)\)-\textsc{Hyperclique} problem for \(k=2\).} are subcubically equivalent\footnote{Meaning one problem has a \(O(n^{3-\varepsilon})\) algorithm if and only if the other has one.} with respect to combinatorial algorithms. This was proven in~\cite{williams2010subcubic} for the more general result of subcubic equivalence of the \((\min,+)\) matrix product\footnote{This is the standard matrix product, but addition is replaced by \(\min\) and multiplication is replaced by addition.} and the \textsc{Negative Triangle} problem\footnote{Given a 3-colored graph with edge weights, check if it contains a triangle of negative total weight.} and is a fundamental result in Fine-Grained Complexity Theory.

\subsection{From Hyperclique to Boolean \(k\)-Wise Matrix Product}\label{apx:hyperclique-to-boolean}

\begin{lemma}\label{hyperclique-to-boolean-mp}
  If the \textsc{\(k\)-wise Matrix Product} problem has an algorithm running in time \(O(n^{k+1-\varepsilon})\) (for some \(\varepsilon > 0\)), then the \(k\)-uniform \((k+1)\)-\textsc{Hyperclique} problem has an algorithm running in time \(O(n^{k+1-\varepsilon'})\) (for some \(\varepsilon' > 0\)).
\end{lemma}

\begin{proof}
  First, recall that in the proof of theorem~\ref{upper-bound-detection} we defined how to index tensors via configurations. We re-use this definition here.

  Given a hyperclique instance \((G,f)\) where \(G\) is a \(k\)-uniform hypergraph and \(f: V(G) \to V(C_{k+1})\) is a color homomorphism to the \(k\)-uniform \((k+1)\)-hyperclique \(C_{k+1}\), we solve it via a single application of the \(k\)-wise matrix product. To do this, we create \(k\) tensors \(A^1, \ldots, A^k\) of order \(k\) with dimensions \(n\times \ldots \times n\). We initialize them with zeroes as entries. The tensor \(A^i\) will be indexed by configurations of \(V(C_{k+1}) \setminus \{i\}\) with ordering \((1, \ldots, i-1, k+1, i+1, \ldots, k)\). For each \(i\) and each configurations \(R\) of \(V(C_{k+1}) \setminus \{i\}\), we set \(A^i[R]\) to one if and only if \(R(1)\ldots R(i-1)R(i+1)\ldots R(k+1) \in E(G)\). We use truth values and 0/1 interchangeably to declutter notation.

  Now let \(A(R) := \MP_k(A^1, \ldots, A^k)\). Then for a configuration \(R\) of \(V(C_{k+1}) \setminus \{k+1\}\), we have that
  \begin{align*}
    A[R] &= \bigvee_{\ell \in [n]} A^1[\ell, R(2), \ldots, R(k)]\land \ldots \land A^k[R(1), \ldots, R(k), \ell]\\
                      &= \exists v' \in f^{-1}(k+1): \forall i \in [k]: A^i[R(1), \ldots, R(i-1), v', R(i+1), \ldots, R(k)]\\
                      &= \exists v' \in f^{-1}(k+1): \forall i \in [k]: R(1)\ldots R(i-1)R(i+1)\ldots R(v)v' \in E(G) \tag{\dag}
  \end{align*}
  
But now note that there is a hyperclique in \(G\) if and only if there exists a configuration \(R\) of \(V(C_{k+1})\setminus \{k+1\}\) such that \((\dag)\) is true and \(R(1)\ldots R(k) \in E(G)\). Hence we simply calculate \(A[R]\) via the fast \(k\)-wise matrix product algorithm, then check if such a configuration exists by iterating through every possible configuration. This takes time \(O(n^{k+1-\varepsilon} + n^k) = O(n^{k+1-\varepsilon})\).
\end{proof}

\subsection{From Boolean \(k\)-Wise Matrix Product to Hyperclique}\label{apx:boolean-to-hyperclique}

\begin{proposition}\label{prop:hyperclique-to-finding-hyperclique}
  Let \(T(n) = O(n^c)\) for some \(c\geq1\). If there is a \(T(n)\) time algorithm for the \(k\)-uniform \((k+1)\)-\textsc{Hyperclique} problem running in time \(T(n)\), then there is an \(O(T(n))\) time algorithm solving the \(k\)-uniform \((k+1)\)-\textsc{Hyperclique} problem which also outputs the vertices of the \((k+1)\)-hyperclique (if one exists).
\end{proposition}
\begin{proof}
  The proof is a straightforward binary-search-like algorithm. Let \(C_{k+1}\) be the \(k\)-uniform \((k+1)\)-hyperclique and define \(V(C_{k+1}) = \{1, \ldots, k+1\}\). For each \(i \in [k+1]\), we split the preimage \(f^{-1}(i)\) into two equal parts of size \(n/2\). Now for each of the \(2^{k+1}\) possible \((k+1)\)-tuples of halves from different preimages, we check if there exists a \((k+1)\)-hyperclique between these halves in time \(T(n/2)\). If no tuple has a hyperclique, we return no. If at least one tuple has a hyperclique, we recurse on one of them, with the new preimages being the halves selected in the tuple. The base case is reached if there is a single vertex in each preimage, which we then return.

  The running time of this recursive algorithm is \(S(n) = 2^{k+1}T(n/2) + S(n/2) = 2^{k+1}(T(n/2) + T(n/4) + T(n/8) + \ldots)\). Since \(T(n) = O(n^c)\) for \(c\geq 1\), we have \(S(n) = O(T(n))\).
\end{proof}

Hence if we can check if a hyperclique exists, we can also find it in essentially the same time. We now use this to prove the second direction, which we already stated in lemma~\ref{lem:hyperclique-to-k-wise-mp}. Recall that the statement of lemma~\ref{lem:hyperclique-to-k-wise-mp} is the following:

\begin{center}
  \emph{
    If the \(k\)-uniform \((k+1)\)-\textsc{Hyperclique} problem has an algorithm running in time \\ \(O(n^{k+1-\varepsilon})\) (for some \(\varepsilon > 0\)), then the \textsc{\(k\)-wise Matrix Product} problem \\ has an algorithm running in time \(O(n^{k+1-\varepsilon'})\) (for some \(\varepsilon' > 0\)).}
\end{center}


\begin{proof}
  In this proof, we again use the definition of indexing tensors via configurations from the proof of theorem~\ref{upper-bound-detection} (also used above).

  
  Let the input tensors \(A^1, \ldots, A^k\), each of order \(k\) and dimensions \(n\times \ldots \times n\) be given. Now consider the \(k\)-uniform \((k+1)\)-hyperclique \(C_{k+1}\) and define \(V(C_{k+1}) = \{1, \ldots, k+1\}\). We let \(G\) be the graph with \(n\) vertices per preimage of \(C_{k+1}\), with no hyperedges (yet). In the following, we construct the edges of \(G\) and the color homomorphism \(f: V(G) \to V(C_{k+1})\).

  For each \(i \in [k+1]\), we let the tensor \(A^i\) be indexed by configurations of \(V(C_{k+1})\setminus \{i\}\) with ordering \((1, \ldots, i-1, k+1, i+1, \ldots, k)\). Now for each configuration \(R\) of \(V(C_{k+1})\setminus \{i\}\), we add the hyperedge \(R(1)\ldots R(i-1)R(i+1)\ldots R(k+1)\) to \(G\) if and only if \(A^i[R] = 1\). Hence we are encoding the \(i\)-th input tensor in the hyperedges between the preimages of \(V(C_{k+1}) \setminus \{i\}\).

  Finally, in \(V(C_{k+1})\setminus \{k+1\}\), we encode the all-ones tensor. That is, for each configuration \(R\) of \(V(C_{k+1})\setminus \{k+1\}\), we add the hyperedge \(R(1)\ldots R(k)\) to \(G\).

  We now show how to take advantage of the fast algorithm for \(k\)-uniform \((k+1)\)-\textsc{Hyperclique} to calculate the \(k\)-wise matrix product by using it repeatedly on subinstances of this new hyperclique instance. Let \(A_{\text{res}}\) be the output tensor, initialized to all-zeroes. Now take \(g \in \mathbb{N}\) (chosen later). For each \(i \in [k+1]\), we split the preimage \(f^{-1}(i)\) into \(g\) parts (i.e.\ subsets of vertices) \((f^{-1}(i))_1, \ldots, (f^{-1}(i))_g\), each of size \(n/g\) up to rounding.

  For any \((k+1)\)-tuple \((i(1), \ldots, i(k+1)) \in [g]^{k+1}\), we do the following: \((\star)\) While there is a hyperclique between the parts \((f^{-1}(1))_{i(1)}, \ldots, (f^{-1}(k+1))_{i(k+1)}\) (which can be checked via the fast algorithm in time \(O((n/g)^{k+1-\varepsilon})\)), we find this hyperclique in time \(O((n/g)^{k+1-\varepsilon})\) via the algorithm from Proposition~\ref{prop:hyperclique-to-finding-hyperclique}. Let the vertices of the found hyperclique be \(j(1), \ldots, j(k+1)\). Then we set \(A_{\text{res}}[j(1), \ldots, j(k)]\) to one and delete the edge \(j(1)\ldots j(k)\) from \(G\).

  After having gone through all \((k+1)\)-tuples, we output \(A_{\text{res}}\).

  We argue correctness. We use truth values and 0/1 interchangeably. Let the correct output be \(A_{\text{cor}} := \MP_k(A^1, \ldots, A^k)\). We want to prove \(A_{\text{cor}} = A_{\text{res}}\). Let \(A_{\text{cor}}\) and \(A_{\text{res}}\) be indexed by configurations of \(V(C_{k+1})\setminus \{k+1\}\) with ordering \((1, \ldots, k)\). Then for a configuration \(R\) of \(V(C_{k+1})\setminus \{k+1\}\), we have that \(A_{\text{cor}}[R]\) also satisfies the equation \((\dag)\) from the proof of lemma~\ref{hyperclique-to-boolean-mp}.

  Now suppose \(A_{\text{res}}[R] = 1\). Then certainly \((\dag)\) is true, so \(A_{\text{cor}}[R] = 1\). Conversely, suppose that \(A_{\text{cor}}[R] = 1\) and hence that \((\dag)\) is true. Then when we will find a hyperclique for the tuple \((i(1), \ldots, i(k))\) such that \(\forall j: R(j) \in (f^{-1}(i))_{i(j)}\). When we find it, the hyperedge \(R(1)\ldots R(k)\) either still exists (in which case we set \(A_{\text{cor}}[R]\) to one) or it does not, in which case it must have been deleted in a previous iteration where we must have then already set \(A_{\text{cor}}[R]\) to one.

  Let us analyze the running time. For each successful check in the while loop \((\star)\), we delete an edge in \(E(G)\). Hence there are at most \(O(n^{k})\) successful checks. In each successful check, we take time \(O((n/g)^{k+1-\varepsilon})\) for the check itself and another \(O((n/g)^{k+1-\varepsilon})\) for the execution of the while loop content. As for unsuccessfull while loop checks, we have at most one per \((k+1)\) tuple of parts, so \(g^{k+1}\). Each of these also takes \(O((n/g)^{k+1-\varepsilon})\).

  Hence the overall running time is \(O(n^k(n/g)^{k+1-\varepsilon} + g^{k+1}(n/g)^{k+1-\varepsilon})\). Choosing \(g^{k+1} = n^k\) and hence \(g = n^{k/(k+1)}\), we have a running time of \(O(n^{k}n^{(1-k/(k+1))(k+1-\varepsilon)})\), which works out to \(O(n^{k+1-\varepsilon/(k+1)})\).
\end{proof}

\end{document}